\documentclass[12pt]{article}%
\usepackage[T1]{fontenc}
\usepackage[utf8]{inputenc}%
\usepackage[tracking=true, letterspace=50, expansion=false]{microtype}
\usepackage{tikz} %
\usepackage{amsmath}%
\usepackage{amsthm}%
\usepackage{amssymb}%

\usepackage[margin=1in]{geometry}%
\usepackage{setspace}%
\usepackage{pdflscape}%

\usepackage[title]{appendix}

\newtheorem{lemma}{Lemma}
\newtheorem{corollary}{Corollary}

\newtheorem{prop}{Proposition}
\theoremstyle{definition}
\newtheorem{assumption}{Assumption}
\newtheorem{remark}{Remark}

\usepackage{mathtools}

\DeclareMathOperator*{\argmax}{argmax}

\DeclarePairedDelimiter\abs{\lvert}{\rvert}
\DeclarePairedDelimiter\indicatorfence{\{}{\}}
\newcommand\1{\operatorname{I}\indicatorfence}

\usepackage{xcolor}%
\definecolor{webbrown}{rgb}{.6,0,0}

\newcommand\indep{\protect\mathpalette{\protect\independenT}{\perp}}
\def\independenT#1#2{\mathrel{\rlap{$#1#2$}\mkern2mu{#1#2}}}

\newcommand{\overbar}[1]{\mkern 1.5mu\overline{\mkern-1.5mu#1\mkern-1.5mu}%
  \mkern 1.5mu}%
\newcommand{\OD}{\mathsf{OD}}
\newcommand{\DP}{\mathsf{DP}}

\usepackage{booktabs}
\usepackage[normal, online, flushleft]{threeparttable}%
\usepackage[nolist]{acronym}
\begin{acronym}
  \acro{CI}{confidence interval}
  \acro{OLS}{ordinary least squares}
  \acro{TSLS}{two-stage least squares}
  \acro{LATE}{local average treatment effect}
  \acro{ATE}{average treatment effect}
  \acro{IV}{instrumental variable}
  \acro{MTE}{marginal treatment effects}
\end{acronym}
\usepackage{etoolbox} %
\AtBeginEnvironment{appendices}{%
  \crefalias{section}{appendix}
  \crefalias{subsection}{appendix}
  \counterwithin{figure}{section}
  \counterwithin{table}{section}
}

\usepackage[
 bibencoding=utf8,%
 maxnames=2, %
 minnames=1, %
 maxbibnames=10, %
 minbibnames=5, %
 backend=biber,%
 sortlocale=en_US,%
 style=apa,%
 useprefix=true,%
 doi=true,%
 url=true,
 eprint=true%
]{biblatex}
\DeclareSortingNamekeyTemplate{
  \keypart{\namepart{family}}
  \keypart{\namepart{prefix}}
  \keypart{\namepart{given}}
  \keypart{\namepart{suffix}}
}

\usepackage{hyperref}%
\hypersetup{%
  breaklinks = true,%
  colorlinks = true,%
  anchorcolor = webbrown,%
  citecolor = webbrown,%
  filecolor = webbrown,%
  linkcolor = webbrown,%
  menucolor = webbrown,%
  urlcolor= webbrown,%
  citebordercolor= 1 0 0,%
  menubordercolor=1 0 0,%
  urlbordercolor=1 0 0,%
  runbordercolor=1 0 0,%
  pdftitle=Evaluating Counterfactual Policies Using Instruments,%
  pdfauthor=Michal Kolesár \& José Luis Montiel Olea \& Jonathan Roth}
\usepackage{cleveref}
\crefname{assumption}{assumption}{assumptions}
\Crefname{assumption}{Assumption}{Assumptions}
\Crefname{prop}{Proposition}{Propositions}
\newcommand\nycNw[0]{284,598}
\newcommand\nycNb[0]{310,588}

\newcommand\nycKmainb[0]{167}
\newcommand\nycKmainw[0]{160}
\newcommand\nycEYb[0]{22.9}
\newcommand\nycEYw[0]{20.6}
\newcommand\nycEYbD[0]{32.9}
\newcommand\nycEYwD[0]{26.9}
\newcommand\nycPCTElb[0]{30.8}%
\newcommand\nycPCTElw[0]{31.9}%
\newcommand\nycEYlb[0]{44.9}
\newcommand\nycEYlw[0]{35.7}
\newcommand\nycEYlbCI[0]{26.3}

\newcommand\nycEDb[0]{69.7}
\newcommand\nycEDw[0]{76.5}
\newcommand\nycQuotab[0]{80.2}
\newcommand\nycQuotaw[0]{83.2}
\newcommand\nycDPb[0]{0.025}
\newcommand\nycDPw[0]{0.021}
\newcommand\nycIVb[0]{49.7}
\newcommand\nycIVw[0]{48.1}
\newcommand\nycIVbse[0]{1.2}
\newcommand\nycIVwse[0]{2.0}

\newcommand\nycCIbincrease[0]{6.3}
\newcommand\nycDeltal[0]{4.0}
\newcommand\nycDeltau[0]{8.2}
\newcommand\nycDeltalp[0]{4.3}
\newcommand\nycDeltaup[0]{5.4}
\newcommand\nycDeltalCI[0]{1.0}
\newcommand\nycDeltauCI[0]{9.6}

\newcommand\sufN[0]{67,060}

\newcommand\sufKp[0]{69}
\newcommand\sufEY[0]{34.1}
\newcommand\sufED[0]{20.4}

\newcommand\sufQuotac[0]{29.4}

\newcommand\sufODc[0]{0.019}
\newcommand\sufDP[0]{0.032}
\newcommand\sufDPIV[0]{0.112}
\newcommand\sufrho[0]{0.998}
\newcommand\sufrhoIV[0]{0.972}

\newcommand\sufIVc[0]{-28.8}

\newcommand\sufIVcse[0]{10.0}

\newcommand\sufDeltaDc[0]{9.0}

\usepackage{floatrow}

\title{Evaluating Counterfactual Policies Using Instruments\thanks{All errors
    are our own. We thank numerous seminar participants, Amanda Agan, Isaiah Andrews, Josh
    Angrist, Yuehao Bai, Brigham Frandsen, Peter Hull, Guido Imbens, Toru Kitagawa, Pat Kline, Charles Manski, Thomas Richardson, Jesse Shapiro, Jesper
    Sørensen, Alex Torgovitsky, and Ed Vytlacil for helpful comments. We thank
    Peter Hull for providing us with aggregate statistics on New York City bail
    judges, and Henrik Sigstad for providing data on panels of judges in São Paulo. Kolesár acknowledges support by the
    National Science Foundation Grant SES-22049356. Montiel Olea acknowledges support by the National Science Foundation Grant SES-2315600. Roth acknowledges funding from the NIH under grant NIGMS 1R35GM155224 and from the Alfred P. Sloan Foundation.}}

\author{Michal Kolesár \and José Luis Montiel Olea \and Jonathan Roth}%
\date{\today}

\begin{document}

\maketitle

\begin{abstract}
We study settings in which a researcher has an instrumental variable (IV) and
seeks to evaluate the effects of a counterfactual policy that alters treatment
assignment, such as a directive encouraging randomly assigned judges to release more
defendants. We develop a general and computationally tractable framework for
computing sharp bounds on the effects of such policies. Our approach does not
require the often tenuous IV monotonicity assumption. Moreover, for an important
class of policy exercises, we show that IV monotonicity---while crucial for a
causal interpretation of two-stage least squares---does not tighten the bounds on
the counterfactual policy impact. We analyze the identifying power of
alternative restrictions, including the policy invariance assumption used in the
marginal treatment effect literature, and develop a relaxation of this
assumption. We illustrate our framework using applications to quasi-random
assignment of bail judges in New York City and prosecutors in Massachusetts.
\end{abstract}

\clearpage

\section{Introduction}

The most common approach among practitioners for leveraging an \ac{IV} to tease out causal effects
from observational data is to use \ac{TSLS} and interpret the resulting estimand
as a \ac{LATE}---a (weighted) average of treatment effects for individuals whose
treatment status changes in response to the \ac{IV} \parencite{ImAn94}. This
interpretation relies on the well-known \ac{IV} monotonicity condition.

Yet in many applications of \ac{IV} designs, both statistical evidence and
institutional details point to failure of IV monotonicity. For example, consider
the popular leniency \ac{IV} design, which harnesses as-good-as-random
assignment of judges or other decision-makers who differ in their leniency---the
propensity to grant treatment---to study the effects of treatments such as
incarceration, pretrial detention, or bankruptcy on various
outcomes.\footnote{The leniency \ac{IV} design was pioneered by
  \textcite{kling06}. It has been used to leverage (quasi-)random assignment of a variety of decision-makers, including judges
  \parencite[e.g.][]{AiDo15,DoSo15,dgy18}, patent examiners
  \parencite[e.g.][]{SaWi19}, child welfare investigators
  \parencite[e.g.][]{doyle2007child}, and doctors \parencite[e.g.][]{cgy22}. See Table 1 in \textcite{fll23} for additional references.} IV monotonicity requires that every defendant
released by a given judge would also be released by all judges who are more
lenient on average. Effectively, the judges need to agree on the ranking of the
defendants in terms of their risk; they only disagree on the release cutoff.
This is a stringent requirement, as pointed out in the original
\textcite{ImAn94} analysis (Example 2, p.~472). Institutional knowledge suggests
that decision-makers may have different rankings owing to heterogeneity in
skills \parencite{cgy22} or preferences \parencite{mueller-smith15}. Moreover,
IV monotonicity is frequently rejected by statistical tests
\parencite[e.g.][]{fll23}, as well as by direct evidence from judicial panels
\parencite{sigstad_monotonicity_2023}. To address these concerns, a growing
literature has proposed weaker or alternative conditions that still allow for a
\ac{LATE}-type interpretation of the \ac{TSLS} estimand
\parencite[e.g.][]{dechaisemartin17,strlb17,fll23, mtw21}. Yet, the plausibility of
these alternative assumptions is debated \parencite{MoTo24}.

Often, however, researchers may not ultimately be interested in the LATE \emph{per se}, but rather in evaluating counterfactual policies.\footnote{A recent survey of the empirical
  literature by \textcite{lss25} concludes that researchers are rarely
  interested in the LATE \emph{per se}.}
Researchers conducting a leniency IV, for instance, may be interested in
policies that nudge the decision-makers to be more (or less) lenient, by, say,
imposing release quotas \parencite{adh22}, providing algorithmic recommendations
\parencite{ady23}, or imposing a presumption of non-prosecution for low-level
offenses \parencite{adh23}. In the limit, universal release programs may treat
everyone with certain observable characteristics \parencite{albright22}.

In this paper, we develop a general framework for directly evaluating
counterfactual policies that alter treatment assignment. We make the goal of
policy evaluation explicit by indexing potential treatments as $D(z, a)$, where
$z$ is the value of the instrument (say, the identity of the judge in the leniency IV
example), and $a$ is an indicator for whether the policy (say, a release quota) is implemented. We only observe data under the status quo ($a=0$), with
the observed treatment given by $D(Z, 0)$, and the observed outcome given by $Y(D(Z,0))$, where $Y(d)$ is the potential outcome under treatment $d$. The parameter of interest is $\theta = E[Y(D(Z,1))]$, the average outcome (say, recidivism) under the
counterfactual treatment assignment, $D(Z,1)$. Since $\theta$ is generally not
point-identified, our goal is to derive bounds for it, i.e., an identified set
of values consistent with the observable data.

Within this framework, we develop three sets of results. First, we derive
tractable bounds for $\theta$ without imposing IV monotonicity. Second, we show
that in a range of relevant cases, imposing IV monotonicity does not help
tighten the identified set. Thus, while some form of monotonicity is essential
for a \ac{LATE} interpretation, our results show that IV monotonicity is neither
necessary nor sufficient to learn about counterfactuals. Third, we consider
other assumptions that do help tighten the bounds, and illustrate their power in
two applications.

Specifically, our first main result provides a computationally tractable characterization of the identified set for $\theta$ without imposing IV monotonicity. In related \ac{IV} settings \parencite[e.g.][]{kitagawa21, bai_identifying_2024}, the identified set is commonly computed by
enumerating all possible ``response types'', defined by the values of
$D(\cdot, \cdot)$. This proves computationally infeasible in our setting, because
the number of response types grows exponentially in the number of judges $K$ if
one does not impose \ac{IV} monotonicity. We show that one can bypass response
type enumeration and compute the identified set as the solution to a linear
program that scales linearly in $K$. The key observation is that it suffices to
consider sets of marginals for $(Y(\cdot), D(z, \cdot))$, involving the decision
of one judge $z$ at a time. This is because the observable data does not depend
on the joint distribution of $D(z, a)$ across judges $z$, and without
\ac{IV} monotonicity, this joint distribution is \emph{ex ante} unrestricted. Our result builds on insights in \textcite{RiRo14}, who derived bounds on the average treatment effect in IV settings with binary outcomes involving $O(K^2)$ restrictions.

Our second set of results gives sufficient conditions under which the
identified set for $\theta$ does not depend on whether one imposes IV
monotonicity. Specifically, we show that this is the case when either (a) one of
the potential outcomes is known, or (b) the policy counterfactual is a
``sufficiently strong'' encouragement (where the notion of ``sufficiently
strong'' is formalized in
\Cref{prop:monotonicity_is_not_helpful_strong_encouragement} below). Condition
(a) is satisfied in many criminal justice settings, where, for example, one
cannot commit pretrial misconduct unless one is released. Condition (b) is
trivially satisfied for universal release programs, like that studied in
\textcite{albright22}, and may also plausibly be satisfied by other
counterfactual policies that strongly encourage judges to release more
defendants. An important practical takeaway is that when these sufficient
conditions are satisfied, the debate over IV monotonicity is somewhat of a red
herring: instead of worrying about IV monotonicity, researchers interested in
policy counterfactuals should turn attention to evaluating the validity of other
assumptions that may in fact tighten the identified set.

Our third set of results considers several such assumptions. We begin by
evaluating the policy invariance assumption of \textcite{HeVy05}. While
\ac{IV} monotonicity effectively imposes that judges agree on their rankings of
defendants \emph{under the status quo}; policy invariance additionally imposes
that this ranking is the same \emph{under the counterfactual}. We show that
policy invariance can indeed be helpful in tightening the identified set in some
settings where IV monotonicity alone is not. Of course, since policy invariance is
stronger than IV monotonicity, it may often be implausible in practice. We
therefore develop a relaxation of policy invariance that may be more plausible yet still informative.
In particular, while policy invariance imposes perfect agreement among judges in
the ranking of defendants, our relaxation only imposes that judges do not
disagree too often. We show how bounds on disagreement rates can be calibrated
using settings where multiple decision-makers rule on the same cases, such as \textcite{sigstad_monotonicity_2023}'s data on panels of judges. We
further show that these disagreement bounds can be tractably incorporated into
our linear program for calculating the identified set. We also discuss a variety
of other economically-motivated restrictions that can be incorporated to further
tighten the identified set. For example, in the pretrial release setting, judges
are legally instructed to release defendants unless they are at high risk of
committing pretrial misconduct. A natural assumption is then that the defendants
released by judges under the status quo are lower-risk than the defendants they
would marginally release in response to a policy encouragement.

Two applications illustrate the usefulness of our results. The first
studies bail judges in New York City, who decide whether to release
defendants awaiting trial, using aggregated data from \textcite{adh22}. We evaluate two
counterfactual policies in this context. First, we consider a policy that
releases all defendants ($D(Z,1)=1$), allowing us to evaluate what would happen
if the universal release policy in Kentucky studied by \textcite{albright22}
were applied in this context. This policy is also relevant for evaluating the
``disparate impact'' of pretrial release decisions by race, as in \textcite{adh22}, whose disparate impact
parameter depends only on the (race-specific)
average outcomes under universal release and point-identified summary statistics. We obtain
informative bounds for the universal release counterfactual, imposing only
\ac{IV} validity (i.e., random IV assignment and exclusion). Our bounds suggest,
for example, that the defendants marginally released by this policy will have
higher misconduct rates than those released under the status quo, which is
intuitive given the judges' instructions to only detain high-risk defendants.
Our bounds for the disparate impact parameter are only moderately wider than the
range of estimates obtained by different parametric assumptions in
\textcite{adh22}, and our confidence interval excludes zero, indicating statistically significant disparate
impact by race.

We also use this data to evaluate a quota policy that requires the bottom 90\%
of judges in terms of leniency to increase their release rates to match the top
10\%. For this policy, only imposing IV validity yields trivial bounds that do
not restrict the misconduct rates of the marginally-released defendants. This
occurs because without any restrictions on agreement rates between judges, the
judges below the 90th percentile of leniency may choose to marginally release
only IV ``never-takers'' who were not released by any judge under the status
quo, and the data thus contain no information about their misconduct potential.
To tighten the bounds, we rule out such perfect discordance between judges above
and below the 90th percentile by calibrating a bound on the average disagreement
rate between these two groups of judges using
\textcite{sigstad_monotonicity_2023}'s data on panels of judges who rule on the
same cases. This calibration yields informative bounds on the counterfactual.

The second application uses data from \textcite[][ADH23]{adh23}, who study the
effect of non-prosecution of misdemeanor cases by assistant district attorneys
(ADAs) in Massachusetts on subsequent criminal justice involvement of the
defendants. ADH23 report TSLS estimates. Since the \textcite{fll23} test rejects
IV monotonicity for several of the courts in their data, ADH23 appeal to
\textcite{fll23}'s average monotonicity condition for a weighted-LATE
interpretation of TSLS\@. ADH23 also write that ``if all arraigning ADAs acted
like the most lenient ADAs in our sample\ldots [we] would likely see a reduction
in criminal justice involvement''. Based on this discussion, we evaluate two
counterfactual policies to make ADAs act more similarly to the most lenient ones
(defined as the top 10\%). First, we consider a simple re-allocation of cases to
the most lenient ADAs. The impact of such a policy is point-identified, and our
results suggest it would decrease crime, in line with the discussion in ADH23.
Second, we consider a quota that requires the bottom 90\% of ADAs to match the
non-prosecution rate of the top 10\%. For this policy, stringent bounds on
agreement rates between ADAs are needed to conclude that the policy would
decrease crime; the bounds appear unreasonably high given rejection of IV
monotonicity under the status quo. We can conclude the policy would decrease crime under more plausible restrictions on disagreement, however, if we also restrict heterogeneity of treatment effects among
people on whom the ADAs disagree. Thus, the conclusion that increasing
non-prosecution likely leads to a reduction in crime, as argued in ADH23, can be
obtained if one is comfortable imposing either a specific policy implementation
(such as reallocation) or restrictions on the degree of treatment effect
heterogeneity.

Our paper is motivated in part by the observation in previous work that the
connection between \ac{LATE} and economic policies of interest is not entirely
clear \parencite[e.g.][]{HeVy05,HeVy07i,manski1996jhr,manskiJPEMicro}.
\Textcite{HeVy99,HeVy05} introduced the \ac{MTE} framework for analyzing the
impacts of policies that change treatment assignment. However, the canonical
\ac{MTE} framework requires policy invariance. Our framework allows for
evaluation of a larger class of counterfactual policies by using the
generalized potential treatment function $D(z, a)$ and allowing for a wide range
of restrictions on it, including relaxations of policy invariance.\footnote{In a
  complementary line of work, \textcite{ura_policy_2025} consider a multi-index
  generalization of the MTE framework that allows for relaxations of IV
  monotonicity. Although the framework allows for partial identification of a
  variety of treatment parameters (such as ATE or ATT), analyzing impacts of
  general counterfactual policies appears challenging without further
  assumptions. \Textcite{sigstad_marginal_2024} studies conditions under which
  standard MTE methods correctly estimate parameters such as ATE or LATE even
  when \ac{IV} monotonicity is violated.} We build on \textcite{IcTa00}, which
considers policy evaluation under the assumption that $D(z,0) =D(z',1)$ for a known pair $(z,z')$, so that the treatment assignment for $z'$ under the counterfactual matches that of $z$ under the status quo. Such agreement would arise under policy invariance if $z$'s release rate under the status quo matches that of $z'$ under the counterfactual, $E[D(z,0)]=E[D(z',1)]$. Our framework nests this perfect agreement assumption
as a special case, but also allows for much more general restrictions on the
potential treatments $D(z,a)$, including non-zero disagreement bounds.

Our paper also relates to a large
literature that has considered bounds on parameters other than LATE in \ac{IV}
settings \parencite[e.g.][]{BaPe97,manski_nonparametric_1990,
  swanson_partial_2018, bai_inference_2025, chen_bounds_2015, RiRo14}. These
papers have primarily focused on parameters such as the \ac{ATE}, or average
effects for principal strata, in contrast to our focus on specific
counterfactual policies.

Finally, several papers have provided conditions under which \ac{IV}
monotonicity does not help to tighten the identified set for the \ac{ATE} or the
marginal distributions of potential outcomes
\parencite{BaPe97,kitagawa21,HeVy01bounds,bai_identifying_2024}. By contrast,
\textcite{kamat_identifying_2019} shows that IV monotonicity can matter for
partially identifying the copula of the potential outcomes. We complement this literature by providing conditions under which IV monotonicity does not help to tighten the bounds on the impact of a counterfactual policy, which need not correspond to the \ac{ATE}.

The next section sets up the model and defines the identified set for $\theta$,
the average outcome under the counterfactual. \Cref{sec:no_monotonicity} derives
a tractable characterization of the identified set without imposing IV
monotonicity. \Cref{sec:iv_monotonicity} gives sufficient conditions under
which imposing IV monotonicity does not help tighten the identified set, while
\Cref{sec:other_restrictions_help} gives examples of restrictions that do. The
identifying power of these restrictions is illustrated with two applications in
\Cref{sec:empirical-apps}. The appendices collect proofs and additional results.

\section{Setup}\label{subsection:idset}

\subsection{Observable data and potential outcomes}

We observe data on the triple $(Y, D, Z)$, where $Y$ is an outcome of interest,
$D\in \{0,1\}$ is a binary treatment indicator, and
$Z\in\mathcal{Z} = \{1, \dotsc, K\}$ is a multivalued instrument that influences
the treatment. We wish to use the data to learn about the effect of a
\emph{counterfactual policy}, which we denote by $a=1$. We let $a=0$ denote the
status quo observed in the data. To properly define the effect of this policy,
we use the potential outcome notation, with $Y(d)$ denoting the potential
outcome under treatment status $d\in\{0,1\}$, and $D(z, a)$ denoting the
potential treatment under instrument value $z\in\mathcal{Z}$ and policy
$a\in\{0,1\}$.\footnote{Similar notation appears in \textcite{IcTa00}.} We assume that the support of the joint distribution of $(Y(0), Y(1))$ is
contained in a compact set $\mathcal{Y}^{2} \subseteq \mathbb{R}^{2}$.\footnote{For simplicity, we focus on the case where $Y$ is bounded, which ensures that the bounds on the average policy effect are finite. Our results readily extend to the case with unbounded outcomes if one imposes additional restrictions on the outcome, such as those considered in \Cref{sec:outcome_restrictions}, that guarantee that the average policy effect is finite. Even if the bounds are infinite, one could modify our framework to obtain bounds on quantiles, rather than the average policy effect.} The
observed treatment and outcome are thus given by $D=D(Z,0)$ and
$Y=Y(D)=Y(D(Z,0))$, while the potential treatments $D(z,1)$ are not observed.
The policy effect relative to the status quo is given by $\theta-E[Y]$, where
$\theta:=E[Y(D(Z,1))]$ is the \emph{average counterfactual outcome} under the
new policy.\footnote{For ease of exposition, we assume that the marginal
  distribution of $Z$ is invariant to the policy---e.g., imposing a quota on
  release rates does not affect the assignment of judges. It is straightforward
  to accommodate known changes in the marginal distribution of $Z$ by weighting
  the average counterfactual outcome by the density of the counterfactual $Z$
  distribution relative to its status quo distribution.} Since identification of
$E[Y]$ is trivial, we focus on $\theta$ as the parameter of interest. We further note that in many settings, it may be reasonable to impose that the counterfactual policy has a monotone effect on the treatment in the sense that $D(z,1) \geq D(z,0)$ (which we formalize in \Cref{assumption:policy_monotonicity} below), in which case:
\begin{equation}\label{eq:complier_effect}
 E[Y(1) - Y(0) \mid D(Z,1) > D(Z,0)] = \frac{\theta - E[Y]}{E[D(Z,1)] - E[D]}.
\end{equation}
The left-hand side is the treatment effect for the \emph{policy compliers} who
are induced to adopt the treatment by the counterfactual policy (e.g., the
marginally-released defendants). Thus, when the ``first stage'' impact of the
counterfactual policy is identified, the average effect for policy compliers is
also a simple point-identified transformation of $\theta$. We follow the usual
convention in identification analysis of abstracting from observable covariates;
\Cref{sec:impl-deta} details how we incorporate covariates in our empirical
illustrations.

To fix ideas, as a running example, consider a judge \ac{IV} design, where $D$
is an indicator for release of a defendant, $Z$ is the identity of
a judge who is randomly assigned to the case, and the policy $a=1$ is an
encouragement or a quota to release more defendants. $Y$ may correspond to
various outcomes of interest. In many applications of leniency IV designs, the
outcome is binary, such as an indicator for pretrial misconduct
\parencite{adh22}, recidivism or criminal complaints \parencite{adh23}, mortality \parencite{norris_effect_2024}, employment \parencite{dgy18}, earning above the poverty rate
\parencite{kling06}, or high-school graduation \parencite{AiDo15}.
However, many of our results also allow for continuous outcomes, such as future
academic performance of the defendant's child \parencite{npw21} or average
earnings \parencite{kling06}.

Since the potential outcomes are indexed by the treatment only, our notation
implicitly embeds two exclusion restrictions: both the instrument ($z$) and the
counterfactual policy ($a$) affect the outcome only by impacting the treatment. In the
judge IV running example, the first exclusion restriction amounts to assuming
that the judges affect the defendants' outcomes only through their release decisions. This would
be violated if judges can make other decisions, such as set the terms of
probation or fees owed, that may directly affect outcomes
\parencite{mueller-smith15}. The exclusion restriction with respect to the
counterfactual policy likewise imposes that, say, a quota on how many defendants
are released affects defendants' outcomes only by changing which defendants are
released. This could be violated if releasing more defendants reduces
deterrence, thereby increasing crime directly; more generally, the exclusion
restriction rules out general equilibrium effects of the policy.

Additionally, we assume throughout that the instrument is as-good-as-randomly
assigned. In the judge \ac{IV} running example, this holds if judges are as-good-as-randomly assigned to defendants, which is the case in many of the empirical
papers cited above, at least once we condition on court-by-time fixed effects.
Together with the exclusion restrictions described above, random assignment implies that the
instrument is valid in that it is independent of the potential treatments and
potential outcomes,
\begin{equation}\label{eq:valid_iv}
  \{(Y(d), D(z, a))\}_{d, a\in\{0,1\}, z\in\mathcal{Z}}\indep Z.
\end{equation}
The identification power of the instrument comes from its influence on the
treatment.

\subsection{The identified set for the average counterfactual outcome}

To define the identified set for $\theta$, let $P$ denote the distribution of
the observable data $(Y, D, Z)=(Y(D(Z,0)), D(Z,0), Z)$, and let $P^{*}$ denote the
distribution of the model primitives, that is, the joint distribution of
potential outcomes, potential treatments, and the instrument,
$(Y(\cdot), D(\cdot, \cdot), Z)$. We say that $P^*$ \emph{generates} $P$ if the
implied distribution of $(Y, D, Z)$ under $P^*$ matches the observed data---i.e.,
if under $P^*$, the distribution of the triple $(Y(D(Z,0)), D(Z,0), Z)$ is $P$.

We encode any \emph{ex ante} restrictions on the model primitives by restricting
the family $\mathcal{P}^{*}$ of possible distributions for $P^{*}$. At minimum,
since we maintain IV validity throughout, we must have that
$\mathcal{P}^* \subseteq \mathcal{P}^{*}_{\textnormal{valid}}$, where
$\mathcal{P}^{*}_{\textnormal{valid}}$ is the set of distributions satisfying
the IV validity condition~\eqref{eq:valid_iv}. However, the
family $\mathcal{P}^{*}$ may be smaller if the researcher imposes other
restrictions on the model that we discuss in more detail below, such as IV
monotonicity. Given $\mathcal{P}^{*}$, the identified set for the model
primitives is given by the set of distributions in $\mathcal{P}^{*}$ that could
have generated the observed data:
\begin{equation*}
\mathcal{P}^*_{I}(P; \mathcal{P}^*) := \{P^* \in \mathcal{P}^*\colon P^*\;\text{generates}\;P \}.
\end{equation*}
The identified set for the average counterfactual outcome,
$\theta = E[Y(D(Z,1))]$, corresponds to the set of counterfactual means
consistent with these primitives:
\begin{equation*}
  \Theta_I(P; \mathcal{P}^*) := \{E_{P^*}[Y(D(Z,1))] \, :\, P^* \in \mathcal{P}^*_{I}(P; \mathcal{P}^*) \},
\end{equation*}
where we make explicit that the expectation is taken with respect to the
distribution $P^{*}$. We suppress this dependence whenever it doesn't cause
confusion.

For our analysis, it will be useful to distinguish between assumptions that
restrict the joint distribution of decisions across decision-makers---i.e.,
restrict the dependence between $D(z, \cdot)$ and $D(z', \cdot)$ for
$z \neq z'$---versus assumptions that only restrict the set of marginals
$\{(Y(\cdot), D(z,0), D(z,1))\}_z$ involving one $z$ at a time. We refer to the
former as \emph{cross-judge} restrictions, and the latter as \emph{marginal}
restrictions, as formalized by the next assumption.

\begin{assumption}[Marginal restrictions]\label{assumption:first_stage_restrictions}
  The marginal distributions
  $\{(Y(0), Y(1), D(z, 0), \allowbreak D(z, 1)) \}_{z\in\mathcal{Z}}$ are contained in some
  collection of distributions $\mathcal{R}^{*}$, specified by the researcher.
\end{assumption}
Next, we list examples of both marginal and cross-judge restrictions on the potential treatments that we
will consider in our analysis.
\Cref{sec:no_monotonicity,sec:iv_monotonicity,,sec:other_restrictions_help} then explore the identifying power
of these restrictions (and variants thereof).

\subsection{Examples of restrictions on potential treatments}

Knowledge of policy implementation details will typically allow us to impose
natural restrictions on the marginal distribution of $(D(z,1), D(z,0))$ for each
$z$, which can be encoded by an appropriate choice of $\mathcal{R}^*$ in
\Cref{assumption:first_stage_restrictions}. In the context of our running
example, for many counterfactual policies, the marginal release rates will
either be known or consistently estimable, allowing us to impose that
$E[D(z,1)] = \alpha_{1, z}$ for all $z$. For instance, under a quota policy that
requires judges to release at least a fraction $q$ of the defendants, we may set
$\alpha_{1, z}= \max\{E[D\mid Z=z], q \}$ for each $z$, so that each judge who
is currently below the quota increases their release rate to match it.

For quota policies or other policies that take the form of an
encouragement, it is natural to also impose the following marginal restriction:
\begin{assumption}[Policy monotonicity]\label{assumption:policy_monotonicity}
  $D(z,1)\geq D(z,0)$ for all $z$.
\end{assumption}
If a policy is a directive to treat everyone, such as a universal release
program, then policy monotonicity holds trivially since $D(z,1)=1$ for all $z$.
In the context of a quota policy, \Cref{assumption:policy_monotonicity} simply
imposes that any defendant released without the quota in place would also be
released when the quota is in place, which seems quite reasonable. Likewise, an
algorithm that flags low-risk defendants as candidates for release may be
expected to have a monotonic effect on release
rates.\footnote{\label{fn:release_algorithm}In some contexts, the algorithm may
  recommend release for some defendants, and recommend detention for others. In
  this case, we'd expect the algorithm to weakly increase release rates among
  those recommended to be released, and weakly decrease them among those
  recommended for detention. If the algorithmic recommendation is a
  function of observables, one could impose the appropriately signed version of monotonicity in each subpopulation.} We expect that for
many counterfactual policies of interest policy monotonicity will be reasonable,
and we will therefore maintain this assumption for many of our results.

A different type of monotonicity restriction that is commonly imposed is the
\textcite{ImAn94} IV monotonicity assumption:
\begin{assumption}[IV monotonicity]\label{assumption:instrument_monotonicity}
  For any pair $z, z'$, either $D(z,0) \geq D(z',0)$ (a.s.) or $D(z',0) \geq D(z,0)$ (a.s.).
\end{assumption}
This assumption allows one to interpret the \ac{TSLS} estimand as a \acf{LATE},
a weighted average of treatment effects for individuals whose treatment depends
on the instrument under the status quo, i.e., those for whom $D(z,0)$ varies
with $z$. In contrast to policy monotonicity, IV monotonicity involves
restrictions on the joint behavior of judges, so it is a cross-judge rather than a marginal
restriction. \Textcite{vytlacil02} shows that
\Cref{assumption:instrument_monotonicity} is equivalent to the existence of a
latent index $U$ (not depending on $z$) and thresholds $\alpha_{z}$ such
that
\begin{equation}\label{eq:u0_definition}
  D(z,0) = \1{U \leq \alpha_{z}}\quad\text{for all $z$}.
\end{equation}
Intuitively, this representation says that all judges have the same ranking of
defendants under the status quo ($U$), but potentially disagree on the
threshold to use to determine whether a defendant is released ($\alpha_{z}$). As
described in the introduction, \Cref{assumption:instrument_monotonicity} may
often be questionable in empirical contexts---judges may differ in their
rankings of defendants owing to idiosyncratic perceptions of risk, different
preferences over crime types, or differences in skill---and in fact is often
rejected using statistical tests \parencite[e.g.][]{fll23}, as well as direct
measurement when we observe multiple judges making a decision about the same
case \parencite{sigstad_monotonicity_2023}.

In addition to IV monotonicity, which imposes a common ranking $U$ under the
status quo, we will also consider the even stronger cross-judge assumption that
this common ranking does not change under the
counterfactual---\textcite{HeVy07i,HeVy05} call this condition policy
invariance.
\begin{assumption}[Policy invariance]\label{assumption:single_index}
  $D(z, a)=\1{\alpha_{z} + \beta_z \cdot a \geq U}$, with $\beta_z \geq 0$.
\end{assumption}
Under \Cref{assumption:single_index}, the only effect of the policy is that it
increases the thresholds judges use for release by a judge-specific amount
$\beta_{z}$. The policy invariance assumption underlies the use of the marginal
treatment effect (MTE) framework \parencite{HeVy99,HeVy05} for policy analysis.
In particular, the representation in~\cref{eq:u0_definition} can be used to
define the MTE curve, $MTE(u) = E[Y(1)-Y(0) \mid U=u]$. Under
\Cref{assumption:single_index}, we can then write the change in the average
outcome from implementing the counterfactual policy, $\theta - E[Y]$, simply as
a functional of the MTE curve. But without \Cref{assumption:single_index}, when
$U$ may not correspond to judges' rankings under the counterfactual, the MTE curve may be insufficient for
counterfactual policy analysis.

Note that \Cref{assumption:single_index} implies both
\Cref{assumption:policy_monotonicity} and
\Cref{assumption:instrument_monotonicity}, and thus will be questionable
whenever \Cref{assumption:instrument_monotonicity} is questionable. Another way
to see the difference between
\Cref{assumption:policy_monotonicity,assumption:instrument_monotonicity,,assumption:single_index}
is that if we had data under both the status quo and the counterfactual, we
could consider two possible instruments, $Z$ and $A$, where the latter is an
indicator for whether the counterfactual policy is implemented.
\Cref{assumption:policy_monotonicity} corresponds to the IV monotonicity
assumption where the binary policy variable $A$ is the instrument.
\Cref{assumption:instrument_monotonicity} corresponds to IV monotonicity
for the multivalued instrument $Z$, while \Cref{assumption:single_index}
corresponds to IV monotonicity for the two-dimensional instrument
$\tilde{Z} = (Z, A)$.
However, as argued in \textcite{HeUrVy06} and \textcite{mtw21}, with multiple
instruments, IV monotonicity is often implausibly strong, as it imposes strong
 restrictions on treatment choice.
Nevertheless, it is useful conceptually to consider what we can learn
under \Cref{assumption:single_index}, since as we will see below, for some
policy exercises \Cref{assumption:single_index} will be useful while
\Cref{assumption:instrument_monotonicity} will not. This suggests that
relaxations of \Cref{assumption:single_index} may be useful in practice, a topic
we turn to in \Cref{sec:other_restrictions_help} below.

\section{The identified set without monotonicity or policy invariance}\label{sec:no_monotonicity}

In this section, we derive a tractable characterization of the
identified set $\Theta_I$, without imposing IV monotonicity or policy
invariance. Previous work that has considered partial identification in related
IV settings has typically characterized the identified set by defining response
types (a.k.a.\ principal strata) based on the values of $D(\cdot, \cdot)$
\parencite[e.g.][]{kitagawa21, bai_identifying_2024}. The observable data on
$(Y, D, Z)$ is then a mixture of the distributions of the potential outcomes
across the response types, and one can characterize the identified set for the
primitives $\mathcal{P}_I^*(P;\mathcal{P}^{*})$ by searching over mixtures of
types that match the observed data. The identified set for $\theta$ can then be
calculated by computing the minimum and maximum value of $E_{P^*}[Y(D(Z,1))]$
among $P^* \in \mathcal{P}_I^*(P;\mathcal{P}^{*})$.

The challenge with this approach is that if one does not impose instrument
monotonicity, then the number of response types grows exponentially in the number
of judges $K$. Specifically, if we impose no restrictions on $D(\cdot, \cdot)$,
then there are four possible values of $(D(z,0), D(z,1))$ for each value of $z$,
and hence $4^K$ possible response types. Imposing policy monotonicity
(\Cref{assumption:policy_monotonicity}) rules out response types with
$(D(z,0), D(z,1)) = (1, 0)$ for some $z$, which still leaves $3^K$ possible response types. In
either case, the number of response types becomes extremely large with even a moderate
number of judges: when $K=30$, for example, we have
$3^K \approx 2 \cdot 10^{14}$ and $4^K \approx 10^{18}$, which makes
characterization of the identified set by type enumeration computationally
infeasible.

Fortunately, we will show that it is not necessary to enumerate response types
to derive the identified set for the counterfactual policy of interest. The key
observation is that we need not search over possible distributions $P^*$ for the
model primitives; it is sufficient to restrict our attention to the collection
of marginal distributions $\{(Y(0), Y(1), D(z, \cdot))\}_{z}$ that involve the
potential outcomes and the potential treatments \emph{for one judge at a time}. That is, we need not explicitly consider the dependence between $D(z,\cdot)$ and $D(z',\cdot)$ for $z \neq z'$.
To see why we need not consider the joint behavior of the judges, observe that
the counterfactual outcome can be written
\begin{equation*}
  \theta = E_{P^{*}}[Y(D(Z,1))] = \sum_z P(Z=z) \cdot E_{P^{*}}[Y(D(z,1)) \mid Z=z],
\end{equation*}
which depends on $P^*$ only through the collection of marginals for
$\{(Y(0), Y(1), D(z,1)) \}_z$ and the marginal distribution of $Z$. Likewise, the
probability distribution of the observable data takes the form
\begin{equation*}
  P(Y \in A, D=d, Z=z) = P(Z=z) \cdot P^{*}(Y(d) \in A, D(z,0) =d),
\end{equation*}
which depends on $P^*$ only through the collection of marginals for
$\{(Y(0), Y(1), D(z,0))\}_{z}$ and the marginal distribution of $Z$. It follows
that both the observable data distribution $P$ and the average outcome under the
counterfactual depend on $P^*$ only through the marginal distributions
$\{(Y(0), Y(1), D(z, \cdot))\}_{z}$ and the marginal distribution of $Z$---they do
not depend on the joint distribution of $D(z, \cdot)$ and $D(z', \cdot)$ for
$z \neq z'$. Moreover, if we are not imposing any cross-judge restrictions, then our constraints on the model also only restrict the marginals of $\{(Y(0), Y(1), D(z, \cdot))\}_{z}$. This suggests that to compute the identified set for $\theta$, we
can simply optimize over sets of marginals for $\{(Y(0), Y(1), D(z, \cdot)) \}_z$
that are consistent with the observable data.

\Cref{prop:marginals_one_z} and \Cref{cor:optmize_over_pi} below formalize these
observations. For ease of exposition, we state the results for the special case
with discrete outcomes, and defer the general statement, which involves more
notation, to \Cref{theorem:general_statement} in \Cref{sec:proof_marginals_one_z}. Suppose that the support $\mathcal{Y}$ of the outcomes is
discrete, so that the collection of marginal distributions
$\{(Y(0), Y(1), D(z, \cdot)) \}_{z}$ can be characterized by the collection of
marginal probability mass functions $\{\pi_{z}(\cdot) \}_{z}$, where
$\pi_{z}(y_{0}, y_{1}, d_{0}, d_{1})=P^{*}(Y(0)=y_{0}, Y(1)=y_{1}, D(z,0)=d_{0},
D(z,1)=d_{1})$. Our first result gives a simple way of verifying whether a
conjectured collection of marginals $\{\pi_{z}(\cdot)\}_{z}$ is consistent
with the data in the sense that there exists a distribution $P^{*}$ that
generates both the data distribution $P$ and the conjectured marginals.

\begin{prop}\label{prop:marginals_one_z}
  Suppose that $\mathcal{Y}$ is finite. There exists a joint
  distribution
  $P^* \in \mathcal{P}^*_{I}(P;\mathcal{P}^{*}_{\textnormal{valid}})$ with marginals $\{\pi_{z}(\cdot) \}_{z \in \mathcal{Z}}$ if and only if $\{\pi_{z}(\cdot) \}_{z \in \mathcal{Z}}$ satisfy the following three conditions:
\begin{enumerate}
\item\label{item:match_data} They match the observable data: for every $y \in \mathcal{Y}$,
  $z \in \mathcal{Z}$
    \begin{align*}
      \sum_{y_0 \in \mathcal{Y}, d_1 \in \{0,1\}} \pi_z
      (y_0, y, 1, d_1)&=P(Y=y, D=1 \mid Z=z), \\
      \sum_{y_1 \in \mathcal{Y}, d_1 \in \{0,1\}}
      \pi_z(y, y_1, 0, d_1)&=P(Y=y, D=0 \mid Z=z).
\end{align*}
\item\label{item:same_marginals} They imply the same distribution for
  $(Y(0), Y(1))$: for all $y_0,y_1 \in \mathcal{Y}$ and for any
  $z \in \mathcal{Z}$
  \begin{equation*}
    \sum_{(d_0, d_1) \in \{0,1\}^2} \pi_{z}(y_0, y_1, d_0, d_1) = \sum_{(d_0, d_1) \in \{0,1\}^2} \pi_{1}(y_0, y_1, d_0, d_1).
  \end{equation*}
\item\label{item:valid_pmf} They are valid probability mass functions: for all
  $(y_0, y_1, d_0, d_1) \in \mathcal{Y}^2 \times \{0,1\}^2$, and
  $z \in \mathcal{Z}$, $ \pi_z(y_0, y_1, d_0, d_1) \geq 0$ with
  $\sum_{(y_0, y_1, d_0, d_1) \in \mathcal{Y}^2 \times \{0,1\}^2} \pi_z(y_0, y_1, d_0, d_1) = 1$.
\end{enumerate}
\end{prop}
\Cref{prop:marginals_one_z} implies that to compute the identified set for
$\theta$, instead of searching over all distributions $P^{*}$, it suffices to
search over the marginal probability mass functions that satisfy the intuitive
conditions~\ref{item:match_data}--\ref{item:valid_pmf}. As the following
\namecref{cor:optmize_over_pi} shows, this holds so long as the restrictions on
$\mathcal{P}^{*}$ only take the form of marginal restrictions, as in
\Cref{assumption:first_stage_restrictions}.

\begin{corollary}\label{cor:optmize_over_pi}
  Suppose that $\mathcal{Y}$ is finite and that $P^{*}$
  satisfies~\cref{eq:valid_iv} and
  \Cref{assumption:first_stage_restrictions} for some convex $\mathcal{R}^{*}$,
  but no further restrictions are placed on $\mathcal{P}^{*}$, that is
  $\mathcal{P}^{*}=\mathcal{P}^{*}_{\textnormal{valid}}\cap \{P^{*}\colon
  \{\pi_{z}(\cdot)\}_{z \in \mathcal{Z}}\in\mathcal{R}^{*}\}$. Then
  ${\Theta}_{I}(P;\mathcal{P}^{*})$ is given by an interval, with the upper
  endpoint given by the optimization
  \begin{equation}\label{eq:theta_value}
    \sup_{\{\pi_{z}(\cdot) \}_{z \in \mathcal{Z}} \in \mathcal{R}^*} \sum_{z \in \mathcal{Z}} P(Z=z)
    \sum_{(y_0, y_1, d_0, d_1) \in \mathcal{Y}^2 \times \{0,1\}^2} \left(d_1 y_1 +
      (1-d_1)y_0 \right) \cdot \pi_{z}(y_0, y_1, d_0, d_1)
  \end{equation}
  subject to the constraints~\ref{item:match_data}--\ref{item:valid_pmf}
  in~\Cref{prop:marginals_one_z}; the lower endpoint is given by an analogous
  minimization.
\end{corollary}

\paragraph{Computation via linear programming}\label{sec:lp_implementation}

\Cref{cor:optmize_over_pi} implies that if $\mathcal{R}^*$ restricts the
marginals $\{\pi_z\}_{z \in \mathcal{Z}}$ linearly, the identified set can be
computed by solving a linear program that scales \emph{linearly} with the number
of instruments $K$.

To illustrate, suppose that the only restriction imposed by $\mathcal{R}^{*}$ is
that it restricts the fraction of defendants released under the counterfactual,
$E[D(z,1)]=\alpha_{1, z}$. For a given collection of marginal probability mass
functions $\{\pi_{z}(\cdot)\}_{z}$, let $\pi$ denote the vector of length
$4\abs{\mathcal{Y}}^{2}K$ that stacks all values of the marginals. Then the
value of the objective function in \cref{eq:theta_value} can be written
$\sum_{z, y_{0}, y_1, d_0, d_1} P(Z=z) (d_1 y_1 + (1-d_1) y_0) \pi_{z}(y_0, y_1,
d_0, d_1)= \omega' \pi$, where the weighting vector $\omega$ stacks the weights
$P(Z=z) (d_1 y_1 + (1-d_1) y_0)$. Furthermore, the
constraints~\ref{item:match_data}--\ref{item:valid_pmf}
in~\Cref{prop:marginals_one_z} can be written as
$2\abs{\mathcal{Y}}K+\abs{\mathcal{Y}}^{2}(K-1)+4\abs{\mathcal{Y}}^{2}K+1=O(\abs{\mathcal{Y}}^{2}K)$
linear constraints in $\pi$. Finally, the $K$ constraints
$E[D(z,1)]=\alpha_{1, z}$ can be written as
$\sum_{y_{0}, y_{1}, d_{0}}\pi_{z}(y_{0}, y_{1}, d_{0},1)=\alpha_{1, z}$, which are
also linear in $\pi$. It follows that the optimization problem
in~\cref{eq:theta_value} can be written as a linear program in
$O(K\abs{\mathcal{Y}}^{2})$ variables, subject to $O(K\abs{\mathcal{Y}}^{2})$
constraints. A related result appeared in \textcite{RiRo14}, who show that for
binary outcomes, bounds on $E[Y(d)]$ when $\mathcal{R}^{*}$ is unrestricted can
be obtained as a solution to $O(K^2)$ inequalities. Concurrent work by \textcite{song_categorical_2025} shows that bounds on the marginal distribution of $Y(\cdot)$ when $\mathcal{R}^*$ is unrestricted can be characterized by $O(K)$ inequalities. Their result, however, does not directly cover policy counterfactuals like $\theta = E[Y(D(Z,1))]$, which is a functional of the joint distribution of $Y(\cdot)$ and $D(\cdot, \cdot)$.

The key feature of this program is that its dimension scales linearly with $K$,
rather than exponentially, which makes it computationally fast even when $K$ is
large. The linear scaling is preserved if we restrict $\mathcal{R}^{*}$ in
other ways, so long as the restrictions are linear. For example, imposing policy
monotonicity (\Cref{assumption:policy_monotonicity}) amounts to adding the $K$
constraints that $\sum_{y_{0}, y_{1}}\pi_z(y_0,y_1, 1, 0) = 0$ for all $z=1, \dotsc, K$.

\begin{remark}[Discretizing $Y$]
  If $Y$ is continuous, one can obtain conservative bounds by considering a
  discretized version of $Y$. Given an initial grid of $Q$ points, let $Y^{ub}$
  be $Y$ rounded up to the nearest grid point. Since by construction
  $Y^{ub} \geq Y$, it follows that
  $E[Y^{ub}(D(Z,1))] \geq E[Y(D(Z,1))] = \theta$. Hence, computing the upper
  bound of the identified set for $E[Y^{ub}(D(Z,1))]$ yields a potentially
  non-sharp upper bound on $\theta$. Likewise, one can obtain a conservative
  lower bound by computing the lower bound of the identified set after rounding
  down $Y$ to the nearest grid point.
\end{remark}

\begin{remark}[Comparing multiple policies]
For simplicity, we have focused on evaluating a single counterfactual policy ($a=1$). In some settings, we may be interested in comparing the impacts of two counterfactual policies (e.g. a quota vs. automatic release), denoted $a=1$ and $a=2$. It is straightforward to extend the approach outlined above to this case by setting $\pi_z$ to correspond to the joint marginal distribution of $(Y(\cdot), D(z,0), D(z,1), D(z,2))$. The objective $E[ Y(D(Z,2)) - Y(D(Z,1))]$ is linear in $\pi_z$, and the dimension scales linearly in $K$ as before.
\end{remark}

\begin{remark}[Estimation and inference]
  For estimation and inference in our empirical illustrations, we exploit the
  fact that the data only enters the linear program through the
  $2\abs{\mathcal{Y}}K$ probabilities $P(Y=y,D=d\mid Z=z)$, which appear on the
  right-hand side of the data compatibility constraint~\ref{item:match_data} and in the constraints on $E[D(z,1)]$ under the quota policy we consider. We form plug-in estimates of the identified set by replacing these probabilities with their sample
  analogs. For inference, we use a projection approach, whereby we first form a
  confidence band for these data probabilities, and then optimize over both $\pi$ and data-probabilities within the confidence band. An advantage of this approach is that it can
  accommodate many-decision-maker asymptotics, where $K$ grows with the sample
  size. We provide details in \Cref{sec:impl-deta}.
\end{remark}

\section{Does IV monotonicity help tighten the identified set?}\label{sec:iv_monotonicity}

The previous section characterized the identified set for the counterfactual
outcome $\theta = E[Y(D(Z,1))]$ without imposing IV monotonicity. This section evaluates the extent to which imposing IV monotonicity helps tighten the identified set for $\theta$. We present
sufficient conditions under which imposing IV monotonicity alone does
not help tighten the identified set. If these conditions hold, then
the debate over the validity of IV monotonicity is somewhat of a red
herring from the perspective of learning about counterfactuals. A
researcher interested in counterfactual policy evaluation should therefore focus on alternative assumptions---such as policy invariance
(and relaxations thereof)---that may help tighten the identified set.

\paragraph{Simple example for intuition.} We begin with a simple example to
illustrate why imposing IV monotonicity need not help tighten the identified
set. Suppose that $a=1$ corresponds to a universal release policy, so that
$D(Z,1)=1$ with probability 1. Suppose further that $Y$ is binary and that
$Y(0)$ is known to equal zero with probability one. This restriction arises
frequently in the criminal justice setting, where, for example, one cannot fail
to appear in court if detained while awaiting trial.\footnote{There are other
  leniency designs in which one of the potential outcomes is known to equal
  zero. For example, \textcite{baron_discrimination_2024} study child welfare
  investigations and define the outcome $Y$ to be 1 if there is evidence of
  future misconduct in the child's original home, and thus by construction $Y=0$
  if a child is placed in foster care (i.e. $Y(1)=0$). Likewise,
  \textcite{cgy22} study whether doctors diagnose patients with pneumonia, and
  define $Y$ to be 1 if the patient subsequently shows signs of undiagnosed
  pneumonia. Thus, in their setting $Y(1)=0$ by construction.} The parameter of
interest then reduces to $\theta = E[Y(1)]$.

In this setting, it is straightforward to show that without IV monotonicity, the sharp identified set that we characterized in \Cref{cor:optmize_over_pi} corresponds to an intersection of judge-specific intervals, as suggested in \textcite{manski_nonparametric_1990} (Lemma 3.1 of \textcite{bai_identifying_2024} makes a similar observation). In particular, bounds
for $\theta$ using \emph{only} data on a given judge $z$ are given by the
interval
\begin{equation}\label{eq:judge_specific_I}
\mathcal{I}_z = [P(Y=1, D=1 \mid Z=z), 1-P(Y=0, D=1 \mid Z=z)].
\end{equation}
The lower bound corresponds to the value of $E[Y(1)]$ if everybody not released
by $z$ has $Y(1)=0$; the upper bound to the value if everybody not released by
$z$ has $Y(1)=1$. The identified set for $\theta$ then corresponds to the
intersection of the judge-specific intervals,
$\Theta_I = \bigcap_z \mathcal{I}_z$.

On the other hand, under IV monotonicity, it is straightforward to show
that the identified set corresponds to the judge-specific interval for the most
lenient judge under the status quo, i.e. $\mathcal{I}_{z_{\max}}$, where
$z_{\max} = \argmax_z E[D(z,0)]$ is the identity of the most lenient judge under
the status quo. Specifically, note that IV monotonicity implies that any
defendant not released by judge $z_{\max}$ is also not released by any other
judge, so that $D(z_{\max},0) =0$ implies $D(z,0) = 0$ for all $z$. It follows
that the distribution of the observable data does not depend at all on the value
of $Y(1)$ for defendants not released by judge $z_{\max}$, i.e., any binary
distribution for $Y(1) \mid D(z_{\max},0)= 0$ is compatible with the observable
data. Note, however, that by iterated expectations, $E[Y(1)]$ is a weighted
average of $Y(1)$ for the defendants released and not released by judge
$z_{\max}$, i.e.
$E[Y(1)] = P^{*}(D(z_{\max},0)=1) P^{*}(Y(1) =1 \mid D(z_{\max},0)=1) +
P^{*}(D(z_{\max},0)=0) P^{*}(Y(1)=1 \mid D(z_{\max},0) = 0)$. The first term is
point-identified from the defendants that judge $z_{\max}$ releases. The
identified set therefore corresponds to the interval obtained by placing trivial
bounds on the outcome for defendants not released by $z_{\max}$,
$P^{*}(Y(1)=1 \mid D(z_{\max},0) = 0) \in [0,1]$, which yields the interval
$\mathcal{I}_{z_{\max}}$.

We thus see that without imposing IV monotonicity, we obtain the identified set
$\bigcap_z \mathcal{I}_z$, whereas if we do impose monotonicity, we obtain the
identified set $\mathcal{I}_{z_{\max}}$. Since adding an assumption
cannot widen the identified set, and since $\bigcap_z \mathcal{I}_z$ is
contained in $\mathcal{I}_{z_{\max}}$, it follows that when IV monotonicity
holds, the identified set is the same whether we impose IV monotonicity or not. On the
other hand, if IV monotonicity does not hold and is rejected by the data, it may
be the case that the set $\bigcap_z \mathcal{I}_z$ is a strict subset of
$\mathcal{I}_{z_{\max}}$: the sharp identified set from
\Cref{cor:optmize_over_pi} above without imposing monotonicity is then tighter than
the \emph{naïve} identified set $\mathcal{I}_{z_{\max}}$ that assumes IV
monotonicity holds in the data (if IV monotonicity is rejected, then the
identified set under IV monotonicity is formally empty).

The intuition for why IV monotonicity does not help in this setting is that
under IV monotonicity, we learn nothing about the defendants not released by the
most lenient judge, $z_{\max}$. By contrast, if IV monotonicity is violated,
then some defendants not released by $z_{\max}$ \emph{may be released} by
another judge $z'$, and thus the data provides some information about their
value of $Y(1)$. Hence, IV monotonicity is in this sense the \emph{least
  favorable} configuration of judge release decisions, and thus imposing IV
monotonicity does not help us tighten the identified set.

\paragraph{More general sufficient conditions.} Our simple example above had two
salient features: (i) one of the potential outcomes was known ($Y(0)=0$), and
(ii) the policy encouragement was very strong ($D(z,1)=1$). We next present a
generalization showing that either of these features on its own is sufficient
for IV monotonicity not to tighten the identified set. We first show that if
$Y(0)=0$, then the identified set for any counterfactual policy satisfying
policy monotonicity ($D(z,1) \geq D(z,0))$ does not depend on whether we impose
IV monotonicity, regardless of whether the policy encouragement is strong or
weak. Second, we show that if both potential outcomes are non-trivial, then IV
monotonicity does not help to tighten the identified set provided that one
assumes the policy encouragement is ``sufficiently strong'', as formalized in
condition~(iii) below.

We first consider the setting where one of the potential outcomes is known. For
simplicity of notation, we establish our result in the context of a
counterfactual policy that imposes an \emph{average quota policy} that restricts
the average value of $D(Z,1)$. Fix $\alpha \in [0,1]$. Let
$\mathcal{P}_{\alpha}^*$ be the set of distributions over the random vector
$(Y(0), Y(1), D(\cdot, \cdot), Z)$ satisfying the \ac{IV} validity
condition~\eqref{eq:valid_iv}, policy monotonicity
(\Cref{assumption:instrument_monotonicity}), and also
\begin{enumerate}
\item [(i)] \emph{$\alpha$-average quota policy:} $P^*(D(Z,1)=1) = \alpha$,
\item [(ii)] \emph{Known outcome under $D=0$:} $P^*(Y(0) = 0)=1$.
\end{enumerate}
Let $\mathcal{P}^*_{\alpha, Mon}$ be the subset of distributions in
$\mathcal{P}_{\alpha}^*$ that further satisfy IV monotonicity
(\Cref{assumption:instrument_monotonicity}).

\begin{prop}[No identifying power of monotonicity with known $Y(0)$]\label{prop:monotonicity_is_not_helpful}
If the identified set $\Theta_I(P; \mathcal{P}^*_{\alpha, Mon})$ is non-empty, then \Cref{assumption:instrument_monotonicity} has no identifying power in the sense that $\Theta_I(P; \mathcal{P}^*_{\alpha, Mon})=\Theta_I(P; \mathcal{P}_{\alpha}^*)$.
\end{prop}

Although condition (i) focuses on a policy
that imposes an average quota, it is possible to extend the proposition to the
case in which there is a unit-specific quota $\alpha_{z,1}$, at the cost of
additional notation.

We next consider the setting with a sufficiently strong policy encouragement, in the sense that the policy satisfies
the following condition:
\begin{enumerate}
\item[(iii)]\label{item:iv_strong_encouragement}\emph{Sufficiently strong
    encouragement:} $P^*(D(z,1) =1 \mid D(z_{\max},0) =1)=1$ for all $z$.
\end{enumerate}
Condition (iii) imposes that all defendants who would be released by the most lenient judge under
the status quo would be released with probability 1 under the counterfactual.
This is trivially satisfied under a universal release program ($D(z,1)=1$), but
is somewhat weaker. For example, consider a program that releases all defendants
except those predicted to be high-risk. Then (iii) will be satisfied if
the most lenient judge is not currently releasing any such high-risk defendants
under the status quo. Our next result shows that monotonicity also has no
identifying if we replace condition (ii) defined above with condition (iii) in
the statement of \Cref{prop:monotonicity_is_not_helpful}.

\begin{prop}[No identifying power of monotonicity with strong
  encouragement]\label{prop:monotonicity_is_not_helpful_strong_encouragement}
  The result in \Cref{prop:monotonicity_is_not_helpful} holds if condition (ii)
  is replaced with condition (iii) in the definitions of $\mathcal{P}^*_{\alpha}$
  and $\mathcal{P}^*_{\alpha, Mon}$.
\end{prop}

\begin{remark}[Interaction of IV monotonicity with other constraints]
  \Cref{prop:monotonicity_is_not_helpful,prop:monotonicity_is_not_helpful_strong_encouragement}
  provide conditions under which IV monotonicity \emph{alone} does not help
  tighten the identified set for the average counterfactual outcome. As shown by
  \textcite{machado_instrumental_2019}, if we impose additional restrictions
  ---such as assume that all units share the same sign of the treatment effect
  $Y(1)-Y(0)$---then adding monotonicity can help to further tighten the
  identified set.
\end{remark}

\begin{remark}[Weaker monotonicity conditions]
  \textcite{fll23} propose a weakening of IV monotonicity called \emph{average monotonicity}.
  \Cref{prop:monotonicity_is_not_helpful,prop:monotonicity_is_not_helpful_strong_encouragement}
  provide conditions under which imposing IV monotonicity does not help tighten
  the identified set. It follows immediately that under the same conditions,
  imposing the weaker notion of average monotonicity also does not help tighten
  the identified set.
\end{remark}

\begin{remark}[Relationship to literature]
  \Cref{prop:monotonicity_is_not_helpful,prop:monotonicity_is_not_helpful_strong_encouragement}
  can be viewed as an extension of some important existing results in the
  literature which show that IV monotonicity has no identifying power
  for the average treatment effect parameter or the marginal distributions of potential outcomes
  \parencite{BaPe97,kitagawa21,bai_identifying_2024, RiRo14}. An important difference
  vis-\`a-vis our result is that we focus on more general policy
  counterfactuals.
\end{remark}

\begin{remark}[Can monotonicity help?]
  In \Cref{sec:monotonicity_helps}, we give an example of a policy and a data distribution such that IV
  monotonicity sharpens the identified set for $\theta$. This point relates to an observation in
  \textcite{kamat_identifying_2019}, showing that the identified set for the
  joint distribution of $(Y(1), Y(0))$ shrinks under IV monotonicity even though the identified sets for the marginals of $Y(1)$ and $Y(0)$ do not. For
  policies that do not satisfy conditions (ii) or (iii) above,
  IV monotonicity can potentially help by restricting the possible couplings of the
  marginal distributions of $Y(0)$ and $Y(1)$, even though it has no identifying
  power for the marginal distributions.
\end{remark}

\section{What assumptions help tighten the identified set?}\label{sec:other_restrictions_help}

The previous section outlined sufficient conditions under which IV monotonicity alone does not help to tighten the identified set. We now explore other assumptions that can potentially help to tighten it. We begin by showing that policy invariance can be helpful in some settings where IV monotonicity is not. Policy invariance may be too strong an assumption in practice, however, and we therefore introduce relaxations of policy invariance that may be more plausible, but nevertheless tighten the identified set. We then briefly discuss other economically motivated restrictions that may be reasonable and further help tighten the identified set.

\subsection{Policy invariance and its relaxations}

We begin by providing an intuitive example where the conditions of
\Cref{prop:monotonicity_is_not_helpful} are satisfied, so that IV monotonicity
alone does not help tighten the identified set, but policy invariance is
helpful.

\paragraph{Example: quota policy.} Suppose that $Y(0)=0$ (e.g., you cannot
commit a crime while in jail). Consider a quota policy that requires all judges
to increase their release rate to match that of the most lenient judge under the
status quo, so that $E[D(z,1)] = E[D(z_{\max},0)]$ for all $z$, where again
$z_{\max}$ denotes the identity of the most lenient judge. It seems reasonable
to impose policy monotonicity in this setting ($D(z,1) \geq D(z,0)$), which
states that defendants who would be released without the quota would also be
released with the quota.

Strengthening policy monotonicity by imposing policy invariance leads to point
identification. Intuitively, under policy invariance, all judges have the same
ranking of defendants under both the status quo and counterfactual, but they
disagree on the cutoff for when a defendant should be released. The quota policy
forces them to use the same cutoff as the most lenient judge, so that the
counterfactual outcomes for all judges will match that of the most lenient judge
under the status quo:
$E[Y(D(Z,1))] = E[Y(D(z_{\max},0))] = E[Y \mid Z =z_{\max}]$.

By contrast, \Cref{prop:monotonicity_is_not_helpful} implies that imposing
IV monotonicity in addition to policy monotonicity does not help tighten
the identified set. In fact, the identified set may be trivial in the sense that
it implies trivial bounds for the outcomes of policy compliers,
$E[Y(1)\mid D(Z,1)>D(Z,0)]$. Consider, for example, a setting where there are
two judges, who respectively release 10\% and 20\% of defendants. Under IV
monotonicity alone, the first judge could choose to match the quota by marginally releasing
only status quo \emph{never-takers}---those not released by either judge under
the status quo. The data does not restrict $Y(1)$ for these individuals, so we
only have trivial bounds $[0,1]$ on their treated outcomes,
implying trivial bounds for the policy complier treatment effect. Correspondingly, by
\cref{eq:complier_effect}, the identified set for $\theta - E[Y]$ under IV monotonicity is
given by multiplying the unit interval by the mass of policy compliers. Thus, policy
invariance---which restricts agreement under judges under both the
counterfactual and the status quo---has strong identifying power, leading to
point identification. By contrast, IV monotonicity---which only restricts agreement
under the status quo---does not help tighten the trivial bounds on the
identified set.

Of course, policy invariance may be implausibly strong in
many applied settings. Indeed, since policy invariance is \emph{stronger} than
IV monotonicity, if researchers doubt the validity of IV monotonicity (or it is
rejected by the data) then they must necessarily doubt the validity of policy
invariance as well. Nevertheless, the fact that there are relevant settings
where policy invariance can help tighten the identified set, but IV monotonicity
cannot, suggests that considering \emph{relaxations} of policy invariance may be
a more natural starting point than considering relaxations of IV monotonicity.

\paragraph{Relaxing policy invariance with disagreement bounds.} To that end, we
now introduce a relaxation of policy invariance that may be more plausible but
nevertheless help to tighten the identified set. At a high level, policy
invariance requires that judges perfectly agree on the ranking of defendants; we
consider a relaxation that bounds the extent to which they can disagree. More
concretely, consider two distinct judges $z$ and $z'$, and two policies
$a, a'\in \{0,1\}$. Without loss of generality, suppose that
$E[D(z, a)] \geq E[D(z', a')]$ so that judge $z$ releases more people under policy
$a$ than $z'$ does under $a'$. Under policy invariance, judge $z$ under policy
$a$ must release \emph{all} defendants released by $z'$ under policy $a'$, so
that $P^*(D(z, a)=1 \mid D(z', a')=1) = 1$. Such perfect agreement is likely to be
too strong in many settings. Nevertheless, it also may be unreasonable to expect
that the two judges perfectly \emph{disagree}, so that judge $z$ under policy
$a$ releases \emph{none} of the defendants released by $z'$ under $a'$. A natural middle-ground is to impose
\begin{equation}\label{eqn:disagreement-bounds}
P^*(D(z, a)=0 \mid D(z', a')=1)\leq \delta_{z, z', a, a'},
\end{equation} so that judge $z$ under policy $a$ disagrees with no more than $\delta_{z, z', a, a'}$ fraction of defendants released by judge $z'$ under policy $a'$. This nests policy invariance as the special case with $\delta_{z, z', a, a'} = 0$, but allows for non-trivial disagreement for $\delta_{z, z', a, a'} \in (0,1)$.\footnote{\textcite{IcTa00} assume that $D(z,0) = D(z',1)$ (a.s.) for a known pair (or set of pairs) $z,z'$. This can be formalized by imposing \eqref{eqn:disagreement-bounds} with $\delta_{z,z',0,1} =0$ and  $\delta_{z',z,1,0} = 0$.}

\paragraph{Calculating the identified set under disagreement bounds.} The
disagreement bound in \cref{eqn:disagreement-bounds} imposes restrictions on the
joint distribution of $(D(z, a), D(z', a'))$. At first glance, such restrictions
appear to be difficult to incorporate into the linear programming approach for
calculating the identified set described in \Cref{sec:lp_implementation}, which
only optimizes over the marginals $\{\pi_{z}(\cdot)\}$, and does not enumerate
the joint distribution of $(D(z, a), D(z', a'))$. It turns out, however, that given a set of marginals $\{\pi_{z}(\cdot)\}$, there is a
simple formula for the minimal value of $P^{*}(D(z, a)=0 \mid D(z', a')=1)$
consistent with a joint distribution of primitives that matches these marginals
and satisfies policy monotonicity. This allows us to still tractably compute the
identified set under policy monotonicity by optimizing only over the marginals $\{\pi_{z}(\cdot)\}$ even
while imposing disagreement bounds such as \cref{eqn:disagreement-bounds}. This
is formalized in the following results, which give analogs to
\Cref{prop:marginals_one_z} and \Cref{cor:optmize_over_pi} that allow for
imposing disagreement bounds.

\begin{prop}\label{prop:coupling_result_with_disagreement} Suppose that
  $\mathcal{Y}$ is finite. Fix disagreement bounds
  $\delta_{z, z', a, a'} \in [0,1]$ for all $z, z', a, a'$. Let
  $\mathcal{P}^*_{\textnormal{DB}}$ be the subset of distributions in
  $\mathcal{P}^{*}_{\textnormal{valid}}$ satisfying
  \Cref{assumption:policy_monotonicity} and the disagreement bounds in
  \cref{eqn:disagreement-bounds} for all $z, z', a, a'$. Consider a collection
  $\{\pi_{z}(\cdot)\}_{z\in\mathcal{Z}}$ of marginals. There exists a joint distribution
  $P^* \in \mathcal{P}^*_{I}(P;\mathcal{P}^{*}_{\textnormal{DB}})$ with marginals $\{\pi_{z}(\cdot) \}_{z \in \mathcal{Z}}$ if and only if the
  marginals $\{\pi_{z}(\cdot) \}_{z \in \mathcal{Z}}$ are consistent with \Cref{assumption:policy_monotonicity} and satisfy conditions~\ref{item:match_data}--\ref{item:valid_pmf} in
  \Cref{prop:marginals_one_z} as well as
  \begin{itemize}
      \item[4.]\label{item:disagreement-bound} Disagreement bound: For all $z, z', a, a'$,
      \begin{multline}\label{eqn:disagreement-bound-inequality}
        \sum_{y_1, y_0} \min\{\pi(y_0, y_1, D(z, a) = 1),
          \pi(y_0, y_1, D(z', a') = 1) \} \\
        \geq (1-\delta_{z, z', a, a'}) \cdot \pi(D(z', a')=1),
      \end{multline}
      where for any $(y_0,y_1) \in \mathcal{Y}^2$ and $(a, z)\in \{0,1\} \times \mathcal{Z}$,
      \begin{align*}
      \pi(y_0, y_1, D(z, a)=1) &:= \sum_{d_a=1, d_{1-a} \in \{0,1\}} \pi_{z}(y_0,y_1, d_0, d_1), \\
      \pi(D(z, a)=1) &:= \sum_{d_{a}=1, d_{1-a}\in \{0,1\}} \sum_{(y_0,y_1) \in \mathcal{Y}^2} \pi_{z}(y_0, y_1, d_0, d_1).
 \end{align*}
  \end{itemize}
\end{prop}

\begin{corollary}\label{cor:optmize_over_pi2}
  Suppose that $\mathcal{Y}$ is finite and let $\mathcal{P}^* = \mathcal{P}_{DB}^* \cap \{P^{*}\colon
  \{\pi_{z}(\cdot)\}_{z \in \mathcal{Z}}\in\mathcal{R}^{*}\}$ for some convex $\mathcal{R}^*$, where $\mathcal{P}_{DB}^*$ as defined in \Cref{prop:coupling_result_with_disagreement}. Then
  ${\Theta}_{I}(P;\mathcal{P}^{*})$ is given by an interval, with the upper
  endpoint given by the optimization
  \begin{equation}\label{eq:theta_value_w_disagreement}
    \sup_{\{\pi_{z} \}_{z \in \mathcal{Z}} \in \mathcal{R}^*} \sum_{z \in \mathcal{Z}} P(Z=z)
    \sum_{y_0,y_{1}, d_0, d_1 \in \mathcal{Y}^2 \times \{0,1\}^2} \left(d_1
      y_1 + (1-d_1)y_0 \right) \cdot \pi_{z}(y_0, y_1, d_0, d_1)
  \end{equation}
  subject to the
  constraints~\ref{item:match_data}--\ref{item:disagreement-bound} given
  in~\Cref{prop:marginals_one_z,prop:coupling_result_with_disagreement} and
  subject to policy monotonicity ($\pi_z(y_0, y_1, 1, 0) = 0$ for all
  $z,y_0,y_1$); the lower endpoint is given by an analogous minimization.
\end{corollary}

The intuition for \Cref{prop:coupling_result_with_disagreement} is as follows. Note that \cref{eqn:disagreement-bounds} can equivalently be written as a lower bound on the joint probability $P^*(D(z, a) = 1, D(z', a') = 1)$,
\begin{equation}
 (1-\delta_{z, z', a, a'}) P^*(D(z', a')=1)\leq P^*(D(z, a) = 1, D(z', a') = 1). \label{eqn:reformulated-disagreement}
\end{equation}
By the law of total probability, the right-hand side equals
\begin{multline*}
 \sum_{y_1,y_0} P^*(Y(0)=y_0,Y(1) = y_1, D(z, a) = 1, D(z', a') = 1)\leq \\
 \sum_{y_1,y_0} \min\{P^*(Y(0)=y_0, Y(1) = y_1, D(z, a) = 1), P^*(Y(0)=y_0,Y(1) = y_1, D(z', a') = 1) \},
\end{multline*}
where the inequality uses the fact that $P(A \cap B) \leq \min\{ P(A), P(B) \}$. \Cref{eqn:disagreement-bound-inequality} follows from replacing
$P^*(D(z, a) = 1, D(z', a') = 1)$ in \eqref{eqn:reformulated-disagreement} with the upper bound given in the previous
display. This turns out not to come at the cost of sharpness, however, because
given a set of marginals for $(Y(1), Y(0), D(z, \cdot))$, there always
exists a coupling such that the inequality in the previous display holds with
equality. In particular, the proof of
\Cref{prop:coupling_result_with_disagreement} shows that the upper bound is
achieved under a latent threshold crossing model wherein judges agree on the
rankings of defendants \emph{conditional} on their potential outcomes, so that
$D(z, a) = \1{\alpha_{z, y_0, y_1} + \beta_{z, y_0, y_1} a \geq V_{y_0, y_{1}}}$ for
a latent index $V_{y_0, y_1}$ that depends only on the potential outcomes but not
on $z$.\footnote{If one does not impose policy monotonicity, then \eqref{eqn:disagreement-bound-inequality} is still implied by \eqref{eqn:disagreement-bounds}, and so one can obtain an outer set for $\Theta_I$ by running the analog to the optimization in \Cref{cor:optmize_over_pi2} without the policy monotonicity constraint. Without policy monotonicity, however, the interval obtained may no longer be sharp.}

We note that it is straightforward to
impose~\eqref{eqn:disagreement-bound-inequality} in a linear program by
introducing auxiliary parameters corresponding to the $\min\{\cdot, \cdot\}$
terms. Specifically, \cref{eqn:disagreement-bound-inequality} holds if and only
if there exist constants $\eta^{y_0,y_1}_{z, z'} \geq 0$ that are less
than the arguments of the $\min\{\cdot, \cdot\}$,
\begin{align}\label{eq:extra_restriction_dp_eta1}
  \eta^{y_0, y_1}_{z, z'} &\leq \pi(y_0,y_1, D(z, a) = 1), \\
  \eta^{y_0, y_1}_{z, z'} &\leq \pi(y_0,y_1, D(z', a') = 1)\label{eq:extra_restriction_dp_eta2}
\end{align}
such that
\begin{equation}\label{eq:extra_restriction_dp}
  \sum_{y_1, y_0} \eta^{y_1, y_0}_{z, z'} \geq (1-\delta_{z, z'}) \pi(D(z', a') =1).
\end{equation}
The above formulation is linear in both the marginals $\{\pi_{z}(\cdot)\}$ and
the auxiliary parameter $\eta$. Therefore, we can implement the disagreement
bound in \cref{eqn:disagreement-bounds} by adding the constraints in
\cref{eq:extra_restriction_dp_eta1,eq:extra_restriction_dp_eta2,eq:extra_restriction_dp}
to the linear programming formulation in \Cref{sec:lp_implementation}, and
augmenting the parameter vector $\pi$ by $\eta$.\footnote{We note that if one bounds disagreement rates among all pairs $z,z'$, this introduces $O(K^2)$ auxilliary parameters, so the dimension of the LP scales quadratically in $K$ rather than linearly.}

\paragraph{Bounds on average disagreements.} The results above can easily be
adapted if we wish to weaken \cref{eqn:disagreement-bounds} by only
bounding the \emph{average} pairwise disagreement between groups of judges. For
concreteness, consider the quota policy from our empirical applications in
\Cref{sec:empirical-apps}, which asks the bottom 90\% of judges to match the
release rate of the most lenient decile. Let $\mathcal{Q}$ denote the set of
judges in the bottom nine deciles of leniency, who are subject to the quota, and
let $\mathcal{Q}^c$ denote the top decile of judges, who are not subject to the
quota. To bound the average disagreement between the two groups of judges under
the counterfactual, we can impose the average disagreement probability bound
\begin{equation}\label{eq:dp_bound}
P^*(D(Z_{\mathcal{Q}},1)=0\mid D({Z}_{\mathcal{Q}^{c}},1)=1)\leq \overbar{\DP},
\end{equation}
where $Z_{\mathcal{Q}}$ is a randomly picked judge from the group $\mathcal{Q}$,
and ${Z}_{\mathcal{Q}^{c}}$ is defined similarly. In other words, \eqref{eq:dp_bound} places an upper-bound $\overbar{\DP}$ on the average probability
that a judge in the quota group disagrees with the release decision of a judge
in the top decile. Similarly to the pairwise disagreement bounds, this average
disagreement probability bound can be implemented as a linear program without losing
sharpness of the resulting identified set. In particular, we can introduce auxiliary parameters $\eta_{z, z'}^{y_{1}, y_{0}}$, for each
$z\in\mathcal{Q}$ and $z'\in\mathcal{Q}^{c}$, and impose the linear
constraints \cref{eq:extra_restriction_dp_eta1,eq:extra_restriction_dp_eta2},
as well as the constraint
$\sum_{y_{1}, y_{0}}\frac{1}{\abs{\mathcal{Q}}\abs{\mathcal{Q}^{c}}}\sum_{z\in\mathcal{Q}, z'\in\mathcal{Q}^{c}}
\eta_{z, z'}^{y_{0}, y_{1}}\geq (1-\delta)P^*(D(Z_{\mathcal{Q}^{c}}, 1)=1)$. The attractive feature of the average disagreement probability bound is that it
only requires calibration of a single parameter, $\overbar{\DP}$. In contrast,
imposing pairwise disagreement bounds requires calibration of $\delta_{z, z'}$
for each pair of judges, which may be more difficult in practice.\footnote{Our
  average disagreement bound focuses on the case where we only bound
  disagreement under the counterfactual. This is sufficient for the quota
  counterfactual, since we assume judges in the top decile do not change their
  behavior under the counterfactual. For other policies, it may be informative to impose bounds on
  $P^*(D(Z_{\mathcal{Q}}, a)=0\mid D({Z}_{\mathcal{Q}^{c}},a')=1)$ for pairs of
  $(a, a')$ other than $(1,1)$.}

\paragraph{Calibrating disagreement bounds.} In certain special cases, we can observe the decisions of multiple judges (or other decision-makers) on the same cases, which allows us to directly measure how frequently they disagree. Such special settings can be used to calibrate reasonable bounds on disagreement rates. For example, \textcite{sigstad_monotonicity_2023} studies settings where panels of judges rule on the same cases. Likewise, \textcite{agarwal_combining_2023} ask multiple radiologists to provide diagnoses on the same x-rays, and thus observe the decisions of multiple doctors on the same cases.\footnote{In a similar vein, \textcite{ady23} observe the recommendation of both a pretrial services officer and a judge. One could use this to calibrate disagreement bounds among pairs of judges if one were willing to assume that judges and pretrial services officers disagree at similar rates as pairs of judges.} In our application to bail judges in \Cref{sec:empirical-apps} below, we calibrate the disagreement bound $\overbar{\DP}$ using a Gaussian model parameterized to match the data in \textcite{sigstad_monotonicity_2023}.\footnote{In principle, one could use the
  \textcite{sigstad_monotonicity_2023} data to directly calculate the sample
  analog of the disagreement probability
  $P^*(D(Z_{\mathcal{Q}},0)=0\mid D({Z}_{\mathcal{Q}^{c}},0)=1)$. However, the
  disagreement probability is not invariant to the marginal release rates (if
  two judges have release rates equal to 99\%, they cannot disagree on more
  than 2\% of cases; by contrast, two judges with release rates of 50\% could disagree 100\% of the time). We therefore use a Gaussian signal model to compute the
  implied disagreement probability when the marginal release rates are matched
  to those in our application (see \Cref{sec:impl-deta}). The Gaussian signal model implicitly assumes that
  judges serving on panels rule the same way as when making decisions on their
  own. To weaken this assumption, one could estimate a richer model that allows
  for spillovers \parencite[e.g.][]{IaSh12}.}

\paragraph{Limitations of policy invariance.} The quota
example given above showed that policy invariance can help tighten the
identified set (and even restore point identification) in a setting where IV monotonicity alone is not informative. It turns out, however, that there are some situations in which neither IV monotonicity nor policy invariance helps to tighten the identified set. To gain intuition, consider our earlier example of a universal release
policy such that $D(Z,1)=1$ with probability 1. In that case, by assumption, all
judges perfectly agree on their decisions under the counterfactual policy. It
follows that imposing policy invariance---which imposes IV monotonicity and
restricts agreement under the counterfactual---is \emph{equivalent} to imposing
IV monotonicity in this setting, which we know is not helpful by
\Cref{prop:monotonicity_is_not_helpful_strong_encouragement}. In settings like this, researchers will therefore have to turn to
alternative restrictions to try to tighten the identified set.

\subsection{Outcome restrictions}\label{sec:outcome_restrictions}
We now outline several other restrictions that researchers may consider imposing to help tighten the identified set.

\paragraph{Policy complier (PC) bounds.} In the pretrial release setting, bail judges are
legally instructed to release defendants based on pretrial misconduct potential
if released ($Y(1)$). Consider a policy that encourages judges to release more
defendants (e.g., a quota). Unless the judges' estimates of pretrial misconduct
risk are terribly miscalibrated, policy
compliers---individuals only released under the quota---should be on average higher
risk than those currently released, but lower risk than those who are never
released:
\begin{equation}\label{eq:pc_bound}
  E[Y(1)\mid D(z,0)=1]\leq E[Y(1)\mid D(z,0)<D(z,1)]\leq E[Y(1)\mid D(z,1)=0].
\end{equation}
When the marginal release rates under the counterfactual ($E[D(z,1)]$) are
point-identified, \cref{eq:pc_bound} can be written as a linear restriction on the
marginals $\pi_{z}(\cdot)$, and is thus is straightforward to incorporate in the
linear program for the identified set bounds.

\paragraph{Treatment effect bounds.} It may sometimes be reasonable to impose
bounds on the sign or magnitude of the treatment effects. For example, in some
settings, it may be natural to impose the monotone treatment response assumption
that $Y(1) \geq Y(0)$, as in \textcite{manski97} and
\textcite{machado_instrumental_2019}. This is straightforward to impose by
setting $\pi_{z}(y_0, y_1, d_0, d_1) =0$ whenever $y_1 < y_0$. Similarly, in
some settings it may be reasonable to impose that the treatment effect for the
policy compliers is not too large, e.g. $E[Y(1)-Y(0) \mid D(z,0)<D(z,1)] \leq c$, which can be
implemented as a linear constraint when the share of policy compliers is identified, similar to \cref{eq:pc_bound}.\footnote{\textcite{fll23} argue, for example,
  that it may be reasonable to restrict the magnitude of the treatment effect of
  incarceration among compliers between different sets of judges; here we impose
  an analogous constraint on the compliers with respect to a policy of
  interest.}

\paragraph{Outcome disparity bounds.} Some counterfactual policies directly
encourage a subset $\mathcal{Q}$ of decision-makers to act more like a benchmark
group $\mathcal{Q}^{c}$. For example, in implementing a quota policy directive
that asks the set $\mathcal{Q}$ to match the treatment rate of the group
$\mathcal{Q}^{c}$, a policy-maker may encourage the decision-makers to also
emulate their outcomes. We consider such a setting in our empirical application
in \Cref{sec:suff-county-pros} below. In such cases, an alternative to bounding
the disagreements between their treatment decisions is to bound the disparity in
the outcomes between the two sets of decision-makers by imposing
\begin{equation}\label{eq:od}
  \abs{E[Y(D(Z_{\mathcal{Q}}, 1))-Y(D(Z_{\mathcal{Q}^{c}}, 1))]}\leq \overbar{\OD},
\end{equation}
where, again, $Z_{\mathcal{Q}}$ is a randomly picked judge from the group $\mathcal{Q}$, and ${Z}_{\mathcal{Q}^{c}}$ is defined similarly.

To interpret this bound, it can be helpful to relate it to the disagreement
probability bound in \cref{eq:dp_bound} above. To this end, let
$\mathcal{C}_{z, z'}=\1{D(z, 1)=0, D(z', 1)=1}$ denote the event that $z'$ would
release an individual but not $z$ under the counterfactual. Assume that the
treatment rates of both groups are equal,
$P^*(D(Z_{\mathcal{Q}^c},1)=1)=P^{*}(D(Z_{\mathcal{Q}}, 1)=1)=q$ for some $q$.
Then
$E[\mathcal{C}_{{Z}_{Q}, {Z}_{\mathcal{Q}^{c}}}]=E[\mathcal{C}_{{Z}_{Q^{c}},
  {Z}_{\mathcal{Q}}}]=P^*(D(Z_{\mathcal{Q}}, 1)=0, D(Z_{\mathcal{Q}^c},1)=1)$.
Using this, we may write the left-hand side of \cref{eq:od} as
\begin{multline*}
  | E[Y(D(Z_{\mathcal{Q}}, 1))-Y(D(Z_{\mathcal{Q}^{c}},
  1))] | =P^*(D(Z_{\mathcal{Q}}, 1)=0, D(Z_{\mathcal{Q}^c},1)=1)\\
  \times | E[Y(1)-Y(0)\mid \mathcal{C}_{{Z}_\mathcal{Q}, {Z}_{\mathcal{Q}^{c}}} = 1]
  -E[Y(1)-Y(0)\mid \mathcal{C}_{{Z}_{\mathcal{Q}^{c}}, {Z}_\mathcal{Q}} = 1]| .
\end{multline*}
This is the \emph{product} of the aggregate disagreement probability,
$P^*(D(Z_{\mathcal{Q}}, 1)=0, D(Z_{\mathcal{Q}^c},1)=1)=P^*(D(Z_{\mathcal{Q}}, 1)=0\mid D(Z_{\mathcal{Q}^c},1)=1)q$ times the \emph{difference} between treatment effects for individuals about whom there is disagreement. Thus, if we bound the treatment effect difference by $\Delta_{TE}$, imposing \cref{eq:dp_bound} implies that \cref{eq:od} holds with
\begin{equation}\label{eq:od_dp_relation}
\overline{\OD}=\Delta_{TE} \overbar{\DP}\cdot q.
\end{equation}
This relationship can be used to gauge the strength of the outcome disparity
restriction, as we illustrate in \Cref{sec:suff-county-pros} below.

\section{Empirical applications}\label{sec:empirical-apps}

\subsection{NYC bail judges}\label{sec:nyc-bail-judges}

\paragraph{Background and motivation.} We study pretrial release in New York City (NYC) using
data from \textcite[][henceforth ADH22]{adh22}. A judge ($Z$) decides whether to
release a defendant or hold them in pretrial detention. The defendant is
categorized as released ($D=1$) if they are released without preconditions or
after paying money bail; otherwise they are categorized as detained ($D=0$). We
then see whether a defendant commits pretrial misconduct ($Y \in \{0,1\}$) by
either failing to appear for a court hearing or by being arrested for a new
crime. By construction, a defendant cannot commit pretrial misconduct if they
are detained, so $Y(0) = 0$. We also observe the defendant's race, denoted by
$R \in \{b, w\}$, where $b$ and $w$ denote black and white.

We consider two counterfactual policy exercises in this context. First, we
consider a policy that releases all defendants. \Textcite{albright22} studies a
universal release program for defendants (meeting certain conditions) in
Kentucky, so this policy counterfactual directly addresses what would happen if
such a program were implemented for ADH22's sample of defendants in NYC\@.
Moreover, as we describe below, identification of the parameter of interest in
ADH22, the disparate impact parameter, is equivalent to identification of
race-specific average outcomes under this counterfactual. As a result, our
analysis also delivers bounds on the disparate impact parameter without imposing
the parametric restrictions that allowed ADH22 to point-identify this parameter
by extrapolating from the observed variation in the data using identification-at-infinity type arguments.

Second, we consider a quota policy that asks the judges in the bottom 90\% of
leniency to match the release rate of the top 10\%. The original ADH22 analysis
considered counterfactuals similar to this using a hierarchical MTE model,
whereby a subset of the judges were subject to release quotas. Furthermore,
several policies discussed in the literature encourage stricter judges to
release more defendants while still giving the  judges discretion. For example, beginning in 2016, New York City began
encouraging judges to release some defendants through a supervised release
program rather than requiring them to post bail
\parencite{skemer_pursuing_2020}. Likewise, Kentucky passed a law in 2011 that,
among other things, changed the presumptive default from monetary bail to
non-monetary release, yet provided judges with the discretion to override the
default \parencite{stevenson_assessing_2018}. Both policies were found to
increase the fraction of defendants who were released, but did not increase this
rate to one. Our quota exercise can be thought of as a simple approximation to
reforms like these if one assumes that their main practical impact is to induce
the bottom 90\% of judges to increase their release rate to match the most lenient decile. Of course, if
one had domain-specific knowledge about the likely impact of such reforms on
release rates, the assumptions about their ``first-stage'' impact could be
modified accordingly to make our quota counterfactual resemble them more
closely.

\paragraph{Disparate impact in ADH22.} ADH22's primary goal is to estimate the
``disparate impact'' of pretrial release decisions by race. Their disparate
impact parameter can be written as a function of observable probabilities and
$E[Y(1) \mid R]$. Hence, learning about the disparate impact parameter is
isomorphic to learning about the mean outcome (for each race) under the
counterfactual in which everyone is released. To make this precise, let
$FPR_r := P(D=1 \mid Y(1) =1, R=r)$ be the false positive rate for race $r$,
i.e., the fraction of defendants who are released despite having $Y(1)=1$.
Observe that we can write
$FPR_r = \frac{P(D=1, Y=1 \mid R=r)}{P^*(Y(1) =1 \mid R=r)}$. The numerator is
an observable probability, and thus the only challenge in learning about
$FPR_{r}$ is learning about the misconduct rate under the counterfactual in
which everyone is released, $P^*(Y(1) =1 \mid R=r)$. ADH22 similarly define
$TPR_r := P(D=1 \mid Y(1) =0, R=r)$ to be the true positive rate, which we can
write as $TPR_r = \frac{P(D=1, Y=0 \mid R=r)}{1-P^*(Y(1) =1 \mid R=r)}$. Again,
the numerator is an observable probability. ADH22 define the disparate impact
parameter as a weighted average of the differences in $FPR_r$ and $TPR_r$ across
races,
\begin{equation*}
    \Delta := (1-E[Y(1)]) \cdot (TPR_{b}-TPR_{w}) + E[Y(1)] \cdot (FPR_{b}-FPR_w),
\end{equation*}
where $E[Y(1)]=\sum_{r\in\{b, w\}}P(R=r)P^{*}(Y(1)=1\mid R=r)$. This
relationship allows us to infer bounds on $\Delta$ from bounds on the
race-specific counterfactual outcomes if everyone is released,
$E[Y(1) \mid R=r]$.

\paragraph{Identification in ADH22.} ADH22 consider an
identification-at-infinity type argument for $E[Y(1) \mid R]$. Let
$P_{z, r} = E[D(z,0) \mid R=r]$ be the judge-specific release rate for race $r$. ADH22 consider various
functional forms for $E[Y(1) \mid D=1,P_{z, r}=p, R=r]$ and then extrapolate to $p=1$.
Specifically, they report linear, local linear, and quadratic extrapolations.
For comparison with our nonparametric approach, we replicate the point estimates
and confidence intervals based on these three parametric extrapolation
approaches.

\paragraph{Data.} We analyze aggregated statistics from the main estimation
sample in ADH22, which consists of the universe of arraignments made in NYC
between 2008--2013 involving white or black defendants charged with a felony or
misdemeanor, where the defendant is not already serving jail time for an
unrelated charge. The sample comprises \nycNw\ cases involving white defendants
and \nycNb\ cases involving black defendants. We do not have access to the
individual microdata, and only observe estimates of the judge-specific release
rates ($E[D(z,0)\mid R=r]$) and misconduct rates $E[Y(1)\mid D(z,0)=1,R=r]$ that
are obtained from linear regressions that adjust for court-by-time fixed
effects, along with the number of observations for each judge and race (these
estimates are plotted in Figure 2 in ADH22). The covariate adjustment accounts
for the conditional random assignment of the judges; see \Cref{sec:impl-deta} for details.
Because we do not have access to the microdata, our inference ignores the
statistical uncertainty stemming from the covariate adjustment: we treat the
covariate-adjusted rates as if they were sample means. We show below that our
replication of the results in ADH22 yields very similar standard errors as in
the original, suggesting that the impact of the covariate adjustment on
inference is minimal. To avoid small-sample issues stemming from observing only
a few cases per judge, we pool all judges with fewer than 300 race-specific
cases into a single judge, which leaves us with $K=\nycKmainb$ and
$K=\nycKmainw$ distinct values for the instrument for the samples of black and
white defendants, respectively.\footnote{Since the quota policy we consider
  below applies a quota to judges below the 90th percentile of leniency, we
  separately pool judges with fewer than 300 cases who lie below the 90th
  percentile of leniency, and those lying above it.
  \Cref{tab:estimates-delta-disag,tab:quota_nyc_disag} show that our inference
  is virtually unchanged when we do not pool.}

\paragraph{Summary statistics.}

\Cref{fig:nyc_raw} plots the judge-specific misconduct rates against their
release rates. We see that the misconduct rates increase sharply with the
release rate, which is intuitive since by construction $Y=0$ for non-released
defendants. Correspondingly, the \ac{TSLS} estimate of the effect of release on
misconduct, which is equivalent to fitting a weighted least squares line through
this scatterplot, is quite high, and equals \nycIVb\% for blacks and \nycIVw\%
for whites (with tight standard errors: \nycIVbse\ and \nycIVwse). The
average release rate equals \nycEDb\% for blacks and \nycEDw\% for whites.
Unconditionally, the average misconduct rate equals \nycEYb\% for blacks and
\nycEYw\% for whites; conditional on release, these rates equal \nycEYbD\% and
\nycEYwD\%.

\begin{figure}[tp]
\centering
\caption{Judge- and race-specific sample release and misconduct rates, based on
  NYC bail judge data from \textcite{adh22}.}\label{fig:nyc_raw}
\input{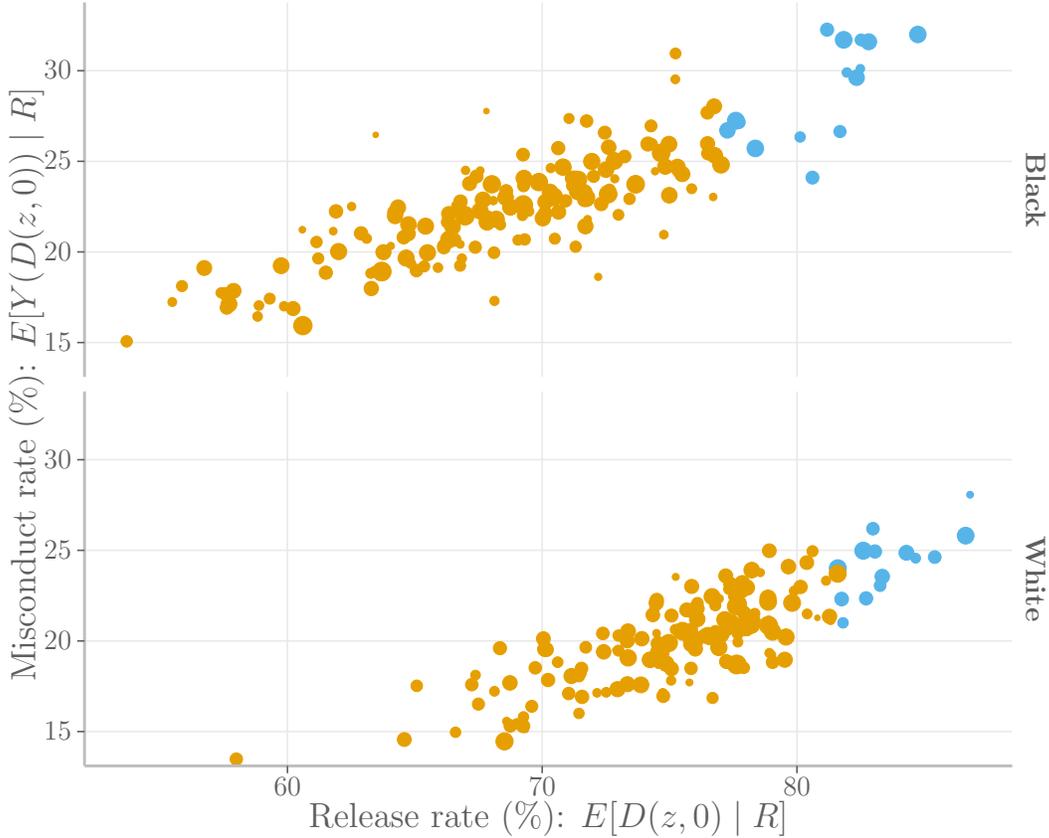} \floatfoot{\emph{Notes:} Both rates adjusted for
  court-by-time fixed effects. Each dot corresponds to a single judge; dot size
  scales with the number of cases handled by the judge. Judges with fewer than
  300 cases are pooled into a single dot. The top decile of most lenient judges
  is plotted in blue, the bottom 90\% are plotted in orange.}
\end{figure}

\paragraph{Results for universal release policy.} We compute plug-in estimates of the
identified set for $E[Y(1) \mid R]$, the race-specific average outcomes under a
universal release program, along with 95\% confidence intervals computed by
projection (see \Cref{sec:impl-deta} for inference details). Because $Y(0) = 0$,
the sharp bounds on $E[Y(1) \mid R=r]$ imposing only IV validity (i.e.,
randomization and exclusion) are given by the intersection of judge-specific
intervals given in \cref{eq:judge_specific_I} above. For illustration,
\Cref{fig:id-set-illustration} shows the form of the bounds (along with simultaneous CIs) for
black defendants: we compute sample analogs of the judge-specific bounds in
\cref{eq:judge_specific_I}, and intersect them to estimate the identified set.
While under monotonicity all the information is obtained from the most lenient
judge, we see that in practice the estimate of the identified set exploits
information from other judges and thus is tighter than the interval obtained by
looking at the most lenient judge alone.

\begin{figure}[tp]
    \centering
    \input{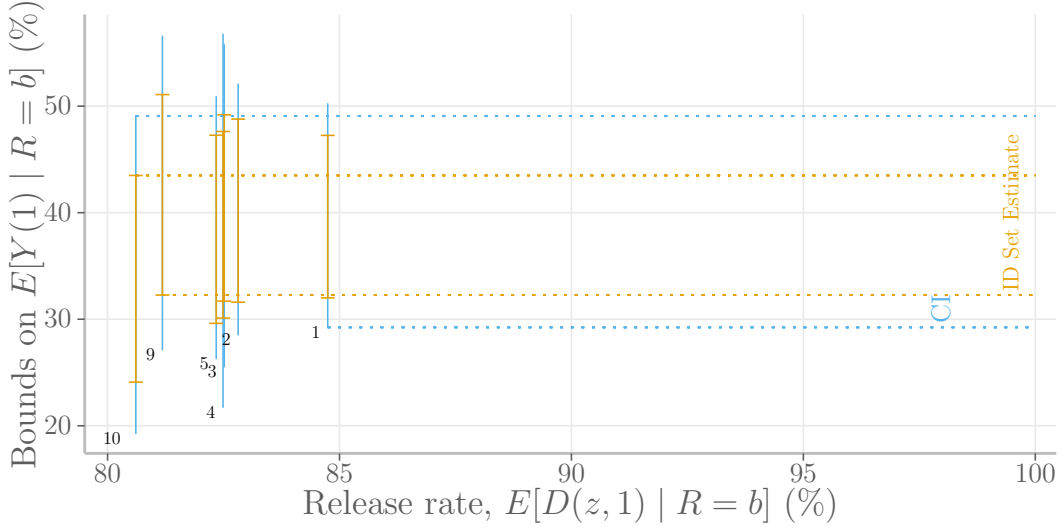}
    \caption{Illustration of identified set for $E[Y(1) \mid R=b]$ using NYC
      bail judge data.}\label{fig:id-set-illustration}
    \floatfoot{\emph{Notes:} The $x$-axis plots estimates of judge-specific
      release rates for each judge, adjusted for court-by-time fixed effects,
      while the $y$-axis plots estimates (in orange) and sup-$t$ confidence intervals
      (in blue) of bounds on $E[Y(1) \mid R=b]$ using data on each judge
      separately. The plug-in estimate of the identified set and the associated
      confidence interval are formed by intersecting the judge-specific bounds and CIs,
      and are depicted on the right-hand side (the orange dotted lines give the plug-in estimates, and the blue dotted lines the CI). For clarity, we plot
      the judge-specific bounds only for a subset of the judges. Judges are
      numbered by their leniency, with $1$ corresponding to the most lenient
      judge.}
\end{figure}

\begin{table}[htbp]
  \begin{threeparttable}
    \caption{Estimates for universal release policy and the disparate impact
      parameter using NYC bail judge data from
      \textcite{adh22}.}\label{tab:estimates-delta-arnold}
    \begin{tabular*}{\linewidth}{@{\extracolsep{\fill}}lcc cc c@{}}
      &\multicolumn{2}{@{}c}{Blacks} & \multicolumn{2}{@{}c}{Whites}\\
      \cmidrule(rl){2-3}\cmidrule(rl){4-5}
      &$E[Y(1)\mid R]$ & {PC TE}
      &$E[Y(1)\mid R]$ & {PC TE}& $\Delta$\\
      Specification & (1) & (2) & (3) & (4) & (5)\\
      \midrule
      Valid IV only & $[32.3, 43.5]$ & $[30.8, 67.9]$ & $[28.1, 39.2]$ & $[31.9, 79.2]$ & $[4.0, 8.2]$\\
 & $(29.2, 49.1)$ & $(20.9, 86.3)$ & $(23.4, 41.9)$ & $(12.0, 90.7)$ & $(1.0, 9.6)$\\
Linear extrap. & $40.0$ & $56.3$ & $33.5$ & $54.8$ & $5.4$\\
 & $(38.9, 41.0)$ & $(52.9, 59.6)$ & $(32.3, 34.6)$ & $(49.8, 59.8)$ & $(5.0, 5.8)$\\
Quadratic extrap. & $42.8$ & $65.6$ & $36.2$ & $66.4$ & $4.8$\\
 & $(39.2, 46.4)$ & $(53.8, 77.4)$ & $(32.7, 39.7)$ & $(51.5, 81.4)$ & $(3.4, 6.3)$\\
Local linear extr. & $43.6$ & $68.3$ & $35.0$ & $61.3$ & $4.3$\\
 & $(39.1, 48.1)$ & $(53.5, 83.0)$ & $(31.0, 38.9)$ & $(44.6, 78.0)$ & $(2.9, 5.7)$\\
\\[-0.9em]
      \bottomrule
    \end{tabular*}
    \begin{tablenotes}
      \begin{footnotesize}
      \item \emph{Notes}: Cols.~(1) and (3) report estimates of the average
        counterfactual outcome under the universal release policy,
        $E[Y(1)\mid R]$, expressed in percentage points, separately for blacks and whites; cols.~(2) and (4)
        report the estimates of the treatment effect for policy compliers,
        $E[Y(1)-Y(0)\mid D(Z, 0)=0]$ implied by these estimates. Col.~(5)
        reports estimates of the disparate impact parameter $\Delta$. For
        specifications leading to set identification, identified set estimates
        are reported in brackets. 95\% confidence intervals in parentheses.
        ``Valid IV only'' assumes only IV validity
        (\cref{eq:valid_iv}). Linear, quadratic and local linear
        extrapolation correspond to the parametric extrapolation methods
        considered in ADH22.
      \end{footnotesize}
    \end{tablenotes}
  \end{threeparttable}
\end{table}

Columns (1) and (3) of \Cref{tab:estimates-delta-arnold} give the estimates and
confidence intervals for the race-specific average misconduct rates computed
using our linear program as well as the parametric extrapolation methods in
ADH22\@. For direct comparability to the confidence intervals using our
approach, we compute standard errors without accounting for covariate
adjustment.\footnote{ADH22 also weight their extrapolations by the inverse of
  the estimated sampling variance for the judge-level covariate-adjusted outcome
  means. Since we do not have these weights, we instead weight by the number of
  defendants the judge released.} \Cref{tab:adh-replication-vs-orig} in the
appendix shows that the results are very similar to those reported in ADH22,
suggesting that accounting for covariate adjustment doesn't substantively affect
inference.

To help interpret the estimates, we also report estimates of the treatment
effects for policy compliers,
$E[Y(1)-Y(0)\mid D(Z,1)>D(Z,0)]=E[Y(1)-Y(0)\mid D(Z,0)=0]$ in columns (2) and
(4). The lower bound for policy compliers is above the misconduct rate for white
defendants released under the status (\nycPCTElw\% vs.\ \nycEYwD\%) and very similar to the status quo rate for black defendants (\nycPCTElb\% vs.\ \nycEYbD\%), suggesting that the marginally-released defendants will have similar or higher crime rates to those released under the status quo. This is intuitive, since we expect judges to be releasing defendants who they deem to
be the lowest risk. The confidence intervals are also reasonably tight, allowing us
to conclude that a universal release program will lead to a substantial increase
in the misconduct rates relative to the status quo: for blacks, they imply that
the misconduct rate will increase by at least $\nycCIbincrease$ percentage
points, which is a substantial increase relative to the status quo rate, $\nycEYb\%$.

We also compute an estimate of the identified set and confidence intervals for
the disparate impact parameter $\Delta$. Comparing the resulting estimates to
those from the parametric extrapolation methods, we see in
\Cref{tab:estimates-delta-arnold} that our estimate of the identified set,
$\nycDeltal$--$\nycDeltau$\%, is somewhat wider than the range of estimates we
obtain based on the parametric extrapolations, \nycDeltalp--\nycDeltaup\%. Our
95\% CIs are also informative: they equal $(\nycDeltalCI, \nycDeltauCI)$\%, so we
can reject the null hypothesis of no disparate impact. Thus, the conclusion in
ADH22 that $\Delta > 0$ is robust to dropping their functional form
assumptions.\footnote{In a robustness check, ADH22 construct non-sharp bounds on
  $E[Y(1) \mid R]$ and $\Delta$ using the aggregate mean outcome among released
  defendants, rather than constructing the sharp bounds by intersecting the
  judge-specific intervals. This yields wider bounds than what we obtain. For
  example, the plug-in bounds on $\Delta$ from this approach are $-1$\% to
  10\%.}

\paragraph{Results for the quota policy.}

\begin{table}[tp]
  \centering
\begin{threeparttable}
  \caption{Estimates for quota policy using NYC bail judge data from
    \textcite{adh22}.}\label{tab:quota_nyc}
  \begin{tabular*}{\linewidth}{@{\extracolsep{\fill}}lcccc cccc@{}}
    &\multicolumn{2}{@{}c}{no PC bound} & \multicolumn{2}{@{}c}{PC bounds}\\
    \cmidrule(rl){2-3}\cmidrule(rl){4-5}
    &Policy effect & {PC TE}
    &Policy effect & {PC TE}\\
    Specification & (1) & (2) & (3) & (4)\\
    \midrule
    \multicolumn{3}{@{}l}{A\@: Blacks}\\
    Reallocation & $5.6$ & $53.2$ & \\
 & $(5.2, 6.1)$ & $(49.2, 57.3)$ & \\
Constant treatment & $5.2$ & $49.7$ & \\
 & $(5.0, 5.5)$ & $(47.4, 52.0)$ & \\
Valid IV only & $[0.0, 10.5]$ & $[0.0, 100.0]$ & $[3.4, 7.0]$ & $[32.3, 66.5]$\\
 & $(0.0, 10.7)$ & $(0.0, 100.0)$ & $(1.8, 10.7)$ & $(16.4, 100.0)$\\
$\overline{DP}=0.025$ & $[4.7, 6.1]$ & $[44.9, 58.0]$ & $[4.9, 5.9]$ & $[46.7, 56.3]$\\
 & $(2.8, 7.9)$ & $(26.3, 76.7)$ & $(2.8, 7.9)$ & $(26.3, 76.7)$\\
\\
    \multicolumn{3}{@{}l}{B\@: Whites}\\
    Reallocation & $3.8$ & $56.6$ & \\
 & $(3.3, 4.2)$ & $(49.7, 63.5)$ & \\
Constant treatment & $3.2$ & $48.1$ & \\
 & $(3.0, 3.5)$ & $(44.5, 51.6)$ & \\
Valid IV only & $[0.0, 6.7]$ & $[0.0, 100.0]$ & $[1.8, 5.2]$ & $[26.2, 77.2]$\\
 & $(0.0, 6.9)$ & $(0.0, 100.0)$ & $(0.7, 6.9)$ & $(9.9, 100.0)$\\
$\overline{DP}=0.021$ & $[2.4, 5.1]$ & $[35.7, 76.9]$ & $[2.5, 5.0]$ & $[37.7, 75.2]$\\
 & $(1.2, 5.9)$ & $(17.0, 91.1)$ & $(1.2, 5.9)$ & $(17.0, 91.1)$\\
\\[-0.9em]
    \bottomrule
  \end{tabular*}
  \begin{tablenotes}
    \begin{footnotesize}
    \item \emph{Notes}: Cols.~(1) and (3) report estimates of the policy effect
      $E[Y(D(Z,1))-Y]$, expressed in percentage points, for a quota policy that asks the bottom 90\% of judges to
      match the release rate of the most lenient decile (which equals
      $\nycQuotab\%$ for blacks and $\nycQuotaw\%$ for whites). Cols.~(2) and
      (4) report the treatment effect for policy compliers implied by the policy
      effect estimates in cols.~(1) and (3). Cols.~(1) and (2) report estimates
      without imposing the PC bound in \cref{eq:pc_bound}, while cols.~(3) and
      (4) impose it. 95\% confidence intervals in parentheses; these do not
      account for covariate adjustment. The rows labeled $\overbar{\DP}$ report bounds imposing the disagreement bound in \eqref{eq:dp_bound} with the stated upper bound on disagreement between judges above and below the 90th percentile of leniency. For specifications leading to set
      identification, identified set estimates are reported in brackets.
  \end{footnotesize}
\end{tablenotes}
\end{threeparttable}
\end{table}

To benchmark the effect of the quota policy, we first compute the policy effect
of a simpler policy that achieves the same release rates: reallocating all cases
to the most lenient decile. Here the policy effect is point identified under
random assignment of the instrument, and it is given by the difference between
the current average misconduct rate versus the average misconduct rate for
individuals assigned to judges in the most lenient decile.\footnote{We note that
  for analyzing the impact of such a reallocation policy, we do not actually
  require IV exclusion.} We also consider a simple parametric benchmark that
assumes homogeneous treatment effects, which also leads to point identification:
the policy effect is identified by the TSLS estimate, multiplied by the change
in the release rate relative to the status quo. As shown
in~\Cref{tab:quota_nyc}, both benchmarks imply a similar policy effect, implying
a misconduct rate for policy compliers equal to about 50\% for both whites and blacks.

We then estimate the policy effect of the quota under a conservative specification that
assumes only \ac{IV} validity and policy monotonicity. Unlike for the universal
release policy, the bounds for this policy are trivial in the sense that they
imply the trivial bounds $[0, 100]$ for the policy complier treatment effect.
The reason for this is that we cannot rule out that there is a sufficient mass
of instrument never-takers under the status quo (those with $D(z,0)=0$ for all
$z$) for whom we only have trivial bounds on $Y(1)$, and without further
restrictions, policy compliers may be drawn from this pool in an adversarial
way.

However, we can tighten these bounds by imposing two additional sets of
restrictions that exploit the institutional details. First, we can impose the aggregate disagreement probability bound in \cref{eq:dp_bound}. As described in detail in \Cref{sec:impl-deta}, we set the upper bound on disagreement, $\overbar{\DP}$, by calibrating a Gaussian signal model similar to that considered in
\textcite{cgy22} to \textcite{sigstad_monotonicity_2023}'s data on panels of judges ruling on criminal cases in the São Paulo Appeal Court. We assume that the judges observe correlated Gaussian signals $U_{z}$, releasing the
individual if the signal is high enough. To calibrate the correlation matrix, we focus on the first three judges to rule in each case in \textcite{sigstad_marginal_2024}'s sample.\footnote{At least three judges rule on each case. Additional judges may provide an opinion in some cases, but only if the initial three judges disagree. To estimate disagreement rates, we therefore focus on the first three judges to avoid selection bias related to the initial three judges' decisions.} The implied signal correlation between a randomly selected judge in the
bottom 90\% and one in the most lenient decile is 0.989.\footnote{The raw
  disagreement rates are as follows: for a randomly-selected pair of judges $A$
  and $B$ with judge $A$ in the top decile of leniency and judge $B$ in the
  bottom nine deciles in Sigstad's data, we have $P(D_A=1, D_B=0) = 3.8\%$ and
  $P(D_A=1, D_B=0) = 0.6\%$, where $D_A, D_B$ are the decisions of judges $A, B$, respectively.} To calibrate $\overbar{\DP}$ in the NYC data, we set the correlation between any pair of signals $U_{z}$ and $U_{z'}$, where $z$ is a judge in the bottom 90\% and $z'$ is in the top decile, to this value in our Gaussian signal model, and set the judges' cutoffs for release to match those under our quota
counterfactual. We calculate the value of $\overbar{\DP}$ as the average disagreement probability in this model,
which yields
$\overbar{\DP}=\nycDPb$ blacks and \nycDPw\ for whites.

Second, the sole legal
objective of bail judges is to allow most defendants to be released before trial
while minimizing the risk of pretrial misconduct. In line with this objective, we can impose the policy complier bounds given by \cref{eq:pc_bound} for
each judge $z$ in the bottom 90\%.

\Cref{tab:quota_nyc} gives the results for the quota policy when we impose either the calibrated disagreement bound or the policy complier bounds (or both). Confidence intervals are computed by projection, as detailed in
\Cref{sec:impl-deta}.

With the $\overbar{\DP}$ bound imposed, the identified set estimate is fairly tight, implying that the treatment effects
for policy compliers are not too different from the TSLS estimates. The estimates of the bounds exceed the
status quo misconduct rates among those currently released,
$E[Y(1)\mid D(Z,1)=1]$: these rates equal $\nycEYbD$\% for blacks and $\nycEYwD$\%
for whites, while the lower bounds on the average value of $Y(1)$ for policy
compliers, reported in column (2), equal $\nycEYlb$\% and $\nycEYlw$\%,
respectively. While the confidence intervals are wider, the lower endpoints are
close to the status quo misconduct rates, particularly for blacks ($\nycEYlbCI$\% vs.\ $\nycEYbD$\%), implying the
policy compliers have similar or higher misconduct rates than the inframarginal
individuals (those currently released). Additionally, imposing policy complier bounds helps tighten the point estimates, but not the confidence intervals
once the $\overbar{\DP}$ bound is imposed.

\subsection{Suffolk County prosecutors}\label{sec:suff-county-pros}

\paragraph{Background.} We reanalyze the data from \textcite[henceforth
AHD23]{adh23}, who are interested in the effect of non-prosecution of nonviolent
misdemeanor offenses on the subsequent criminal justice involvement of the
defendants. Here the decision-makers ($Z$) are assistant district attorneys
(ADAs), who decide whether to prosecute a case after an initial arraignment
($D=1$) or drop it ($D=0$). We then see whether a criminal complaint against
the defendant has been filed within two years of arraignment ($Y\in\{0,1\}$).
ADH23 use the jackknife IV estimator to estimate the treatment effect of
non-prosecution, and find robustly negative IV estimates; they also estimate an MTE curve that is downward-sloping and negative for nearly all of its support. Based on this
analysis, they conclude (ADH23, p.~1455):

\begin{footnotesize}
  \begin{quote}
    The results of our analysis imply that if all arraigning ADAs acted more
    like the most lenient ADAs in our sample when deciding which cases to
    prosecute, Suffolk County would likely see a reduction in criminal justice
    involvement for these nonviolent misdemeanor defendants.
  \end{quote}
\end{footnotesize}

If we assume homogeneous treatment effects, then the negative IV estimates do
indeed imply that increasing non-prosecution would lower criminal complaints. Likewise, if we impose policy invariance, then the negative MTE curve implies that increasing the non-prosecution rate would lower criminal complaints. We are interested in evaluating the
robustness of this conclusion to dropping the assumptions of homogeneous treatment effects and policy invariance. Importantly, when ADH23 test IV monotonicity
(\Cref{assumption:instrument_monotonicity}) for each of the 9 courts in their
dataset using the test developed by \textcite{fll23}, the test rejects
IV monotonicity (and hence also policy invariance) in three of the courts. Correspondingly, we conduct the
analysis without imposing \Cref{assumption:instrument_monotonicity} or
\Cref{assumption:single_index}.

To apply our policy evaluation approach, we need to make precise the way in which one would encourage stricter ADAs to act ``more like the most lenient ADAs''. As a benchmark, we start by analyzing a policy that reallocates cases to the more lenient ADAs. We next consider a scenario in which, rather than literally re-allocating cases, the policy encourages the stricter ADAs to increase their release rates to match that of the more lenient ones. For comparability with our previous application, we consider a quota on the bottom 90\% of ADAs to match the
non-prosecution rate of the most lenient decile. As with the previous application, we can view this as an approximation to other encouragement policies that have the effect of increasing the non-prosecution rates of the bottom 90\% of ADAs to match that of the top 10\%. For example, ADH23 discuss how in 2019, a new district attorney increased non-prosecution rates by issuing a memo that instituted a presumption of non-prosecution for certain misdemeanor cases. Our results approximate the impact of such a policy to the extent that its impact on release rates was to make the bottom 90\% of ADAs match the top 10\%.

\paragraph{Data.} The data comes from the Suffolk County District Attorney’s
Office in Massachusetts. We focus on the main analysis main sample in ADH23,
which restricts to cases between 2004 and 2008, drops felony charges,
violent crimes, and prosecutors with fewer than 30 cases. This yields a sample
with $\sufN$ cases. ADH23 argue that conditional on court-by-time fixed effects,
the ADAs are as-good-as-randomly assigned to cases, and they control for these
fixed effects linearly in all their specifications. In their main specification,
they also consider an extended set of controls that include case and defendant
characteristics. Correspondingly, to estimate ADA-specific probabilities
$P^*(Y(d)=y, D(z,0)=d)$ we also linearly adjust for the court-by-time fixed
effects and case and defendant characteristics, as detailed in
\Cref{sec:impl-deta}. Since we have access to the microdata, our inference
accounts for this covariate adjustment. As in the NYC bail judge application, we
pool ADAs with fewer than 300 cases to mitigate finite-sample issues (we do this separately by whether leniency is above or below the 90th percentile).
\Cref{tab:quota_suffolk_disag} shows that our inference results remain the same
when we do not pool. This leaves us with $K=\sufKp$ distinct values for the
instrument.

\paragraph{Summary statistics.}

\Cref{fig:suffolk_raw} plots the ADA-specific non-prosecution rates against the
criminal complaint rates. Relative to \Cref{fig:nyc_raw}, there is much more
variation in the outcomes for ADAs with similar non-prosecution rates. An
implication of IV monotonicity is that ADAs with the same prosecution
rate must have the same outcomes, since they must be prosecuting the same set of
individuals; if the prosecution rates differ by $x$, the outcomes cannot differ
by more than $x$ times the range of the support for the outcome. The IV
monotonicity test of \textcite{fll23} tests precisely this implication. Thus,
the large outcome variability in \Cref{fig:nyc_raw} can be seen as giving visual
evidence against IV monotonicity, in line with the rejection of the
\textcite{fll23} test in several subsets of the data.

\begin{figure}[tp]
{\centering
{\small\input{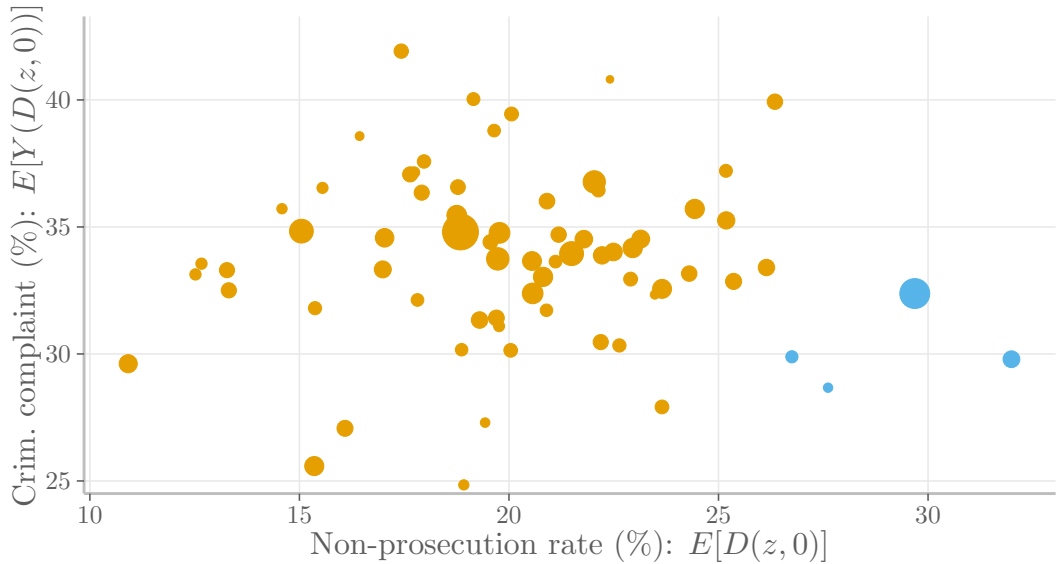}}
\caption{ADA-specific sample non-prosecution rates and criminal complaint rates,
  based on Suffolk County prosecutor data from
  \textcite{adh23}.}\label{fig:suffolk_raw}}
\floatfoot{\emph{Notes:} Both rates are adjusted for court-by-time fixed effects
  and case and defendant characteristics. Each dot corresponds to a single
  ADA\@; dot size scales with the number of cases handled by the ADA\@. ADAs
  with fewer than 300 cases are pooled into a single dot. The top decile of most
  lenient ADAs is plotted in blue, the bottom 90\% are plotted in orange.}
\end{figure}

\Cref{fig:suffolk_raw} also shows that the relationship between non-prosecution
and criminal complaints is negative: the IV estimate correspondingly equals
$\sufIVc$\%, albeit the estimate is an order of magnitude noisier than that in
ADH22, with standard error $\sufIVcse$ (this replicates Table III, col.~(4) in
ADH23). The overall non-prosecution rate in the data is $\sufED$\%, and the
overall criminal complaint rate is \sufEY\%.

\paragraph{Results for the quota policy.}

To benchmark our estimates, we again start out by computing the policy effect of
a reallocation policy that assigns all cases to the most lenient ADAs\@, and a
parametric benchmark that assumes homogeneous treatment effects, which allows for
extrapolation of the IV estimates. As shown in \Cref{tab:quota_suffolk}, these
imply that the quota policy would lead to a reduction in the probability of criminal complaints by about 1--3 percentage points. Notably, however, under the
reallocation policy the confidence interval only marginally excludes zero.

We then estimate the policy effect under a specification that only assumes
\ac{IV} validity and policy monotonicity. Like the NYC bail judge data, the data
is consistent with there being many instrument never-takers. The policy only
moves treatment by $\sufDeltaDc$ percentage points, so without further
assumptions, we cannot rule out that the individuals marginally released---the
policy compliers---are picked from the set of instrument never-takers in an
adversarial way. As a result, the estimate of the identified set (row 3 of
\Cref{tab:quota_suffolk}) implies trivial bounds for the policy complier
treatment effect.

\begin{table}[tp]
  \centering
\begin{threeparttable}
  \caption{Estimates for quota policy using Suffolk County prosecutor data from
  \textcite{adh23}.}\label{tab:quota_suffolk}
  \begin{tabular*}{0.8\linewidth}{@{\extracolsep{\fill}}lcc@{}}
    &Policy effect & {PC TE}\\
    Specification & (1) & (2)\\
    \midrule
    Reallocation & $-1.7$ & $-19.0$ \\
 & $(-3.3, -0.1)$ & $(-36.7, -1.2)$ \\
Constant treatment & $-2.6$ & $-28.8$ \\
 & $(-4.4, -0.8)$ & $(-48.5, -9.1)$ \\
Valid IV only & $[-9.0, 9.0]$ & $[-100.0, 100.0]$ \\
 & $(-9.3, 9.3)$ & $(-100.0, 100.0)$ \\
$\overline{OD}=0.019$ & $[-3.4, 0.0]$ & $[-37.9, 0.0]$ \\
 & $(-5.3, 1.9)$ & $(-61.2, 21.9)$ \\
\\[-0.9em]
    \bottomrule
  \end{tabular*}
  \begin{tablenotes}
    \begin{footnotesize}
    \item \emph{Notes}: Col.~(1) reports estimates of the policy effect
      $E[Y(D(Z,1))-Y]$, expressed in percentage points, for a quota policy that
      asks bottom 90\% of ADAs to match the non-prosecution rate of the most
      lenient decile (which equals $\sufQuotac\%$). Col.~(2) reports the
      treatment effect for policy compliers implied by the policy effect
      estimates. 95\% confidence intervals in parentheses. For specifications
      leading to set identification, identified set estimates reported in
      brackets.
  \end{footnotesize}
\end{tablenotes}
\end{threeparttable}
\end{table}

We next consider an additional assumption that can help us tighten these bounds. In contrast to our previous application, imposing bounds on the outcomes for policy compliers as in
\cref{eq:pc_bound} is hard to motivate in this setting, since prosecution decisions are based in large part on the strength of evidence
against the defendant, which is not directly related to their recidivism risk.
Likewise, we do not have external data from panels of prosecutors to help calibrate disagreement probabilities across ADAs.

We therefore pursue a simple approach that directly restricts the counterfactual
outcome disparity between ADAs in the bottom 90\% and those in the top decile as
in \cref{eq:od} for some bound $\overbar{\OD}$. Imposing $\overbar{\OD}=0$
amounts to assuming that the policy effect matches that under the reallocation
policy. As a simple benchmark, we calibrate $\overbar{\OD}$ to the status quo
outcome disparity between the two groups of ADAs. Since the criminal complaint
rate for those in the top decile is $1.9$ percentage points lower than for those
in the bottom 90\%, we set $\overbar{\OD}=\sufODc$. This is motivated by the
informal discussion in ADH23 who consider a counterfactual where the stricter
ADAs are asked to act ``more like the most lenient ADAs.'' If the policy
directive asks them to act this way, one can reasonably expect the outcome
disparity to be lower under the counterfactual, so that this value of
$\overbar{\OD}$ can be interpreted as a conservative bound. As
\Cref{tab:quota_suffolk} shows, imposing this bound substantially tightens the
estimate of the identified set. As column (2) indicates, the implied treatment
effect for policy compliers lies between $0$ and about 130\% of the IV estimate.

There is an alternative view of this bound that is useful. In particular,
setting $\overbar{\OD}$ to the status quo disparity is the maximal value of
$\overbar{\OD}$ that one can allow for while still guaranteeing identification
of the sign of the policy effect.\footnote{Since ADAs above the quota are
  assumed not to change their behavior, the policy effect is proportional to the difference between the counterfactual and status quo outcome disparities, $\theta-E[Y]=n_{\mathcal{Q}}/n\cdot (E[Y(D(Z_{\mathcal{Q}}, 1))-
  Y(D(Z_{\mathcal{Q}^c}, 1))]-(E[Y\mid Z\in\mathcal{Q}]-E[Y\mid
  Z\in\mathcal{Q}^c]))$, with $n_{\mathcal{Q}}$ denoting the number of cases
  assigned to the ADAs who are subject to the quota. Since the status quo
  outcome disparity is negative, assuming that the counterfactual outcome
  disparity is no larger in magnitude than the status quo disparity implies that the
  policy effect is weakly negative.} Using \cref{eq:od_dp_relation}, we can thus
ask if the implied disagreement probability is reasonable: if so, this suggests
that identification of the sign of the treatment effect is possible under
reasonable restrictions on the disagreement probability. If we make no
restrictions on treatment effect heterogeneity, the implied value of
$\overbar{DP}$ is \sufDP, which corresponds to a signal correlation equal to
$\rho=\sufrho$. This is substantially higher than the corresponding value
calibrated from \textcite{sigstad_monotonicity_2023} ($0.989$), implying that
prosecutors in Massachusetts would have to be much more strongly in agreement
than the judges in Sigstad's data. While it is possible that agreement rates
differ across contexts, this strikes us as a stringent requirement, especially
given the rejection of IV monotonicity, which implies disagreement under the
status quo. We can, however, learn about the sign of the policy impact under
weaker restrictions on disagreement if we also impose some restrictions on
treatment effect heterogeneity. If we impose that $\Delta_{TE}$ is at most twice
the IV estimate, then we only require $\overbar{DP} = \sufDPIV$, which
corresponds to $\rho = \sufrhoIV$ in the Gaussian signal model. This is lower
than the value in Sigstad, and thus seems more reasonable.

\printbibliography%

\newpage

\begin{appendices}

\renewcommand{\figurename}{Appendix Figure}
\setcounter{figure}{0}
\renewcommand{\tablename}{Appendix Table}
\setcounter{table}{0}

\section{Proofs}\label{sec:proofs}

\subsection{Proof of Proposition~\ref{prop:marginals_one_z}}\label{sec:proof_marginals_one_z}

\Cref{theorem:general_statement} below gives a more general version of
\Cref{prop:marginals_one_z} that does not restrict the distribution of the
potential outcomes. Stating this result requires setting up additional notation.
The model primitives $(Y(0), Y(1), D(\cdot, \cdot), Z)$ take values in the space $\mathbb{R}^{2+2K+1}$. We endow this space with its standard Borel $\sigma$-algebra. We consider general Borel probability measures for the model primitives with support $\mathcal{Y}^2 \times \{0,1\}^{2K} \times \mathcal{Z}$, where $\mathcal{Y}\subseteq\mathbb{R}$. Proposition 10.2.8 in \textcite{dudley2018real}
implies that any joint probability distribution over the model's primitives has well-defined conditional distributions, for any subset of conditioning variables. We use this fact repeatedly in what follows.

For each $z \in\mathcal{Z}$, let $\Pi^*_{z}$ denote the marginal distribution of
the random vector $(Y(0), \allowbreak Y(1), \allowbreak D(z,0), D(z,1))$ under
$P^*$. It is without loss of generality to assume these marginal distributions
are all dominated by some probability measure, $\mu$ (for example, set
$\mu := (1/K) \sum_{k=1}^{K} \Pi^*_{k}$). We denote the densities of
$\Pi^*_{z}$ with respect to $\mu$ by $\pi^*_{z}$. We refer to the collection of
densities $\{\pi^*_{z}\}_{z \in\mathcal{Z}}$ as the \emph{marginal
  $z$-densities} of $P^*$.

The dominating measure $\mu$ can be thought of as the law of some random vector
$(Y_{0}, Y_{1}, \allowbreak D_{0}, D_{1})$. Let
$\mu_{D_{0}, D_{1}\mid Y_{0}=y_0, Y_{1}=y_1}$ denote the conditional distribution
of $\mu$ given $(Y_0, Y_1)=(y_0, y_1)$, which we abbreviate to
$\mu_{D_{0}, D_{1}\mid y_0,y_1}$ whenever it doesn't cause confusion. We denote
the marginal distribution of $(Y_0,Y_1)$ under the probability measure $\mu$ by
$\mu_{Y_{0}, Y_{1}}$. We follow the same convention to denote other conditional
and marginal distributions.

The observable data given $Z=z$, $(Y(D(z,0)), D(z,0))$, take values in the sample space
$\mathcal{Y} \times \{0,1\}$. Thinking of $\mu$ as a probability, the implied
measure on this space is given by
$\tilde{\mu}(F\times G)= \1{1\in G}\mu_{Y_{1}\mid D_{0}=1}(F)\mu_{D_{0}}(\{1\})+
\1{0\in G}\mu_{Y_{0}\mid D_{0}=0}(F)\mu_{D_{0}}(\{0\})$. Let $P_{Y, D\mid z}$
denote the conditional distribution of $P$ given $Z=z$.

\begin{prop}\label{theorem:general_statement}
  There exists
  $P^* \in \mathcal{P}^*_{I}(P;\mathcal{P}^{*}_{\textnormal{valid}})$ with
  marginal $z$-densities $\{\pi_{z}\}_{z \in \mathcal{Z}}$ (with respect to some probability measure $\mu$) if and only if
  $\{\pi_{z}\}_{z \in\mathcal{Z}}$ satisfy the following conditions:
  \begin{enumerate}
  \item\label{item:match_data_general} They match the observable data: for every
    $z\in\mathcal{Z}$, $P_{Y, D\mid z}$ is absolutely continuous with respect to
    $\tilde{\mu}$ with density $p_{Y, D\mid z}$ given $\tilde{\mu}$-a.e.\ by
    \begin{align*}
      p_{Y, D\mid z}(y,1) &=\int_{\mathcal{Y} \times \{0,1\}} \pi_{z} (y_0, y,
      1, d_1)\, d\mu_{Y_{0}, D_{1}\mid Y_{1}=y, D_{0}=1}, \\
      p_{Y, D\mid z}(y,0) &=\int_{\mathcal{Y} \times \{0,1\}} \pi_{z} (y, y_1, 0, d_1)\,
                 d\mu_{Y_{1}, D_{1}\mid Y_{0}=y, D_{0}=0}.
    \end{align*}
  \item\label{item:same_marginals_general} They imply the same distribution for
    $(Y(0), Y(1))$: for all $z, z' \in \mathcal{Z}$, the equality
      \begin{equation*}
        \int_{\{0,1\}^2} \pi_{z}(y_0, y_1, d_0, d_1)
        \, d\mu_{D_{0}, D_{1}\mid y_{0}, y_{1}}
        =\int_{\{0,1\}^2} \pi_{z'}(y_0, y_1, d_0, d_1)
        \, d\mu_{D_{0}, D_{1}\mid y_{0}, y_{1}}
      \end{equation*}
      holds $\mu_{Y_{0}, Y_{1}}$-a.e.
    \item\label{item:valid_pmf_general} They are valid density functions: for
      all $z \in\mathcal{Z}$, $\pi_{z}(y_0, y_1, d_0, d_1) \geq 0$ ($\mu$-a.e.)
      and
      $\int_{\mathcal{Y}^2 \times \{0,1\}^2} \pi_{z}(y_{0}, y_{1}, d_{0}, d_{1})
      \, d\mu= 1$.
\end{enumerate}
\end{prop}

\begin{proof}
  We first show the $\implies$ part of the statement. Namely, if there exists
  $P^* \in \mathcal{P}^*_{I}(P; \mathcal{P}^*)$ with marginal $z$-densities
  $\{\pi_{z}\}_{z \in \mathcal{Z}}$, then
  conditions~\ref{item:match_data_general}--\ref{item:valid_pmf_general} hold. This is conceptually straightforward, and most of the work just involves translating statements about probabilities to statements about densities.

  To establish condition~\ref{item:match_data_general}, observe that by definition of the density $p_{Y, D\mid z}$, for any measurable $F \subseteq \mathcal{Y}$,
  \begin{equation*}
      P_{Y, D\mid z}(F\times \{1\}) = \int_{F \times \{1\}} p_{Y, D \mid z}(y, 1) \, d \tilde{\mu}.
  \end{equation*}
  Further, observe that
  \begin{equation*}
    \begin{split}
      P_{Y, D\mid z}(F\times \{1\})
      &=P^{*}(Y(1)\in F, D(z,0)=1)\\
      &
        =\int_{\mathcal{Y} \times F \times \{0,1\} \times \{1\}}
        \pi_{z} (y_0, y_1, d_0, d_1)\, d \mu\\
      &=\int_{\{1\}}\int_{F}\int_{\mathcal{Y} \times \{0,1\}}
        \pi_{z} (y_0, y_1,
        d_0, d_1)\, d\mu_{Y_{0}, D_{1}\mid y_{1}, d_{0}}
        d \mu_{Y_{1}\mid D_{0}=d_{0}}d\mu_{D_{0}}\\
      &=\int_{F}\left[\int_{\mathcal{Y} \times \{0,1\}}
        \pi_{z} (y_0, y_1,
        1, d_1)\, d\mu_{Y_{0}, D_{1}\mid Y_{1}=y_{1}, D_{0}=1}\right]\,
        d \mu_{Y_{1}\mid D_{0}=1}\cdot \mu_{D_{0}}(\{1\})\\
      &= \int_{F \times \{1\}} \left[\int_{\mathcal{Y} \times \{0,1\}}
        \pi_{z} (y_0, y_1,
        1, d_1)\, d\mu_{Y_{0}, D_{1}\mid Y_{1}=y_{1}, D_{0}=1}\right] \, d\tilde{\mu}_{Y, D}
    \end{split}
  \end{equation*}
  where the first line uses the fact that since
  $P^* \in \mathcal{P}^*_{I}(P; \mathcal{P}^*)$, it generates $P$ and satisfies
  \cref{eq:valid_iv}, the second line uses the definition of the
  density $\pi_{z}$, the third line uses law of iterated expectations (e.g.~Part
  II of Theorem 10.2.1 in \textcite{dudley2018real}), the fourth line follows by
  algebraic manipulation, and last line uses the definition of $\tilde{\mu}$. Since the last two displays hold for all measurable $F$, it follows that $p_{Y, D \mid z}(y,1) = \int_{\mathcal{Y} \times \{0,1\}}
        \pi_{z} (y_0, y_1,
        1, d_1)\, d\mu_{Y_{0}, D_{1}\mid Y_{1}=y_{1}, D_{0}=1}$ ($\tilde{\mu}$-a.e.), which gives the first expression in condition~\ref{item:match_data_general}. The second expression in condition~\ref{item:match_data_general} can be verified by using an analogous argument to show that
   \begin{equation*}
   P_{Y, D\mid z}(F\times \{0\}) = \int_{F \times \{0\}} p_{Y, D \mid z}(y, 0) \, d \tilde{\mu} = \int_{\mathcal{Y} \times \{0,1\}} \pi_{z} (y_0, y_1,
        0, d_1)\, d\mu_{Y_{1}, D_{1}\mid Y_{0}=y_{0}, D_{0}=0}.
        \end{equation*}

  To establish condition~\ref{item:same_marginals_general}, fix any
  $z\in\mathcal{Z}$. Then for any measurable sets
  $F_{0}, F_{1} \subseteq \mathcal{Y}$:
  \begin{equation*}
    \begin{split}
      P^*(Y(0) \in F_0, Y(1) \in F_{1})
      &=\int_{F_0 \times F_1 \times \{0,1\}^2} \pi_{z}(y_0, y_1, d_0, d_1) \, d \mu\\
      &= \int_{F_0 \times F_1}\int_{\{0,1\}^2} \pi_{z}(y_0, y_1, d_0, d_1)
        \, d\mu_{D_{0}, D_{1}\mid y_{0}, y_{1}}\, d\mu_{Y_{0}, Y_{1}},
    \end{split}
  \end{equation*}
  where the first equality follows by definition of the density $\pi_{z}$, and
  the second by the law of iterated expectations. Thus, for any measurable
  $F_0, F_1 \subseteq \mathcal{Y}$ and any pair $z, z'\in\mathcal{Z}$,
  \begin{multline*}
    \int_{F_0 \times F_1}\int_{\{0,1\}^2} \pi_{z}(y_0, y_1, d_0, d_1)
    \, d\mu_{D_{0}, D_{1}\mid y_{0}, y_{1}}\, d\mu_{Y_{0}, Y_{1}}\\
    =\int_{F_0 \times F_1}\int_{\{0,1\}^2} \pi_{z'}(y_0, y_1, d_0, d_1)
    \, d\mu_{D_{0}, D_{1}\mid y_{0}, y_{1}}\, d\mu_{Y_{0}, Y_{1}},
  \end{multline*}
  which implies condition~\ref{item:same_marginals_general}. Finally, condition~\ref{item:valid_pmf_general} is immediate from the definition of a marginal
  $z$-density.

  \paragraph{\mbox{}} We now show the $\impliedby$ part of the statement.
  That is, given $\{\pi_{z}\}_{z \in \mathcal{Z}}$ that satisfy
  conditions~\ref{item:match_data_general}--\ref{item:valid_pmf_general}, we
  need to show that we can build a distribution $P^*$ with marginal
  $z$-densities $\{\pi_{z}\}_{z \in \mathcal{Z}}$ that generates $P$. The key step is the observation that since the marginals imply the same distribution of $(Y(0), Y(1))$, to couple them together to form a joint distribution $P^*$, we just need to construct a joint distribution for $D(\cdot, \cdot)$ given $Y(0), Y(1)$. To this end, we show this can be done by assuming that conditional on $(Y(0), Y(1))$, the treatment decisions $D(z,0), D(z,1)$ can be assumed to be independent across judges $z$, with the distribution of $D(z,0), D(z,1)$ set to match that implied by the marginal $\pi_z$.

  Formally, to construct $P^{*}$, for any measurable $F_0,F_1 \subseteq \mathcal{Y}$,
  $G =(G_{0}^1 \times G_{1}^{1} \dotsb G^{K}_{0} \times G^{K}_{1})\subseteq
  \{0,1\}^{2K}$, and $H \subseteq \mathcal{Z}$ we set
  \begin{multline*}
    P^* \left(Y(0) \in F_0, Y(1) \in F_1, D(\cdot, \cdot) \in G, Z \in H \right)\\
    :=P^{*} \left(Y(0) \in F_0, Y(1) \in F_1, D(\cdot, \cdot) \in G \right)
    \cdot
    P(Z \in H),
  \end{multline*}
  set the conditional distribution of $P^{*}$ given $(Y(0), Y(1))=(y_{0}, y_{1})$ to
  \begin{equation}\label{eq:pstar_given_y0y1}
    P^{*}_{D(\cdot, \cdot)\mid y_{0}, y_{1}}:=
    \Pi_{D(1, \cdot)\mid y_0, y_1} \left(G_0^1 \times G^1_1 \right) \dotsb \Pi_{D(K, \cdot)\mid y_0, y_1} \left(G_0^K \times G^K_1 \right),
  \end{equation}
  where $\Pi_{D(z, \cdot)\mid y_0, y_1}$ is the conditional distribution of
  $(D(z, 0), D(z, 1))$ given $(y_0, y_1)$ implied by the marginal $z$-density
  $\pi_{z}$, and set the marginal
  distribution of $P^{*}$ over $(Y(0), Y(1))$ to
  \begin{equation}\label{eq:pstar_marginal_y0y1}
    P^{*}(Y(0)\in F_{0}, Y(1)\in F_{1}):=\int_{F_0 \times F_1}\int_{\{0,1\}^2}
    \pi_{z}(y_0, y_1, d_0, d_1)
    \, d\mu_{D_{0}, D_{1}\mid y_{0}, y_{1}}\, d\mu_{Y_{0}, Y_{1}}.
  \end{equation}
  By conditions~\ref{item:same_marginals_general}
  and~\ref{item:valid_pmf_general}, this marginal distribution is well-defined
  and doesn't depend on $z$.

  Observe that by construction, $P^{*}$
  satisfies~\cref{eq:valid_iv}. It thus remains to show that
  $P^{*}$ generates $P$. Given that the marginals over $Z$ match by
  construction, it suffices to show that for any $z\in\mathcal{Z}$, and
  $F\subseteq \mathcal{Y}$,
  \begin{equation} \label{eqn:aux_condition3_a}
    P_{Y, D\mid z}(Y \in F, D=1) =
    P^{*} \left(Y(1) \in F, D(z, 0)=1
    \right)
  \end{equation}
  and
  \begin{equation} \label{eqn:aux_condition3_b}
    P_{Y, D\mid z}(Y \in F, D=0) =
    P^{*} \left(Y(0) \in F, D(z, 0)=0
    \right).
  \end{equation}
  Let $\Pi_z$ denote the distribution implied by the densities $\pi_z$.
  \Cref{eqn:aux_condition3_a} follows from the following set of equalities:
  \begin{equation*}
    \begin{split}
      P_{Y, D\mid z}(Y \in F, D=1)
      &=
        \int_{F} \int_{\mathcal{Y} \times \{0,1\}} \pi_{z} (y_0, y,
        1, d_1)\, d\mu_{Y_{0}, D_{1}\mid Y_{1}=y, D_{0}=1}\, d\mu_{Y_{1}\mid D_{0}=1}
        \mu_{D_{0}}(\{1\})\\
      & =\int_{\mathcal{Y}\times F\times \{1\}\times \{0,1\}}
        \pi_{z} (y_0, y_{1}, d_{0}, d_1)\, d\mu=\Pi_{z}
        (\mathcal{Y}\times F\times  \{1\}\times \{0,1\})\\
      &=\int_{\mathcal{Y}\times F}\int_{\{1\}\times \{0,1\}}
        \, d\Pi_{D(z, \cdot)\mid y_0,y_1} \, d P^{*}_{Y(0), Y(1)}\\
      &=\int_{\mathcal{Y}\times F}\int_{\{1\}\times \{0,1\}}
        \, d P^{*}_{D(z,0), D(z,1)\mid y_{0}, y_{1}}
        \, d P^{*}_{Y(0), Y(1)}\\
      &=P^{*}(Y(1)\in F, D(z, 0)=1).
    \end{split}
  \end{equation*}
  where the first equality uses the fact that $P_{Y, D\mid z}$ has density
  $p_{Y, D\mid z}$ and condition~\ref{item:match_data_general}, the second
  equality uses iterated expectations, the third equality uses iterated
  expectations and the fact that by definition of $P^{*}$
  in~\cref{eq:pstar_marginal_y0y1}, its marginal distribution over
  $(Y(0), Y(1))$ matches $\Pi_{z}$, the fourth line uses the fact that by
  \cref{eq:pstar_given_y0y1}, the conditional distribution
  $P^{*}_{D(z,0), D(z,1)\mid y_{0}, y_{1}}$ matches $\Pi_{D(z, \cdot)\mid y_0,y_1} $, and
  the last line follows by iterated expectations. Analogous arguments
  establish~\cref{eqn:aux_condition3_b}.
\end{proof}

\subsection{Proof of \texorpdfstring{\Cref{cor:optmize_over_pi}}{Corollary~\ref{cor:optmize_over_pi}}}

The fact that the identified set is an interval follows by convexity of the
optimization problem. To show that \cref{eq:theta_value} gives the upper
endpoint of the interval, note that, by definition of the supremum, there exists
a sequence of distributions $P_{n}^{*}\in\mathcal{P}^{*}$
such that
$\sup_{P^{*}\in\mathcal{P}^{*}}E_{P^{*}}[Y(D(Z,
1))]=\lim_{n\to\infty}E_{P_{n}^{*}}[Y(D(Z, 1))]$. Let $v_{0}$ be the value
function of the optimization problem in~\eqref{eq:theta_value}. Since
$P_{n}^{*}\in\mathcal{P}^{*}$, the marginals of $P_{n}^{*}$, call them $\{\pi^{n}_{z}(\cdot)\}_{z}$, are feasible in the optimization
problem~\eqref{eq:theta_value} by \Cref{prop:marginals_one_z}. It follows that
$\sup_{P^{*}\in\mathcal{P}^{*}} E_{P^{*}}[Y(D(Z, 1))]\leq v_{0}$. To show that
$\sup_{P^{*}\in\mathcal{P}^{*}} E_{P^{*}}[Y(D(Z, 1))]\geq v_{0}$, suppose that
$\{\widetilde{\pi}^{n}_{z}(\cdot)\}_{z}$ is a sequence in $\mathcal{R}^{*}$ such
that the value of~\eqref{eq:theta_value} is given by the limit
\begin{equation*}
  \lim_{n\to\infty}\sum_{z \in \mathcal{Z}} P(Z=z) \sum_{y_{0}, y_{1}, d_{0}, d_{1} \in
    \mathcal{Y}^2 \times \{0,1\}^2} \left(d_1 y_1 + (1-d_1)y_0 \right)
  \cdot\widetilde{\pi}^{n}_{z}(y_0, y_1, d_0, d_1).
\end{equation*}
Then, since $\{\widetilde{\pi}^{n}_{z}(\cdot)\}_{z}$ satisfy the
constraints~\ref{item:match_data}--\ref{item:valid_pmf}
in~\Cref{prop:marginals_one_z}, there exists some sequence
$\tilde{P}_{n}^{*}\in\mathcal{P}^{*}_{\textnormal{valid}}$ with marginals
$\{\widetilde{\pi}^{n}_{z}(\cdot)\}_{z}$. Furthermore, since
$\{\widetilde{\pi}^{n}_{z}(\cdot)\}_{z}\in\mathcal{R}^{*}$, it follows that
$\tilde{P}_{n}^{*}\in\mathcal{P}^{*}$. Thus,
$\sup_{P^{*}\in\mathcal{P}^{*}} E_{P^{*}}[Y(D(Z, 1))]\geq v_{0}$.

\subsection{Proof of \texorpdfstring{\Cref{prop:monotonicity_is_not_helpful}}{Proposition~\ref{prop:monotonicity_is_not_helpful}}}

Let $\alpha_0=P(D=1)$ denote the status quo release rate. If $\alpha=\alpha_0$,
then $\theta$ is point-identified by $E_{P}[Y]$; assume therefore that
$\alpha>\alpha_0$. Since $Y(0)=0$ and $D(Z,1)\geq D(Z,0)$ by policy
monotonicity, it follows by iterated expectations that
\begin{equation*}
  \begin{split}
    E_{P^*}[Y(D(Z,1))]
    &=E_{P^*}[Y(1) \mid D(Z,1) > D(Z,0)] (\alpha-\alpha_{0}) + E_{P^*}[Y(1) \mid
      D(Z,0)=1]\alpha_{0}\\
    & =E_{P^*}[Y(1) \mid D(Z,1) > D(Z,0)] (\alpha-\alpha_{0}) + E_{P}[Y \mid D=1]\alpha_{0},
  \end{split}
\end{equation*}
where we use $P^{*}(D(Z,1)>D(Z,0))=\alpha-\alpha_{0}$, and
$P^{*}(D(Z,0))=\alpha_{0}$. Apart from the average treated outcome for policy
compliers, $E_{P^*}[Y(1) \mid D(Z,1) > D(Z,0)]$, all other quantities are known.
It therefore suffices to show that the identified set for
$E_{P^*}[Y(1) \mid D(Z,1) > D(Z,0)]$ does not depend on IV monotonicity.

\Cref{lemma:policy_complier_bounds} in \Cref{sec:auxiliary-lemmas} derives lower
and upper bounds on the cdf $F_{P^{*}}$ of $Y(1) \mid D(Z,1) > D(Z,0)$ over
$P^{*}\in\mathcal{P}_{I}^{*}(P, \mathcal{P}^{*}_{\alpha})$ of the form
$F^{lb} \leq F_{P^{*}} \leq F^{ub}$.
\Cref{lemma:upper_lower_bounds_EY_D_1_tight} shows that if the set
$\mathcal{P}_{I}^{*}(P, \mathcal{P}^{*}_{Mon, \alpha})$ is not empty, there
exist probabilities
$P^{ub}, P^{lb}\in \mathcal{P}_{I}^{*}(P, \mathcal{P}^{*}_{Mon, \alpha})$ that
generate the cdfs $F^{lb}$ and $F^{ub}$. Since first-order stochastic dominance,
$P(Y>t)\leq Q(Y>t)$, implies $E_{P}[Y]\leq E_{Q}[Y]$
\parencite[e.g.][Proposition A.2, p.~694]{MaOl11}, it follows by convexity of
$\mathcal{P}_{I}^{*}(P, \mathcal{P}^{*}_{\alpha})$ that whenever this set is not
empty, the identified set for $\theta$ is given by an interval, with endpoints
given by the expected value of the distributions $F^{ub}$ and $F^{lb}$,
respectively, whether or not IV monotonicity is imposed.

\subsection{Proof of Proposition~\ref{prop:monotonicity_is_not_helpful_strong_encouragement}}

Let $\alpha_{\max}=E[D(z_{\max}, 0)]$, where, by definition,
$z_{\max}\in\argmax_{z\in\mathcal{Z}} E[D(z,0)]$. By iterated expectations, for
any $P\in \mathcal{P}^*_{I}(P; \mathcal{P}_{\alpha}^*)$,
\begin{multline*}
  E_{P^{*}}[Y(D(Z, 1))]\\
  =\alpha_{\max}E_{P^{*}}[Y(D(Z,1))\mid D(z_{\max},0)=1]
  +(1-\alpha_{\max})E_{P^{*}}[Y(D(Z,1))\mid D(z_{\max},0)=0]\\
  =\alpha_{\max}E_{P^{*}}[Y(1)\mid D(z_{\max},0)=1]
  +(1-\alpha_{\max})E_{P^{*}}[Y(D(Z,1))\mid D(z_{\max},0)=0],
\end{multline*}
where the second inequality follows by the definition of the sufficiently strong
encouragement condition (iii). Since the first term is point identified,
$E_{P^{*}}[Y(1)\mid D(z_{\max},0)=1]=E_{P}[Y\mid D=1,Z=z_{\max}]$, if
$\alpha_{\max}=1$, then the proposition holds trivially. Thus, we focus on the
case $\alpha_{\max}<1$. We will show that the identified set for the conditional
distribution $Y(1), Y(0), D(Z,1)\mid D(z_{\max},0)=0$ doesn't depend on whether IV
monotonicity (\Cref{assumption:instrument_monotonicity}) is imposed. Since $E_{P^{*}}[Y(D(Z,1))\mid D(z_{\max},0)=0]$ is a
functional of this distribution, the result will then follow.

The distribution of $D(Z,1) \mid D(z_{\max},0)=0$ is Bernoulli with parameter
$(\alpha-\alpha_{\max})/(1-\alpha_{\max}) \geq 0$, regardless of whether we impose
IV monotonicity. This follows from the fact that
\begin{equation*}
  \begin{split}
    \alpha& = P^*(D(Z,1)=1) \\
          & = P^*(D(Z,1)=1, D(z_{\max,0})=1) + P^*(D(Z,1)=1, D(z_{\max,0})=0)\\
          & = \alpha_{\max} + P^*(D(Z,1)=1, D(z_{\max,0})=0),
  \end{split}
\end{equation*}
where the last equality holds because $P^*$ satisfies the strong encouragement
condition.

Moreover, since the distribution of the observable data does not depend on $D(Z,1)$, any joint distribution for $Y(1), Y(0), D(Z,1) \mid D(z_{\max},0) = 0$ is in the identified set if the implied marginal for $Y(1), Y(0) \mid D(z_{\max},0) = 0$ is in the identified set and $P(D(Z,1) = 1 \mid D(z_{\max}) =0) $ is Bernoulli with parameter $(\alpha-\alpha_{\max})/(1-\alpha_{\max})$.

Thus, it is sufficient to show that the identified set for the conditional distribution $Y(1), Y(0) \mid D(z_{\max},0) = 0$ does not depend on IV monotonicity. To this end, note that $Y(0) \mid D(z_{\max},0) = 0$ is identified by the distribution of $Y \mid D=0, Z = z_{\max}$. Thus, without imposing IV monotonicity, the identified set for $(Y(1), Y(0)) \mid D(z_{\max},0) =0$ is a weak subset of the possible joint distributions that match the identified marginal for $Y(0) \mid D(z_{\max},0)=0$ and the constraint that $Y(1) \in \mathcal{Y}$. However, any joint distribution in this set is achievable under IV monotonicity, since IV monotonicity implies that the observed data do not depend on $Y(1) \mid D(z_{\max},0) = 0$, and thus any choice for $P(Y(1) \mid Y(0), D(z_{\max},0)=0)$ is consistent with the observable data.

\subsection{Proof of Proposition~\ref{prop:coupling_result_with_disagreement}}

We first show the $\implies$ direction that if a distribution
$P^* \in \mathcal{P}^*_{I}(P;\mathcal{P}^{*}_{\textnormal{DB}})$ generates the
collection of marginal probability mass functions
$\{\pi_{z}(\cdot) \}_{z \in \mathcal{Z}}$, then these marginals satisfy
conditions~\ref{item:match_data}--\ref{item:disagreement-bound}. Since
$\mathcal{P}^*_{\textrm{DB}}\subseteq\mathcal{P}^*_{\textnormal{valid}}$,
\Cref{prop:marginals_one_z} implies that the collection
$\{\pi_{z}(\cdot) \}_{z \in \mathcal{Z}}$ satisfies
conditions~\ref{item:match_data}--\ref{item:valid_pmf}. Hence, it only remains
to verify condition~\ref{item:disagreement-bound}. To this end, note
that~\eqref{eqn:disagreement-bounds} is equivalent to
  \begin{equation*}
    \begin{split}
      P^*(D(z, a)=1, D(z', a')=1)
      &\geq
        (1-\delta_{z, z', a, a'})P^*(D(z', a')=1)\\
      & =(1-\delta_{z, z', a, a'})\pi(D(z', a')=1).
    \end{split}
  \end{equation*}
  By the law of total probability, the left-hand side equals
  \begin{multline*}
    \sum_{y_1,y_0} P^*(Y(0)=y_0,Y(1) = y_1, D(z, a) = 1, D(z', a') = 1)\leq \\
    \sum_{y_1,y_0} \min\{
    \pi(y_{0}, y_{1}, D(z, a)=1), \pi(y_{0}, y_{1}, D(z', a')=1) \},
  \end{multline*}
  where the inequality is obtained from the fact that $P(A \cap B) \leq \min\{P(A),P(B)\}$ for events $A,B$, and the identity
  $\pi(y_{0}, y_{1}, D(z, a)=1)=P^*(Y(0)=y_0, Y(1) = y_1, D(z, a) = 1)$.
  Combining the preceding two displays then yields the
  condition~\ref{item:disagreement-bound}.

  Now we show the $\impliedby$ direction that if a collection of marginals
  $\{\pi_{z}(\cdot)\}_{z \in \mathcal{Z}}$ satisfies
  conditions~\ref{item:match_data}--\ref{item:disagreement-bound} and is consistent with \Cref{assumption:policy_monotonicity}, then there
  exists a distribution
  $P^* \in \mathcal{P}^*_{I}(P;\mathcal{P}^{*}_{\textnormal{DB}})$ that
  generates $\{\pi_{z}(\cdot) \}_{z \in \mathcal{Z}}$. At a high-level, the specified $\{\pi_z\}_z$ pin down the marginal distributions of $\{Y(\cdot), D(z, \cdot) \}_z$, so what remains is to specify a coupling of these marginals under $P^*$ to match the disagreement bound. The key idea is to construct $D(\cdot, \cdot)$ via a latent threshold crossing model conditional on the potential outcomes, so that $D(z, a) = \1{V_{y_1,y_0} \leq \alpha_z(a, y_1,y_0)}$, where $V_{y_1,y_0}$ is a common ranking conditional on $Y(1)=y_1,Y(0)=y_0$. In other words, the judges under this DGP agree on the ranking of defendants within groups of people with the same potential outcomes, but possibly disagree on the potential outcome-specific release cutoffs.

  Formally, to construct $P^{*}$, we set
  \begin{multline*}
    P^* \left(Y(0)=y_{0}, Y(1)=y_{1},
      \{D(z,0)=d_{z0},
      D(z,1)=d_{z1}\}_{z\in \mathcal{Z}}, Z=z \right)\\
    :=P^{*}\left(Y(0)=y_{0}, Y(1)=y_{1},
      \{D(z,0)=d_{z0},
      D(z,1)=d_{z1}\}_{z\in \mathcal{Z}} \right) \cdot
    P \left(Z =z \right),
  \end{multline*}
  set the marginal distribution of $P^{*}$ over $(Y(0), Y(1))$ to equal the
  distribution of $(Y(0), Y(1))$ implied by the marginals $\pi_{z}$,
  \begin{equation*}
    P^{*}(Y(0)=y_{0}, Y(1)=y_{1}):=\pi(y_{0}, y_{1}), \quad
    \pi(y_{0}, y_{1}):=
    \sum_{d_{0}, d_{1}\in\{0,1\}^{2}}\pi_{z}(y_{0},
    y_{1}, d_{0}, d_{1}),
  \end{equation*}
  (by condition~\ref{item:same_marginals}, $\pi(y_{0}, y_{1})$ doesn't depend on
  $z$), and finally, to define the conditional distribution of
  $\{D(z,0), D(z,1)\}_{z\in \mathcal{Z}}$ given $(Y(0), Y(1))=(y_{0}, y_{1})$
  we set
\begin{equation}\label{eqn: construction for dza}
  D(z, a) = \1{V_{y_0,y_1} \leq \alpha_{z}(a, y_{0}, y_{1})}, \quad
\alpha_{z}(a, y_{0}, y_{1})=\frac{\pi_{z}(y_0,y_1,1,1)}{\pi(y_0,y_1)} + a \frac{\pi_z(y_0,y_1,0,1)}{\pi(y_0,y_1)},
\end{equation}
where $V_{y_{0}, y_{1}}$ is uniform on $[0,1]$ conditional on $Y(1) = y_1, Y(0) = y_0$.

This construction implies that \cref{eq:valid_iv} holds. Further, since $\alpha_z(1,y_0,y_1) \geq \alpha_z(0,y_0,y_1)$ by construction, this implies that $P^*(D(z,0)=1, D(z,1)=0 \mid Y(0)=y_0, Y(1)=y_1)= 0$, so that
\Cref{assumption:policy_monotonicity} holds. The construction also implies that
\begin{align*}
  P^*(D(z,0)=0, D(z,1)=1 \mid Y(0)=y_0, Y(1)=y_1)& = \frac{\pi_z(y_0,y_1,0,1)}{\pi(y_0,y_1)}, \\
  P^*(D(z,0)=1, D(z,1)=1 \mid Y(0)=y_0, Y(1)=y_1) &= \frac{\pi_z(y_0,y_1,1,1)}{\pi(y_0,y_1)},
\end{align*}
so that $P^{*}$ generates the marginals $\{\pi_{z}\}$. Moreover,
condition~\ref{item:match_data} implies that $P^*$ generates $P$. It remains to
show that $P^{*}$ satisfies the disagreement bounds
in~\eqref{eqn:disagreement-bounds}. It follows from the definition of $D(z, a)$
in \cref{eqn: construction for dza} that for any $z,z',a,a'$
\begin{equation*}
  P^{*}(D(z, a)=1, D(z', a')=1\mid Y(0)=y_{0}, Y(0)=y_{1})=
  \min\{\alpha_{z}(a, y_{0}, y_{1}), \alpha_{z'}(a', y_{0}, y_{1})\}.
\end{equation*}
Hence, unconditionally,
\begin{multline}\label{eq:dz_unconditional}
  P^{*}(D(z, a)=1, D(z', a')=1)=
  \sum_{y_{0}, y_{1}} \min\{\alpha_{z}(a, y_{0}, y_{1}),
  \alpha_{z'}(a', y_{0}, y_{1})\}\pi(y_{0}, y_{1})\\
  =\sum_{(y_0,y_1) \in \mathcal{Y}^2} \min\left\{\pi_z(y_0,y_1,1,1)+a\pi_z(y_0,y_1,0,1), \pi_{z'}(y_0,y_1,1,1)+a'\pi_{z'}(y_0,y_1,0,1) \right \}.
\end{multline}
We now claim that for any $(z, a) \in \mathcal{Z} \times \{0,1\}$,
\begin{equation}\label{eq:remaining_step}
  \pi_z(y_0,y_1,1,1)+a\pi_z(y_0,y_1,0,1) = \sum_{d_a=1, d_{1-a} \in \{0,1\}}
    \pi_z(y_{0}, y_{1}, d_{0}, d_{1}).
\end{equation}
If $a=1$, \cref{eq:remaining_step} is
immediate. If $a=0$, \cref{eq:remaining_step} may be written
$\pi_z(y_0,y_1,1,1)=\pi_z(y_0,y_1,1,0)+\pi_z(y_0, y_1, 1, 1)$. But, by assumption, the densities $\pi_z$ are consistent with \Cref{assumption:policy_monotonicity}, and thus
$\pi_z(y_0,y_1,1,0)=0$, and hence we see~\cref{eq:remaining_step} holds. Substituting \cref{eq:remaining_step} into \cref{eq:dz_unconditional}, we then obtain that
\begin{multline*}
  P^{*}(D(z, a)=1, D(z', a')=1)\\
  =\sum_{(y_0,y_1) \in \mathcal{Y}^2} \min\left\{
    \pi(y_{0}, y_{1},
    D(z, a)=1), \pi(y_{0}, y_{1},
    D(z', a')=1) \right \}\\
  \geq (1-\delta_{z, z, a, a'})P^{*}(D(z, a)=1),
\end{multline*}
where the inequality is by condition~\ref{item:disagreement-bound}. We have thus verified that the disagreement bound in \cref{eqn:disagreement-bounds} is satisfied under $P^*$.

\subsection{Proof of \texorpdfstring{\Cref{cor:optmize_over_pi2}}{Corollary~\ref{cor:optmize_over_pi2}}}
The proof is completely analogous to that for \Cref{cor:optmize_over_pi}, except appealing to \Cref{prop:coupling_result_with_disagreement} in place of \Cref{prop:marginals_one_z}.

\subsection{Auxiliary lemmas}\label{sec:auxiliary-lemmas}
For the results in this section, let
$z_{\max}\in\argmax_{z\in\mathcal{Z}} E[D\mid Z=z]$ denote the most lenient
decision-maker, and let $\alpha_{\max}=E[D\mid Z=z_{\max}]$ denote their
leniency. Let $\alpha_{0}=P(D=1)$ denote the treatment rate under the status
quo.

\begin{lemma}\label{lemma:policy_complier_bounds}
  Suppose $\alpha_0<\alpha$. Let $F_{P^{*}}$ denote the cdf of
  $Y(1)\mid D(Z,1)>D(Z,0)$ under $P^{*}$. Then for any
  $P^* \in \mathcal{P}^*_{I}(P; \mathcal{P}_{\alpha}^*)$ and $t\in\mathcal{Y}$,
  \begin{multline*}
    F^{lb}(t):=
    \max\left\{\frac{1-\alpha_{0}}{\alpha-\alpha_{0}}
      F^{lb}_{Y(1)\mid D(Z,0)=0}(t)
      -\frac{1-\alpha}{\alpha-\alpha_{0}}, 0\right\}\\
    \leq F_{P^{*}}(t) \leq
    F^{ub}(t):=\min\left\{\frac{1-\alpha_{0}}{\alpha-\alpha_{0}}
      F^{ub}_{Y(1)\mid D(Z,0)=0}(t), 1\right\},
  \end{multline*}
  where
\begin{align}
  F^{lb}_{Y(1)\mid D(Z,0)=0}(t) &:= \max \left\{\frac{1}{1-\alpha_{0}} \left(F^{lb}_{Y(1)}(t) - \alpha_{0} P(Y \leq t \mid D=1) \right),0 \right\}, \\
  F^{ub}_{Y(1)\mid D(Z,0)=0}(t) &:= \min \left\{\frac{1}{1-\alpha_{0}}
                                  \left(F^{ub}_{Y(1)}(t) - \alpha_{0} P(Y \leq t
                                  \mid D=1) \right),1 \right\}\label{eq:y1_d0_ub}\\
  F^{lb}_{Y(1)}(t) &:= P(Y \leq t \mid D=1, Z=z_{\max})\alpha_{\max},
                     \label{eq:y1_lb}\\
  F^{ub}_{Y(1)}(t) &:= P(Y \leq t \mid D=1, Z=z_{\max}) \alpha_{\max}
                     + 1-\alpha_{\max}.\label{eq:y1_ub}
\end{align}
\end{lemma}
\begin{proof}
  Since $P^{*}(D(Z,1)=1\mid D(Z,0)=0)=(\alpha-\alpha_{0})/(1-\alpha_{0})$, by the law of total probability, we have
  \begin{multline}\label{eq:complier_bound}
    P^*(Y(1) \leq t \mid D(Z,0) =0)=
    P^*(Y(1) \leq t \mid D(Z,1)> D(Z,0))
    \frac{\alpha-\alpha_{0}}{1-\alpha_{0}} \\
    + P^*(Y(1) \leq t \mid D(Z,1)=D(Z,0) =0)\frac{1-\alpha}{1-\alpha_{0}}.
  \end{multline}
  Rearranging the expression and using the fact that
  $P^*(Y(1) \leq t \mid D(Z,1)=0, D(Z,0) =0)$ and $F_{P^{*}}(t)$ both lie in the
  unit interval yields
  \begin{multline*}
    \max\left\{\frac{1-\alpha_{0}}{\alpha-\alpha_{0}} P^*(Y(1) \leq t \mid D(Z,0) =0)
    -\frac{1-\alpha}{\alpha-\alpha_{0}}, 0\right\}\leq
    F_{P^{*}}(t) \leq\\
\min\left\{\frac{1-\alpha_{0}}{\alpha-\alpha_{0}} P^*(Y(1) \leq t \mid D(Z,0)
  =0), 1\right\}.
  \end{multline*}
  The claim follows if we can show that the cdf of $Y(1)\mid D(Z,0)=0$ can be
  bounded by $F^{lb}_{Y(1)\mid D(Z,0)=0}(t)$ and
  $F^{ub}_{Y(1)\mid D(Z,0)=0}(t)$. To this end, note that by the law of total
  probability,
  \begin{equation*}
    P^*(Y(1) \leq t) = P^*(Y(1) \leq t \mid D(Z,0)=1) \alpha_{0} \: + \: P^*(Y(1)
    \leq t \mid D(Z,0)=0) (1-\alpha_{0}).
  \end{equation*}
  Rearranging, and using the fact that
  $P^*(Y(1) \leq t \mid D(Z,0)=1)=P(Y \leq t \mid D=1)$ gives
  \begin{equation}\label{eq:y1d0_intermediate}
    P^*(Y(1) \leq t\mid D(Z,0)=0) = (P^*(Y(1) \leq t) -\alpha_{0} P(Y \leq t \mid D=1)) / (1-\alpha_{0}).
  \end{equation}
  We now bound $P^*(Y(1) \leq t)$. Note that since
  $P^*(Y(1) \leq t \mid D(z_{\max},0)=1)=P(Y \leq t\mid D=1, Z=z_{\max})$, applying the law of total probability again, we have
  \begin{multline*}
    P^*(Y(1) \leq t) \\= P(Y \leq t \mid D=1,Z=z_{\max})\alpha_{\max}
    + P^*(Y(1) \leq t \mid D=0,Z=z_{\max}) (1-\alpha_{\max}).
  \end{multline*}
  Using the trivial bounds $P^*(Y(1) \leq t \mid D=0,Z=z_{\max})\in [0,1]$
  yields $F^{lb}_{Y(1)}(t) \leq P^*(Y(1) \leq t) \leq F^{ub}_{Y(1)}(t)$.
  Plugging these bounds into \cref{eq:y1d0_intermediate} then yields the bounds
  $F^{lb}_{Y(1)\mid D(Z,0)=0}(t)\leq P^{*}(Y(1)\leq t\mid D(Z,0)=0)\leq
  F^{ub}_{Y(1)\mid D(Z,0)=0}(t)$, as claimed.
\end{proof}

\begin{lemma}\label{lemma:upper_lower_bounds_EY_D_1_tight}
  Suppose $\alpha_{0}<\alpha$. If
  $\mathcal{P}^*_{I}(P; \mathcal{P}^*_{\alpha, Mon}) \neq \emptyset$, then there
  exist distributions
  $P^{lb}, P^{ub} \in\mathcal{P}^*_{I}(P; \mathcal{P}^*_{\alpha, Mon})$ such
  that $P^{lb}(Y(1)\leq t\mid D(Z,1)>D(Z,0))=F^{lb}(t)$ and
  $P^{ub}(Y(1)\leq t\mid D(Z,1)>D(Z,0))=F^{ub}(t)$ for all $t \in \mathcal{Y}$, with $F^{ub}(t)$ and
  $F^{lb}(t)$ defined as in \Cref{lemma:policy_complier_bounds}.
\end{lemma}
\begin{proof}
 By definition, $P^* \in \mathcal{P}^*_{I}(P; \mathcal{P}^{*}_{\alpha, Mon})$ if and only if $P^*$ satisfies the following
  properties: (a) $P^{*}(D(Z,1)=1)=\alpha$, (b) $P^{*}$ generates $P$, (c)
  \Cref{assumption:policy_monotonicity} holds (d)
  \Cref{assumption:instrument_monotonicity} holds; (e) \ac{IV} validity
  (\cref{eq:valid_iv}) holds, and (f) $P^{*}(Y(0)=0)=1$.

    Fix an arbitrary
  $P^{*}\in\mathcal{P}^*_{I}(P; \mathcal{P}^{*}_{\alpha, Mon})$. Property (f) implies that it is sufficient to think of $P^{*}$ as a
  distribution of the random vector $(Y(1), D(\cdot, 0), D(\cdot, 1), Z)$.
  Furthermore, under property (b), property (e) is equivalent to $Z$ being
  independent of $(Y(1), D(\cdot, 0), D(\cdot, 1))$, with marginal distribution
  $P^{*}(Z=z)=P(Z=z)$. To construct $P^{ub}$, we tweak the distribution of
  $(Y(1), D(\cdot, 0), D(\cdot, 1))$ under $P^{*}$, so that the resulting
  distribution satisfies properties (a)--(d), as well as
  $P^{ub}(Y(1) \leq t\mid D(Z,1)>D(Z,0))= F^{ub}(t)$.

  Since $(D(\cdot, 0), D(\cdot, 1))$ are discrete, to specify the distribution
  of $(Y(1), D(\cdot, 0), D(\cdot, 1))$, it suffices to specify the
  probabilities $P^{ub}(Y(1)\leq t, D(\cdot, 0)= G_{0}, D(\cdot, 1)= G_{1})$ for
  any $t\in\mathcal{Y}$, with $G_{d}=(G_{d}^{1}, \dotsc, G_{d}^{K})$ and
  $G_{d}^{k}\in \{0,1\}$ for $d=0,1$. We set
  \begin{multline*}
    P^{ub}(Y(1)\leq t, D(\cdot, 0)= G_{0}, D(\cdot, 1)= G_{1})=\\
    \int \1{y\leq t}P^{ub}(D(\cdot, 1)= G_{1}\mid Y(1)=y,D(\cdot, 0)= G_{0})\, d P^{ub}_{Y(1)\mid G_{0}}(y)\cdot
    P^{*}(D(\cdot, 0)=G_{0}),
  \end{multline*}
  where
  \begin{equation*}
    P^{ub}(D(\cdot, 1)= G_{1}\mid Y(1)=y, D(\cdot, 0)= G_{0})
    =\prod_{k=1}^{K}\phi(G_{1}^{k}\mid G_{0}^{k}, y),
  \end{equation*}
  with $\phi(0\mid G_{0},y)=1-\phi(1\mid G_{0},y)$, $\phi(1\mid 1,y)=1$, and
  $\phi(1\mid 0,y)$ specified below, and, writing $G_0=0$ as a shorthand for the
  vector of $K$ zeros, $(0, \dotsc, 0)$,
  \begin{equation*}
    P^{ub}_{Y(1)\mid G_{0}}(t):=
    P^{ub}(Y(1)\leq t\mid D(\cdot, 0)= G_{0})=
    \begin{cases}
      1 & \text{if $G_{0}=0$,}\\
      P^{*}(Y(1)\leq t\mid D(\cdot, 0)= G_{0}) & \text{otherwise,}
    \end{cases}
  \end{equation*}
  denotes the conditional distribution $Y(1)\mid D(\cdot, 0)= G_{0}$. This
  construction implies that $P^{ub}$ matches $P^{*}$ except that the conditional
  distribution of $Y(1)\mid D(\cdot,0)=0$ is degenerate and equal to the
  smallest value in the support $\mathcal{Y}$ (as we shall see below, this
  ensures that the marginal distribution of $Y(1)$ matches $F^{ub}_{Y(1)}$ in
  \cref{eq:y1_ub}), and the conditional distribution of
  $D(z,1)\mid D(z,0), Y(1)$ is the same for all $z$, and conditional on $Y(1)$,
  the vectors $\{(D(z,0), D(z,1))\}_z$ are independent across $z$.

  We first verify that
  $P^{ub}\in \mathcal{P}^*_{I}(P; \mathcal{P}^*_{\alpha, Mon})$ so long as
  $\phi$ is chosen so that
  \begin{equation}
    E_{P^{ub}}[D(Z,1)]=\alpha\label{eq:ensure_alpha}.
  \end{equation}
  Second, we show that the distribution $P^{ub}_{Y(1)\mid 0}$ equals
  $F^{ub}_{Y(1)\mid D(Z,0)=0}$ in~\cref{eq:y1_d0_ub}. Finally, we choose
  $\phi(1\mid 0,y)$ so that the conditional cdf of $Y(1)\mid D(Z,1)>D(Z,0)$
  matches $F^{ub}$. Note \cref{eq:ensure_alpha} implies that property (a) holds.
  Since $Y(0)=0$, the data distribution depends only on the marginal
  distribution of $D(\cdot, 0)$ and the conditional distribution of
  $Y(1)\mid D(\cdot, 0)\neq 0$. Since these coincide with $P^{*}$, $P^{ub}$
  generates $P$, so that property (b) also holds. Likewise, since IV
  monotonicity depends only on the distribution of $D(\cdot, 0)$, and this
  matches $P^{*}$, property (d) also holds. Finally, since $\phi(0\mid 1,y)=0$,
  property (c) holds also. Hence, under \cref{eq:ensure_alpha}
  $P^{ub}\in\mathcal{P}^*_{I}(P; \mathcal{P}^*_{\alpha, Mon})$ as claimed.

  Next, we verify that $P^{ub}_{Y(1)\mid 0}=F^{ub}_{Y(1)\mid D(Z,0)=0}$ as
  defined in~\cref{eq:y1_d0_ub}. Since
  $P^{ub}\in\mathcal{P}^*_{I}(P; \mathcal{P}^*_{\alpha})$, replacing $P^{*}$
  with $P^{ub}$ in \cref{eq:y1d0_intermediate} yields
  \begin{equation}\label{eq:y1d0_intermediate2}
        P^{ub}(Y(1) \leq t\mid D(Z,0)=0) = (P^{ub}(Y(1) \leq t) -\alpha_{0} P(Y \leq t \mid D=1)) / (1-\alpha_{0}).
  \end{equation}
  Under IV monotonicity, the event that $D(z,0)=0$ for all $z$ is equivalent to
  $D(z_{\max},0)=0$. Hence, by the law of total probability, and the fact that
  $P^{ub}$ generates $P$,
  \begin{multline*}
    P^{ub}(Y(1)\leq t)\\
    = P^{ub}(Y(1)\leq t\mid D(\cdot, 0)=0)(1-\alpha_{\max})+
    P(Y\leq t\mid D=1,Z=z_{\max})\alpha_{\max}=F_{Y(1)}^{ub}(t),
  \end{multline*}
  with $F_{Y(1)}^{ub}$ defined in \cref{eq:y1_ub}. Replacing
  $P^{ub}(Y(1)\leq t)$ with $F_{Y(1)}^{ub}(t)$ in \cref{eq:y1d0_intermediate2}
  yields
  \begin{equation*}
    P^{ub}(Y(1) \leq t\mid D(Z,0)=0) = \frac{1}{1-\alpha_{0}}(F_{Y(1)}^{ub}(t) -\alpha_{0} P(Y \leq t \mid D=1)).
  \end{equation*}
  The right-hand side must be smaller than 1, since $P^{ub}$ is a valid
  probability measure (because
  $P^{ub}\in\mathcal{P}^*_{I}(P; \mathcal{P}^*_{\alpha})$). Hence,
  $P^{ub}_{Y(1)\mid 0}=F^{ub}_{Y(1)\mid D(Z,0)=0}$ as claimed.

  Finally, we specify $\phi(1\mid 0,y)$ so that the conditional cdf of
  $Y(1)\mid D(Z,1)>D(Z,0)$ matches $F^{ub}$. Let $y_{\eta}$ denote the $\eta$
  quantile of the conditional distribution of $Y(1)\mid D(Z,0)=0$ under
  $P^{ub}$, with $\eta:=(\alpha-\alpha_{0})/(1-\alpha_{0})$. We set
  \begin{equation}\label{eq:dz1_definition}
    \phi(1\mid 0, y)=
    \begin{cases}
      1 & \text{if $y < y_{\eta}$,} \\
      \frac{\eta-P^{ub}_{Y(1)\mid D(Z,0)=0}(Y(1)<y_\eta)}{
      P^{ub}_{Y(1)\mid D(Z,0)=0}(Y(1)=y_\eta)} & \text{if $y=y_\eta$,}\\
      0 & \text{if $y>y_{\eta}$,}
    \end{cases}
  \end{equation}
  where we interpret $0/0$ as $0$ if
  $P^{ub}_{Y(1)\mid D(Z,0)=0}(Y(1)=y_\eta)=0$. The intuition for this choice is
  that by Bayes rule, the density of $Y(1)\mid D(z,1)>D(z,0)$ is proportional to
  $P^{ub}(D(z, 1)=1\mid D(z, 0)=0, Y(1)=y) P^{ub}(Y(1)=y \mid D(z,0) =0)$, so that this choice
  shifts the density as far left as possible while satisfying the requirement
  that $E[D(Z,1)]=\alpha$, thus maximizing $F_{P^{*}}$, the cdf of
  $Y(1) \mid D(Z,1) > D(Z,0)$. Under this choice, by iterated expectations,
  \begin{equation*}
    \begin{split}
      P^{ub}(D(Z,1)=1)& = P^{ub}(D(Z,1)=1\mid D(Z,0)=0)(1-\alpha_{0})+\alpha_{0}\\
                      & =E_{P^{ub}}[P^{ub}(D(z, 1)=1\mid D(z, 0)=0, Y(1))\mid D(z, 0)=0] (1-\alpha_{0})+\alpha_{0}\\
                      &=\alpha,
    \end{split}
  \end{equation*}
  so that
  \cref{eq:ensure_alpha} holds. Furthermore, since
  $P^{ub}\in \mathcal{P}^*_{I}(P; \mathcal{P}^*_{\alpha})$, rearranging
  \cref{eq:complier_bound} with $P^{*}$ replaced by $P^{ub}$ yields
  \begin{multline*}
    P^{ub}(Y(1) \leq t \mid D(Z,1)> D(Z,0))\\
    =
    \frac{1-\alpha_{0}}{\alpha-\alpha_{0}}P^{ub}(Y(1) \leq t \mid D(Z,0) =0)-
    \frac{1-\alpha}{\alpha-\alpha_{0}}P^{ub}(Y(1) \leq t \mid D(Z,1)=D(Z,0) =0)\\
    =
    \frac{1-\alpha_{0}}{\alpha-\alpha_{0}}F^{ub}_{Y(1)\mid D(Z,0)=0}(t)\left(
      1- P^{ub}(D(Z,1)=0 \mid Y(1) \leq t, D(Z,0) =0)\right)
  \end{multline*}
  where the second equality uses Bayes rule and the identity
  $P^{ub}(D(Z,1)=0\mid D(Z,0)=0)=(1-\alpha)/(1-\alpha_{0})$. By iterated
  expectations and \cref{eq:dz1_definition},
  $1-P^{ub}(D(Z,1)=0 \mid Y(1) \leq t, D(Z,0) =0)=1$ if $t< y_{\eta}$,
  and equals $\eta/P^{ub}_{Y(1)\mid D(Z,0)=0}(Y(1)\leq t)$ otherwise. From the definition of $\eta=(\alpha-\alpha_{0})/(1-\alpha_{0})$ and the equality $P^{ub}_{Y(1)\mid D(Z,0)=0}(Y(1)\leq t) = F^{ub}_{Y(1)\mid D(Z, 0)=0}(t)$ shown earlier, we have that $\eta/P_{Y(1)\mid D(Z,0)=0}(Y(1)\leq
  t)=\frac{\alpha-\alpha_{0}}{1-\alpha_{0}} / F^{ub}_{Y(1)\mid D(Z, 0)=0}(t)$. Substituting into the previous display, we then see that  \begin{equation*}
    P^{ub}(Y(1) \leq t \mid D(Z,1)> D(Z,0))=\min\left\{\frac{1-\alpha_{0}}{\alpha-\alpha_{0}}F^{ub}_{Y(1)\mid D(Z,0)=0}(t),1\right\},
  \end{equation*}
  which matches \cref{eq:y1_d0_ub} as claimed.

  The construction of $P^{lb}$ is identical, except we set the distribution of
  $Y(1)\mid D(\cdot,0)=0$ to be degenerate and equal to the largest value in the
  support $\mathcal{Y}$, and set
  \begin{equation*}
    \phi(1 \mid 0,y)=
    \begin{cases}
      0 & \text{if $y < y_{1-\eta}$,} \\
      \frac{\eta-P^{lb}_{Y(1)\mid D(Z,0)=0}(Y(1)>y_{1-\eta})}{
      P^{lb}_{Y(1)\mid D(Z,0)=0}(Y(1)=y_{1-\eta})} & \text{if $y=y_{1-\eta}$,}\\
      1 & \text{if $y>y_{1-\eta}$,}
    \end{cases}
  \end{equation*}
  where $y_{1-\eta}$ denotes the $1-\eta$ quantile of the conditional
  distribution of $Y(1)\mid D(Z,0)=0$ (and we again interpret $0/0$ as $0$ if
  $P^{lb}_{Y(1)\mid D(Z,0)=0}(Y(1)=y_{1-\eta})=0$).
\end{proof}

\section{Additional details}\label{sec:addit-deta-empir}

\subsection{Example where IV monotonicity helps}\label{sec:monotonicity_helps}

Consider a setting with a binary instrument and binary outcomes, such that the observed data probabilities are given by:\\
  \begin{tabular}{@{}c|cc@{}}
    $P(Y=y, D=d\mid Z=1)$ & {$Y=1$} & {$Y=0$} \\
    \midrule
    $D=1$ & 1/3 & 1/3 \\
    $D=0$ & 1/3 & 0
  \end{tabular}\hspace{1ex}
  \begin{tabular}{@{}c|cc@{}}
    $P(Y=y, D=d\mid Z=0)$ & {$Y=1$} & {$Y=0$} \\
    \midrule
    $D=1$ & 1/3 & 0 \\
    $D=0$ & 1/3 & 1/3
  \end{tabular}

Assume IV monotonicity holds. Then there are three response types under the status quo  (i.e. types for $D(\cdot,0)$), instrument
compliers ($C$), instrument always takers ($A$), and instrument never-takers
($N$), and it follows from the data probabilities that these all have population
share $1/3$. Furthermore, since the only untreated type when $Z=1$ is $N$, and
$P(Y=0, D=0\mid Z=1)=0$, it follows that $P(Y(0)=1\mid N)=1$. But since $P(Y=0, D=0\mid Z=0)=1/3$, it follows that $P(Y(0)=0\mid C)=1$. Similarly,
since the only treated type when $Z=0$ is $A$, and $P(Y=0, D=1\mid Z=0)=0$, it
follows that $P(Y(1)=1\mid A)=1$, which, combined with the fact that $P(Y=0, D=1\mid Z=1)=1/3$ implies that $P(Y(1)=0\mid C)=1$. Hence, instrument compliers have zero treatment
effect, while that for instrument never takers lies between $-1$ and $0$. Consider a counterfactual policy satisfying policy monotonicity that increases treatment take-up by some small amount $\epsilon$. Since policy compliers are drawn from the pool of instrument never-takers and instrument compliers, it follows that the bounds for policy compliers are given by $E[Y(1)-Y(0)\mid D(Z,1)>D(Z,0)]\in [-1,0]$. In particular, IV monotonicity allows us to conclude that the policy must have weakly negative impacts.

Without IV monotonicity, however, the data is also compatible with the
population comprising instrument defiers ($F$) with population share $1/3$ such
that $P(Y(1)=Y(0)=1\mid F)=1$, with the remainder of the population being
instrument compliers, half of whom have treatment effect equal to $-1$ and half
of whom have a treatment effect equal to $1$. Thus, without IV monotonicity, we
only obtain trivial bounds on the effect for policy compliers, $E[Y(1)-Y(0)\mid
D(Z,1)>D(Z,0)]\in [-1, 1]$.

This example is quite extreme, however, in that we were able to rule out a positive treatment effect for instrument never-takers under monotonicity. So long as we cannot rule out that at least mass $\epsilon$ of instrument never-takers have a treatment effect equal to each of $-1$ and $1$, assuming IV monotonicity will also lead to trivial bounds on the effect for policy compliers for policies that only marginally increase take-up.

\subsection{Implementation details}\label{sec:impl-deta}

\paragraph{Calibration of disagreement probability bounds}

To calibrate the parameter $\overbar{\DP}$ in \cref{eq:dp_bound}, we proceed in two steps. First, we relate it to a correlation parameter
in a Gaussian signal model with correlated signals. Second, we match the correlation parameter in this model to the data from \textcite{sigstad_monotonicity_2023}. We then calculate the implied value of $\overbar{\DP}$.

We consider a model in which judges obtain noisy signals of the criminal misconduct risk of defendants, $V_z$, similar to the model in \textcite[Section V]{cgy22}. We assume that for each pair of judges $z\in\mathcal{Q}$ and $z'\in\mathcal{Q}^c$,
the signals are jointly normal, with means normalized to zero and variances normalized to one,
\begin{equation*}
    \begin{pmatrix}
        V_z\\V_{z'}
    \end{pmatrix}\sim \mathcal{N}\left(0,\begin{pmatrix}
        1 & \rho\\ \rho & 1
    \end{pmatrix}\right),
\end{equation*}
where $\rho$ governs the correlation between the signals. Each judge then releases the defendant if the signal is low enough, $D(z,1)=\1{V_z\leq \alpha_z}$, where the cutoffs $\alpha_z$ are calibrated to match the release rate of judge $z$, $\alpha_z=\Phi^{-1}(q)$ for $z\in\mathcal{Q}$, where $q$ is the release quota and $\Phi$ standard normal cdf,
and $\alpha_{z'}=\Phi^{-1}(E[D(z',0)])$ for $z'\in\mathcal{Q}^c$, since judges not subject to the quota are assumed not to change their behavior under the counterfactual. Under this model, the average disagreement probability equals
\begin{multline*}
  P(D(Z_\mathcal{Q},1)=0, D(Z_{\mathcal{Q}^c},1)=1)
  = \frac{1}{\abs{\mathcal{Q}}\abs{\mathcal{Q}^c}}\sum_{z\in\mathcal{Q},
    z'\in\mathcal{Q}^c}
  P(D(z,1)=0, D(z',1)=1)\\
  = \frac{1}{\abs{\mathcal{Q}}\abs{\mathcal{Q}^c}}\sum_{z\in\mathcal{Q}, z'\in\mathcal{Q}^c}
  P(V_z>\alpha_z, V_{z'}<\alpha_{z'}).
\end{multline*}
Thus, given a value of $\rho$, we can compute the probability on the right-hand side, and calculate the implied value of the disagreement probability as $\overbar{\DP}=\frac{1}{\abs{\mathcal{Q}^c}}\sum_{z'\in\mathcal{Q}^c}
  P(V_z>\Phi^{-1}(q), V_{z'}<\alpha_{z'})/q$, using the fact that the marginal release rates under the counterfactual equal the quota $q$ for all judges in $\mathcal{Q}$. In our application, instead of a simple average, we use a weighted average of the probabilities $P(V_z>\Phi^{-1}(q), V_{z'}<\alpha_{z'})$, weighted by the number of cases handled by $z$ and $z'$.

As a benchmark, we use the average value of $\rho$ across pairs of judges from the top decile and bottom nine deciles of leniency in the data from \textcite{sigstad_monotonicity_2023}. This data has information on panels of judges ruling on criminal cases in the São Paulo Appeal Court. As described in the main text, we focus on three judge panels. To calculate judge leniencies, we regress judge decisions on judge and case fixed effects, and form leniency deciles. To reduce estimation noise, we then restrict the data pairs of judges who see at least a thousand cases together. Then, for each pair of judges $z$ and $z'$ such that $z$ is in the bottom 90\% and $z'$ in the top decile, we calculate the value of correlation $\rho_{z, z'}$ implied by the joint distribution of their decisions. Averaging across all such pairs then yields the value $\rho=0.989$.

\paragraph{Covariate adjustment}
For simplicity, our theoretical analysis has focused on the case where $Z$ is randomly assigned. In practical applications, it may only be plausible to impose that $Z$ is randomly assigned conditional on covariates, $Z \indep (Y(\cdot), D(\cdot,\cdot)) \mid X$. For example, defendants may only be randomly assigned to judges conditional on their assigned court and the time at which they were arrested. In principle, the analysis described above under random assignment could be conducted conditional on each value of $X$, and these conditional bounds could be aggregated to obtain bounds on the unconditional average outcome under the counterfactual. However, implementing this approach in practice would require estimating the propensity scores $P(Z \mid X)$ or the outcome model $P(Y, D \mid Z,X)$, which may be challenging in practice if the covariates $X$ are high-dimensional. For example, there may not be a large number of defendants within each court-by-time cell, and each judge may only be observed in a subset of the court-by-time cells.

In practice, applied researchers typically rely on a linear outcome model to adjust for covariates, and often report estimates of the covariate-adjusted means for each judge. For example, as described in more detail below, ADH22 report estimates of $\mu_{yd}^z  :=P(Y(d)=y,D(z,0)=d)$ for each judge $z$.\footnote{More precisely, they report estimates of $P(Y(d) =1 \mid D(z,0)=1)$ and $E(D(z,0))$ which, since $Y(0)=0$, imply estimates of $P(Y(d)=y,D(z,0)=d)$.} This corresponds to the probability that $Y=y$ and $D=d$ for defendants assigned to judge $z$ that would be realized \emph{if} defendants were unconditionally randomly assigned under the status quo. In our applications, we follow this standard approach in the empirical literature and obtain bounds on $\theta = \sum_z P(Z=z) E[Y(D(z,1))]$ using the machinery developed above, replacing $P(Y=y,D=d \mid Z=z)$ in our analysis with the corresponding covariate-adjusted value of $\mu^z_{yd}$. We note that $\theta$ corresponds to the counterfactual outcome that would be realized if \emph{both} judges were randomly assigned to defendants \emph{and} the counterfactual policy $a=1$ were implemented. Likewise, the implied policy effect $\tau := \theta - \sum_z P(Z=z) E[Y(D(z,0))]$ corresponds to the marginal impact of changing from $a=0$ to $a=1$ if judges were unconditionally randomly assigned. Below, we also outline certain functional form restrictions under which $\tau$ can also be interpreted as the marginal impact of changing from $a=0$ to $a=1$ while preserving the same conditional-on-$X$ assignment rule as under the status quo.

We now provide additional details on how the $\mu_{yd}^z$ are obtained in each of our applications. Let $\mu$ be the vector that stacks the $\mu_{yd}^z$. ADH22 report covariate-adjusted means from linear regressions that adjust court-by-time fixed effects $X$ as well as judge dummies. In particular, they estimate linear regression specifications with additively separable specifications for $E[D\mid Z, X, R=r]=Z'a^{r}+X'c^{r}$ and
$E[Y\mid D=1,Z, X, R=r]=Z'b^{r}+X'd^{r}$ for each race $r$. (Robustness checks in
ADH22 indicate that more flexible specifications do not alter their results.) They then report the predicted value for each judge at the average value of $X$, $\bar{a}_z^r: = z' a^r + E[X \mid R=r]' c^r$ and $\bar{b}_z^r:=z' b^r + E[X \mid R=r]'d^r$. Under the additively separable specification imposed by ADH22, and recalling that $Y(0)=0$ in this application, we then have that $\mu_{10 \mid r}^z = 0, \mu_{00\mid r}^z = 1 - \bar{a}_z^r, \mu_{11 \mid r}^z = \bar{a}_z^r \bar{b}_z^r$, and $\mu_{01\mid r}^z = \bar{a}_z^r (1-\bar{b}_z^r)$, where $\mu_{yd\mid r}^z=P(Y(d)=y, D(z,0)=d\mid R=r)$ are race-specific probabilities. Our estimates $\hat{\mu}$ of $\mu$ replace $\bar{a}_z^r$ and $\bar{b}_z^r$ in these expressions with the least-squares estimates provided to us by ADH22 (and reported in their Figure 2).

We similarly control for covariates linearly in the Suffolk County prosecutor data, although the analysis is somewhat more complicated since both $Y(1)$ and $Y(0)$ are both non-degenerate binary random variables. Specifically, we account for covariates using an
additively separable specification for the model primitives $P^{*}$, the
joint distribution of $(Y(\cdot), D(\cdot, \cdot), Z)$ conditional on $X$, where
$X$ consists of court-by-time fixed effects and defendant characteristics. As we
explain below, this additively separable model also allows for an alternative interpretation
of the policy effect. In particular, we assume
\begin{multline}\label{eq:linear-separable}
  P^{*}(Y(0)=y_{0}, Y(1)=y_{1}, D(z,0)=d_{0}, D(z,1)=d_{1}\mid Z=z, X=x)\\
  =
  \pi_{z}(y_{0}, y_{1}, d_{0}, d_{1})+(x-E[X])'\rho_{y_{0}, y_{1}, d_{0}, d_{1}}.
\end{multline}
This implies that the data probabilities are also additively separable, with
$P(Y=y, D=d\mid Z=z, X=x)=\mu_{yd}^z+(x-E[X])'\beta_{yd}$, with
\begin{equation}\label{eq:data_comp_x}
  \sum_{y_{0}, d}\pi_z(y_{0}, y, 1, d)=\mu_{y 1}^z, \quad
  \sum_{y_{1}, d}\pi_z(y, y_{1}, 0, d)=\mu_{y 0}^z, \quad
  \sum_{y_{0}, d}\rho_{y_{0} y 1 d}=\beta_{y 1}, \quad
  \sum_{y_{1}, d}\rho_{y_{1} y 0 d}=\beta_{y 0}.
\end{equation}
Hence, we obtain estimates $\hat{\mu}$ of $\mu$ as estimates of ADA fixed
effects in a regression of indicators $\1{Y=y, D=d}$ on ADA indicators $Z$ and
the covariate vector $X$.\footnote{We note that this approach assumes an additive linear model for $E[YD \mid Z,X]$, whereas in the previous application, we used estimates reported by ADH22 based on an additive linear model for $E[Y D \mid D=1,Z,X]$. We adopt the former approach since, as described below, it yields an intuitive alternative interpretation to the policy parameter. Since we do not have micro-data for ADH22, we cannot apply this approach in their data.}

This additive structure implies that the policy parameter $\tau$, which we argued above corresponds to the marginal impact of changing from $a=0$ to $a=1$ under random assignment, has a second useful interpretation: it is also the marginal impact of a policy $a=1$ that imposes covariate-adjusted restrictions on judges' release rates while preserving conditional-on-$X$ assignment as in the status quo. Specifically, for a given set of target counterfactual release rates $\alpha_{z,1}$, consider an alternative policy $a=1$ that preserves the conditional-on-$X$ assignment under the status quo but requires that, for $\rho^* = \sum_{d_0,y_0,y_1} \rho_{y_0,y_1,d_0,1}$, \begin{equation}
P(D(z,1)=1 \mid Z=z) - (E[X\mid Z=z]-E[X])'\rho^* = \alpha_{z,1}. \label{eqn:covariate-adjusted-quota}
\end{equation}
This imposes that each judges' \emph{covariate-adjusted} release rate $P(D(z,1)=1 \mid Z=z) - (E[X\mid Z=z]-E[X])' \rho^*$ matches $\alpha_{z,1}$. This could, for example, correspond to a policy that imposes a covariate-adjusted quota on judges' release rates. Under the additively separable model, the average outcome under this counterfactual corresponds exactly to that under random assignment, since the covariate adjustment cancels out after averaging:
\begin{align*}
E[Y(D(Z,1))] &= E[ E[Y(D(Z,1)) \mid Z, X ] ] \\
& = E[ E[Y(D(Z,1) \mid Z, X= E[X]] + (E[X \mid Z] - E[X])' \bar\rho ]\\
& = \sum_{z} P(Z=z) E[Y(D(z,1))] + ( E[ E[X \mid Z] ] - E[X])' \bar\rho \\
& = \sum_{z} P(Z=z) E[Y(D(z,1))],
\end{align*}
where $\bar\rho := \sum_{y_0,y_1,d_0,d_1} (y_1 d_1 + y_0 (1-d_1)) \rho_{y_0,y_1,d_0,d_1}$. An analogous argument shows that $E[Y(D(Z,0))] = \sum P(Z=z) E[Y(D(z,0)]$. Likewise, the constraint in \eqref{eqn:covariate-adjusted-quota} can be written as
\begin{align*}
& E[ E[D(z,1) \mid X,Z=z] \mid Z=z ] - (E[X \mid Z=z] - E[X])'\rho^* = \alpha_{1,z} \\
\iff &[ E[D(z,1)] +  (E[X \mid Z=z] - E[X])'\rho^*] - (E[X \mid Z=z] - E[X])'\rho^* = \alpha_{1,z}  \\
\iff & E[D(z,1)] = \alpha_{1,z}
\end{align*}
and thus we see that imposing the covariate-adjusted quota \eqref{eqn:covariate-adjusted-quota} under conditional-on-$X$ assignment is equivalent to imposing a quota of $\alpha_{1,z}$ under random assignment. As a result, the bounds we obtain on $\theta$ assuming random assignment with quotas $\alpha_{1,z}$ can also be interpreted as bounds on $E[Y(D(Z,1))]$ when imposing the covariate-adjusted quota in \eqref{eqn:covariate-adjusted-quota} and preserving conditional-on-$X$ assignment; analogously, our bounds on $\tau$ correspond to bounds on the marginal impact of this covariate-adjusted quota.

Finally, we note that under the additively separable model in
\cref{eq:linear-separable}, the bounds on $\theta$ we obtain by working with the
estimates $\hat{\mu}$ alone are valid, but not necessarily sharp. Let $\pi$
denote the vector that stacks $\pi_{z}(y_{0}, y_{1}, d_{0}, d_{1})$, and let
$\rho$ stack $\rho_{y_{0}, y_{1}, d_{0}, d_{1}}$. Under the additively separable
model, an application of \Cref{prop:marginals_one_z} conditional on covariates
implies that there exists a joint distribution $P^{*}$ compatible with
$(\pi, \rho)$ if and only if (i) \cref{eq:data_comp_x} holds, (ii)
$(\pi, \rho)$ imply the same distribution for $(Y(0), Y(1))$,
$\sum_{d_{0}, d_{1}}\pi_{z}(y_{0}, y_{1}, d_{0}, d_{1})
=\sum_{d_{0}, d_{1}}\pi_{1}(y_{0}, y_{1}, d_{0}, d_{1})$ for all $z$, and (iii) $(\pi, \rho)$
imply valid probability mass functions,
\begin{equation}\label{eq:pms_valid_pi}
  1= \sum_{y_{0}, y_{1}, d_{0}, d_{1}}\pi_{z}(y_{0}, y_{1}, d_{0}, d_{1}), \quad
  \pi_{z}(y_{0}, y_{1}, d_{0}, d_{1})\geq 0
\end{equation}
and
\begin{equation}\label{eq:pms_valid_rho}
  \sum_{y_{0}, y_{1}, d_{0}, d_{1}}\rho_{y_{0}, y_{1}, d_{0}, d_{1}}=0,\quad
  \pi_{z}(y_{0}, y_{1}, d_{0}, d_{1})+(x-E[X])'\rho_{y_{0}, y_{1}, d_{0}, d_{1}}\geq 0
\end{equation}
for all $x$. If we only work with estimates of $\mu$ and ignore estimates of
$\beta$, we only verify the first two conditions in
\cref{eq:data_comp_x}, (ii), and~\cref{eq:pms_valid_pi}. We do not check
whether there exists a vector $\rho$ satisfying \cref{eq:pms_valid_rho} and the
last two constraints in \cref{eq:data_comp_x}, so that the identified
set for $\theta$ is not necessarily sharp. If one wanted to take the additively separable model literally, one could potentially obtain tighter bounds by exploiting the additional restrictions described above; we view the additively
separable model merely as a reasonable approximation that allows us to estimate
$\mu$ in a simple manner, following the standard practice in the literature, and therefore do not exploit the additional restrictions on
$\rho$.

\paragraph{Details on inference}

For inference on parameters that can be written as a solution to a linear
program, we use a \emph{projection} approach. We first describe the general
construction, and then specialize to each parameter that we consider in the
empirical illustrations. Following the description in
\Cref{sec:no_monotonicity}, suppose the lower and upper bounds on $\theta$ are given by
\begin{equation}\label{eq:lp_implementation}
\min_{\pi} / \max_{\pi} \omega'\pi \text{ s.t. } B\pi\geq b \text{ and } A\pi\geq\mu,
\end{equation}
where the inequalities should be interpreted elementwise. Here $\pi$ stacks the marginal probability mass functions $\{\pi_{z}(\cdot)\}_{z}$, $\omega$ stacks the weights
$P(Z=z)(d_{1}y_{1}+(1-d_{1})y_{0})$, and $\mu$ is a vector of length $8K$ that
stacks the $4K$ probabilities $P(Y=y, D=d\mid Z=z)$, and the additive inverses
$-P(Y=y, D=d\mid Z=z)$, and $b$ is a known vector not depending on the data.  The constraints $B\pi\geq b$ collect restrictions that
do not involve observable data (e.g.\ policy monotonicity), while the
constraints $A\pi\geq \mu$ collect the data-compatibility constraints
(condition~\ref{item:match_data} in \Cref{prop:marginals_one_z}) and any other data-dependent linear restrictions. At a high-level, suppose that we have access to a 95\% uniformly valid confidence band for $\mu$: that is, a set $\hat{\mathcal{C}}$ such that the probability that $\mu \in \hat{\mathcal{C}}$ is at least 95\% asymptotically (uniformly over a suitably regulated class of DGPs). It is then immediate---by the usual argument that justifies the projection approach---that the interval obtained by solving the programs
\begin{equation}\label{eqn:lp_projection_CS}
\min_{\pi,\tilde\mu} / \max_{\pi,\tilde\mu} \omega'\pi \text{ s.t. } B\pi\geq b, A\pi\geq \tilde\mu \text{ and } \tilde\mu \in \hat{\mathcal{C}},
\end{equation}
will be a 95\% uniformly valid confidence interval for the identified set for
$\theta$.\footnote{We note that $95\%$ coverage of the identified set implies
  coverage of the parameter $\theta \in \Theta_I$ of at least $95\%$, although
  coverage of $\theta$ may be conservative.} To preserve the linear programming
structure, we consider one-sided rectangular uniform confidence bands
$\hat{\mathcal{C}}$ formed using a sup-$t$ approach \parencite[as in,
e.g.,][]{olea_simultaneous_2019}. To construct $\hat{\mathcal{C}}$, we assume
the estimates $\hat{\mu}$ of $\mu$ are unbiased with covariance matrix
$\Sigma/n$, and that we have a consistent estimator $\widehat{\Sigma}$ for
$\Sigma$. For the NYC bail judge application, we estimate $\Sigma$ under the
assumption that the data is independent across judges (thus ignoring the
covariate adjustment described above), while for the Suffolk County prosecutor
data, we use the influence function of $\hat{\mu}$ to calculate the covariance
matrix (thus accounting for covariate adjustment). To construct a basic
one-sided sup-$t$ confidence band for $\mu$, we can choose
$\hat{\mathcal{C}} = \{\mu \colon \mu \geq \hat{\mu}-cv
\operatorname{diag}(\widehat{\Sigma}/n)^{1/2}\}$, where the critical value $cv$
is set to be the $95\%$ quantile of the maximum element of
$Z \sim \mathrm{N}(0,\widehat{\Sigma}/n)$, calculated via simulation. Then
\cref{eqn:lp_projection_CS} is equivalent to a linear program that replaces
$\mu$ in \cref{eq:lp_implementation} with the lower confidence bound
$\hat{\mu}-cv \operatorname{diag}(\widehat{\Sigma}/n)^{1/2}$,
\begin{equation*}
  \min_{\pi,\tilde\mu} / \max_{\pi,\tilde\mu} \omega'\pi \text{ s.t. } B\pi\geq b, A\pi\geq \hat{\mu}-cv \operatorname{diag}(\widehat{\Sigma}/n)^{1/2}.
\end{equation*}
One can verify that this approach is equivalent to forming a confidence set by inverting the tests based on ``least favorable'' critical values discussed in \textcite{andrews_inference_2023}. An advantage of this approach is that asymptotic validity of this sup-$t$
confidence band does not require that $\hat{\mu}$ be jointly asymptotically
Gaussian. In particular, so long as the estimates $\hat{\mu}$ are computed using
regression, they can be written as a sample mean of an $8K$-dimensional random
vector. \Textcite{chernozhukov2013a} show that in this case, the maximum of
$\hat{\mu}-\mu$ behaves in large samples like a maximum of Gaussian random
variables, even when $K\to\infty$ as the sample size $n\to\infty$. This result
can be used to justify the projection approach in settings with a large number
$K$ of decision-makers.

In practice, we use the details of each counterfactual to refine the basic
confidence band $\hat{\mathcal{C}}$. For the universal release policy, we showed
above that the identified set depends on $\mu$ only through the probabilities
that enter the intervals $\mathcal{I}_z$, and therefore we construct a
sup-$t$ confidence band only for those probabilities rather than for the full
vector $\mu$.\footnote{Equivalently, we form a joint confidence set for $\mu$
  where the bounds on the irrelevant elements are $(-\infty,\infty)$.} Specifically,
we form a lower confidence band for the probabilities
$\{P^*(Y(1)=1, D(z,0)=1\mid R=r), P^*(Y(1)=0, D(z,0)=1\mid
R=r)\}_{z\in\mathcal{Z}}$, separately for each race, since these are the
probabilities that affect the identified set. The linear program in
\cref{eqn:lp_projection_CS} is then equivalent to intersecting confidence
intervals for the judge-specific intervals $\mathcal{I}_z$ in
\cref{eq:judge_specific_I}, as shown in \Cref{fig:id-set-illustration}.

For inference on the quota policy, we intersect the basic sup-$t$ confidence band with a confidence band for certain aggregate moments. Adding aggregate moments does not change plug-in point estimates of the bounds, but it can help tighten inference if it places sharper restrictions on moments that are closely related to the bounds of the identified set, which may only be imprecisely bounded using the basic rectangular confidence set. Let $\mu_A$ be the vector containing $E[Y\mid Z\in \mathcal{Q}]-E[Y\mid Z\in \mathcal{Q}^c]$ (up to scaling, this is
the estimate of the policy effect under the reallocation policy), as well as the
aggregate release rates $E[D\mid Z\in\mathcal{Q}]$ and
$E[D\mid Z\in\mathcal{Q}^c]$. We form a two-sided sup-$t$ confidence band for $\mu_A$, say $\hat{\mathcal{C}}^{\mu_A}$, and let $\hat{\mathcal{C}}^A$ denote the set of values for $\mu$ consistent with $\mu_A \in \hat{\mathcal{C}}^{\mu_A}$. We then form our confidence band $\hat{\mathcal{C}} = \hat{\mathcal{C}}^{basic} \cap \hat{\mathcal{C}}^{A}$, where $\hat{\mathcal{C}}^{basic}$ is the basic sup-$t$ confidence band for $\mu$ described above. We use a 99\% level band for $\hat{\mathcal{C}}^{basic}$ and a 96\% level for $\hat{\mathcal{C}}^A$, so that the overall confidence band has level at least 95\%.

To compute the bounds and confidence intervals for policy compliers, we use the
fact that the parameter of interest is a ratio,
$\theta/\frac{n_{\mathcal{Q}}}{n}(E[D(Z_\mathcal{Q}, 0)]-q)$, where
$n_{\mathcal{Q}}$ is the number of cases handled by the bottom 90\% of judges.
The denominator is an affine function of $\pi$, so this is a linear-fractional
program. We transform it to a linear program using the Charnes–Cooper
transformation, and form confidence intervals by projection as described above.

Finally, for inference on the disparate impact parameter $\Delta$, which is a smooth non-linear function of counterfactuals and aggregate moments, we adopt a
two-step approach. We first form a joint confidence band for the race-specific universal release parameters
$(E[Y(1) \mid R=w], E[Y(1) \mid R=b])$ with nominal level 95.5\%, as described
above, except that we form a joint upper 95.5\% confidence band for the
probabilities
$\{P^*(Y(1)=1, D(z,0)=1\mid R=r), P^*(Y(1)=0, D(z,0)=1\mid
R=r)\}_{z\in\mathcal{Z}}$ \emph{jointly} for both races. Intersecting the
judge-specific probabilities for each race then yields a joint 95.5\% confidence
interval for the race-specific universal release parameters. In the second step,
for each value of $(E[Y(1) \mid R=w], E[Y(1) \mid R=b])$ in the confidence set,
we form a delta-method confidence interval with level $99.5\%$ for $\Delta$ that
accounts for the statistical uncertainty in the aggregate moments (more
precisely, we do this for each value within a $100 \times 100$ grid of points
within the confidence band). Our final confidence interval is formed by taking
unions of all the second-step intervals.

\subsection{Additional tables and figures}

\begin{table}[tp]
  \begin{threeparttable}
    \caption{Comparison of estimates and standard errors in ADH22 and our
      Replication.}\label{tab:adh-replication-vs-orig}
    \begin{tabular*}{\linewidth}{@{\extracolsep{\fill}}lcc cc cc@{}}
      &\multicolumn{2}{@{}c}{$E[Y(1)\mid R=b]$} & \multicolumn{2}{@{}c}{$E[Y(1)\mid R=w]$}
      & \multicolumn{2}{@{}c}{$\Delta$}\\
      \cmidrule(rl){2-3}\cmidrule(rl){4-5}\cmidrule(rl){6-7}
      &Orig.  & Repl. &Orig.  & Repl. &Orig.  & Repl. \\
      Specification & (1) & (2) & (3) & (4) & (5) & (6)\\
      \midrule
      Linear extrap. & 40.0 & 40.0 & 33.8 & 33.5 & 5.4 & 5.4\\
 & $(0.6)$ & $(0.5)$ & $(0.7)$ & $(0.6)$ & $(0.2)$ & $(0.2)$\\
Quadratic extrap. & 39.4 & 42.8 & 31.9 & 36.2 & 5.4 & 4.8\\
 & $(2.1)$ & $(1.8)$ & $(2.1)$ & $(1.8)$ & $(0.7)$ & $(0.7)$\\
Local linear extr. & 43.6 & 43.6 & 34.6 & 35.0 & 4.2 & 4.3\\
 & $(1.6)$ & $(2.3)$ & $(1.4)$ & $(2.0)$ & $(0.6)$ & $(0.7)$\\
\\[-0.9em]
      \bottomrule
    \end{tabular*}
    \begin{tablenotes}
      \begin{footnotesize}
      \item \emph{Notes:} This table compares the estimates and standard errors
        (in parentheses) from Table 3 of ADH22 to our replication, as reported
        in \Cref{tab:estimates-delta-arnold}.
      \end{footnotesize}
    \end{tablenotes}
  \end{threeparttable}
\end{table}

\begin{table}[tp]
  \centering
  \begin{threeparttable}
    \caption{Estimates for universal release policy the disparate impact
      parameter using NYC bail judge data with and without
      pooling.}\label{tab:estimates-delta-disag}
  \begin{tabular*}{\linewidth}{@{\extracolsep{\fill}}lcc cc c@{}}
    &\multicolumn{2}{@{}c}{Whites} & \multicolumn{2}{@{}c}{Blacks}\\
    \cmidrule(rl){2-3}\cmidrule(rl){4-5}
    &$E[Y(1)\mid R]$ & {PC TE}
     &$E[Y(1)\mid R]$ & {PC TE}& $\Delta$\\
    Specification & (1) & (2) & (3) & (4) & (5)\\
    \midrule
    Pooled & $(29.2, 49.1)$ & $(20.9, 86.3)$ & $(23.4, 41.9)$ & $(12.0, 90.7)$ & $(1.0, 9.6)$\\
Disaggregated & $(29.1, 49.3)$ & $(20.5, 86.9)$ & $(22.6, 42.8)$ & $(8.7, 94.3)$ & $(0.9, 9.9)$\\
\\[-0.9em]
    \bottomrule
  \end{tabular*}
  \begin{tablenotes}
    \begin{footnotesize}
    \item \emph{Notes}: See \Cref{tab:estimates-delta-arnold}. 95\% confidence
      intervals in parentheses; these do not account for covariate adjustment.
      ``Pooled'' refers to the main specification that pools judges with 300 or
      fewer cases; it is identical to the first row of
      \Cref{tab:estimates-delta-arnold}. ``Disaggregated'' refers to estimates
      without pooling.
    \end{footnotesize}
  \end{tablenotes}
\end{threeparttable}
\end{table}

\begin{table}[tp]
  \centering
\begin{threeparttable}
\caption{Estimates for quota policy using NYC bail judge data with
    and without pooling.}\label{tab:quota_nyc_disag}
  \begin{tabular*}{\linewidth}{@{\extracolsep{\fill}}lcccc cccc@{}}
    &\multicolumn{2}{@{}c}{no PC bound} & \multicolumn{2}{@{}c}{PC bounds}\\
    \cmidrule(rl){2-3}\cmidrule(rl){4-5}
    &Policy effect & {PC TE}
    &Policy effect & {PC TE}\\
    Specification & (1) & (2) & (3) & (4)\\
    \midrule
    \multicolumn{3}{@{}l}{A\@: Blacks (Pooled)}\\
    Valid IV only & $(0.0, 10.7)$ & $(0.0, 100.0)$ & $(1.8, 10.7)$ & $(16.4, 100.0)$\\
$\overline{DP}=0.025$ & $(2.8, 7.9)$ & $(26.3, 76.7)$ & $(2.8, 7.9)$ & $(26.3, 76.7)$\\
\\
    \multicolumn{3}{@{}l}{B\@: Blacks (Disaggregated)}\\
    Valid IV only & $(0.0, 10.7)$ & $(0.0, 100.0)$ & $(1.6, 10.7)$ & $(15.2, 100.0)$\\
$\overline{DP}=0.025$ & $(2.8, 7.9)$ & $(26.3, 76.7)$ & $(2.8, 7.9)$ & $(26.3, 76.7)$\\
\\
    \multicolumn{3}{@{}l}{C\@: Whites (Pooled)}\\
    Valid IV only & $(0.0, 6.9)$ & $(0.0, 100.0)$ & $(0.7, 6.9)$ & $(9.9, 100.0)$\\
$\overline{DP}=0.021$ & $(1.2, 5.9)$ & $(17.0, 91.1)$ & $(1.2, 5.9)$ & $(17.0, 91.1)$\\
\\
    \multicolumn{3}{@{}l}{D\@: Whites (Disaggregated)}\\
    Valid IV only & $(0.0, 6.9)$ & $(0.0, 100.0)$ & $(0.6, 6.9)$ & $(8.8, 100.0)$\\
$\overline{DP}=0.021$ & $(1.2, 5.9)$ & $(17.0, 91.1)$ & $(1.2, 5.9)$ & $(17.0, 91.1)$\\
\\[-0.9em]
    \bottomrule
  \end{tabular*}
  \begin{tablenotes}
    \begin{footnotesize}
    \item \emph{Notes}: See \Cref{tab:quota_nyc}. 95\% confidence intervals in
      parentheses; these do not account for covariate adjustment. ``Pooled''
      refers to the main specification that pools judges with 300 or fewer
      cases; Panels A and C are identical to Panels A and B in
      \Cref{tab:quota_nyc}. ``Disaggregated'' refers to estimates without
      pooling.
  \end{footnotesize}
\end{tablenotes}
\end{threeparttable}
\end{table}

\Cref{tab:adh-replication-vs-orig} compares our estimates and standard errors
for the parametric extrapolation methods, as reported in
\Cref{tab:estimates-delta-arnold}, to the original estimates in Table 3 in
ADH22. The CIs in ADH22 account for the fact that the judge-specific means are
covariate-adjusted, whereas since we do not have the microdata, we conduct
inference assuming the judge-specific means are sample means from i.i.d.\ draws.
There are also some differences in weighting: ADH22 weight by the inverse
precision of the judge fixed effects (accounting for covariate estimation),
whereas we weight by the number of released defendants. In spite of these
differences, the estimates and standard errors are very similar for all
specifications.

\Cref{tab:estimates-delta-disag} and \Cref{tab:quota_nyc_disag} show that
inference for the NYC bail judge data is virually identical if we do not pool
judges with a few cases; in some instances, the confidence intervals are wider,
since the projection-based confidence intervals use larger critical values.
\Cref{tab:quota_suffolk_disag} shows that aggregation has no impact on inference
in the Suffolk County prosecutor data. In all cases, the point estimates without
pooling are empty, since the estimated probabilities $P(Y(d)=y, D(z, 0)=d)$ are
not compatible with \cref{eq:valid_iv}. Pooling reduces the
estimation noise in these probabilities for decision-makers with a few cases; it
also helps ensure that the estimated probabilities are approximately Gaussian,
as assumed by our projection inference method.

\begin{table}[tp]
  \centering
\begin{threeparttable}
  \caption{Estimates for quota policy using Suffolk County prosecutor data with
    and without pooling.}\label{tab:quota_suffolk_disag}
  \begin{tabular*}{0.8\linewidth}{@{\extracolsep{\fill}}lcc@{}}
    &Policy effect & {PC TE}\\
    Specification & (1) & (2)\\
    \midrule
    \multicolumn{3}{@{}l}{A\@: Pooled}\\
    Valid IV only & $(-9.4, 9.4)$ & $(-100.0, 100.0)$ \\
$\overline{OD}=0.019$ & $(-5.4, 2.0)$ & $(-62.6, 23.2)$ \\
\\
    \multicolumn{3}{@{}l}{B\@: Disaggregated}\\
    Valid IV only & $(-9.4, 9.4)$ & $(-100.0, 100.0)$ \\
$\overline{OD}=0.019$ & $(-5.4, 2.0)$ & $(-62.6, 23.2)$ \\
\\[-0.9em]
    \bottomrule
  \end{tabular*}
  \begin{tablenotes}
    \begin{footnotesize}
    \item \emph{Notes}: See \Cref{tab:quota_suffolk}. 95\% confidence intervals
      reported in parentheses. ``Pooled'' refers to the main specification that
      pools ADAs with 300 or fewer cases; Panels A is identical to
      Panels A and B in \Cref{tab:quota_suffolk}. ``Disaggregated'' refers to
      estimates without pooling.
  \end{footnotesize}
\end{tablenotes}
\end{threeparttable}
\end{table}

\end{appendices}

\end{document}